\documentclass[11pt]{article}
\usepackage[utf8]{inputenc}
\usepackage[T1]{fontenc}
\usepackage[utf8]{inputenc}
\usepackage{authblk}

\title{\Large{Steady-State and Dynamical Behavior of a PDE Model of Multilevel Selection with Pairwise Group-Level Competition}}
\author[1]{Konstantinos Alexiou}
\affil[1]{Department of Mathematics, University of St. Andrews, St. Andrews, Scotland, UK}
\author[2,*]{Daniel B. Cooney}
\affil[2]{Department of Mathematics, University of Illinois Urbana-Champaign, Urbana, IL, USA}
\affil[*]{Correspondence to dbcoone2@illinois.edu}

\usepackage{latexsym}
\usepackage{amssymb,amsmath}
\usepackage{amsthm}
\usepackage{mathtools}
\usepackage{graphicx}
\usepackage{color}
\usepackage{blkarray}
\usepackage{multirow}
\usepackage{hyperref}
\usepackage{float}
\usepackage{subfig}
\usepackage{relsize}
\usepackage{array,makecell}
\usepackage{comment}
\usepackage{scalerel,stackengine}
\stackMath
\newcommand\reallywidehat[1]{%
\savestack{\tmpbox}{\stretchto{%
  \scaleto{%
    \scalerel*[\widthof{\ensuremath{#1}}]{\kern.1pt\mathchar"0362\kern.1pt}%
    {\rule{0ex}{\textheight}}
  }{\textheight}%
}{2.4ex}}%
\stackon[-6.9pt]{#1}{\tmpbox}%
}
\usepackage{cleveref}

\usepackage[numbers,sort&compress]{natbib}

\hypersetup{
  colorlinks   = true, 
  urlcolor     = blue, 
  linkcolor    = blue, 
  citecolor   = red 
}

\usepackage{breqn}

\usepackage{xcolor}
\usepackage{cancel}

\usepackage{hyperref}
\usepackage{empheq}

\usepackage[utf8]{inputenc}
\usepackage{nicefrac}

\usepackage{setspace}

\usepackage{enumerate}

\usepackage[titletoc,toc]{appendix}

\usepackage{tocbasic}
\DeclareTOCStyleEntry[
  beforeskip=.2em plus 1pt,
  pagenumberformat=\textbf
]{tocline}{section}

\usepackage[font=small,labelfont=bf]{caption}

\newcommand{\RR}{\mathbb{R}}

\newcommand{\ds}{\displaystyle}
\newcommand{\mc}{\mathcal}
\newcommand{\ol}{\overline}

\newcommand{\bbm}{\begin{bmatrix}}
\newcommand{\bpm}{\begin{pmatrix}}
\newcommand{\ebm}{\end{bmatrix}}
\newcommand{\epm}{\end{pmatrix}}

 \newcommand{\dsdel}[2]{\displaystyle\frac{\partial #1}{\partial #2}}

\newcommand{\dsddx}[2]{\displaystyle\frac{d #1}{d #2}}
\newcommand{\dsddt}[1]{\displaystyle\frac{d #1}{dt}}

\newcommand{\holder}{Hölder  \hspace{0.05mm}}

\renewcommand{\abstractname}{Abstract}
\numberwithin{equation}{section}
\numberwithin{figure}{section}
\renewcommand{\thesection}{\arabic{section}}

\begin{document}

\maketitle

\newtheorem{definition}{Definition}[section]
\newtheorem{theorem}{Theorem}[section]
\newtheorem{lemma}[theorem]{Lemma}
\newtheorem{corollary}[theorem]{Corollary}
\newtheorem{claim}[theorem]{Claim}
\newtheorem{fact}[theorem]{Fact}
\newtheorem{proposition}[theorem]{Proposition}
\newtheorem{remark}[theorem]{Remark}
\newtheorem{example}[theorem]{Example}
\newtheorem{observation}[theorem]{Observation}

\renewcommand{\thesection}{\arabic{section}}
\setcounter{section}{0}

\begin{abstract}
Evolutionary competition often occurs simultaneously at multiple levels of organization, in which traits or behaviors that are costly for an individual can provide collective benefits to groups to which the individual belongs. Building off of recent work that has used ideas from game theory to study evolutionary competition within and among groups, we study a PDE model for multilevel selection that considers group-level evolutionary dynamics through a pairwise conflict depending on the strategic composition of the competing groups. This model allows for incorporation of group-level frequency dependence, facilitating the exploration for how the form of probabilities for victory in a group-level conflict can impact the long-time support for cooperation via multilevel selection. We characterize well-posedness properties for measure-valued solutions of our PDE model and apply these properties to show that the population will converge to a delta-function at the all-defector equilibrium when between-group selection is sufficiently weak. We further provide necessary conditions for the existence of bounded steady state densities for the multilevel dynamics of Prisoners' Dilemma and Hawk-Dove scenarios, using a mix of analytical and numerical techniques to characterize the relative strength of between-group selection required to ensure the long-time survival of cooperation via multilevel selection. We also see that the average payoff at steady state appears to be limited by the average payoff of the all-cooperator group, even for games in which groups achieve maximal average payoff at intermediate levels of cooperation, generalizing behavior that has previously been observed in PDE models of multilevel selection with frequency-indepdent group-level competition.

\end{abstract}
 
\singlespacing

%
{\hypersetup{linkbordercolor=black, linkcolor = black}
\begin{spacing}{0.01}
\renewcommand{\baselinestretch}{0.1}\normalsize
\tableofcontents
\addtocontents{toc}{\protect\setcounter{tocdepth}{2}}
\end{spacing}
\singlespacing

\section{Introduction}

In various natural and social systems, evolutionary dynamics driven by natural selection or cultural transmission can operate across multiple levels of organization, creating tensions between the evolutionary incentives of individuals and the collectives to which the individuals belong. These tensions between levels of selection arise on scales ranging from genetic conflict within cells \cite{sachs2004evolution,paulsson2002multileveled,gitschlag2020nutrient} to the evolution of cooperative behavior in complex animal societies \cite{shaffer2016foundress,waring2017coevolution,boyd2003evolution}. Natural selection operating on aggregates of individuals can help to facilitate major evolutionary transitions like the emergence of protocells \cite{hogeweg2003multilevel,takeuchi2009multilevel,szathmary1987group} the origin of chromosomes \cite{smith1993origin,szathmary1993evolution,szathmary1995major} and the evolution of multicellularity \cite{staps2019emergence,pichugin2019evolution,ispolatov2012division,tarnita2013evolutionary}, while cultural group selection has been attributed as a mechanism for the promotion of cooperative social norms that facilitate the formation of large human societies \cite{boyd2003evolution,henrich2004cultural}. To understand such varied natural phenomena that arise from cross-scale evolutionary competition, it can be helpful to use mathematical modeling to formulate and analyze the tug-of-war between traits or behavior favored at different levels of biological organization. 

Evolutionary game theory provides a mathematical framework that can be helpful to analyze the conflict between an individual incentive to cheat and a collective incentive to achieve cooperation within a group of individuals. Modeling game-theoretic interactions in group-structured populations provides examples of misalignment between individual-level and group-level interests, as social dilemmas may arise in which individual payoff is maximized by a cheating strategy but the average payoff of group members is maximized when at least some members of a group cooperate. By considering two-player, two-strategy games with a range of payoff matrices, it is possible to formulate a variety of social dilemmas in which cooperation is socially beneficial, but in which individual-level replicator dynamics can favor dominance of defection, coexistence between cooperators and defectors, or bistability of all-defector and all-cooperator states \cite{hofbauer1998evolutionary,nowak2005evolution}. Depending on the payoff structure of underlying games, it is also possible to explore scenarios in which the average payoff of group members is maximized by a group composed only of cooperators, as well as games for which an intermediate level of cooperation can maximize the collective payoff for a group \cite{cooney2019replicator}.

A range of mathematical frameworks have been introduced to describe the dynamics of multilevel selection, incorporating different ways to describe group-structured populations and how the evolution of cooperation can be achieved through group-level competition. These frameworks include trait-group models in which collective replication of transiently-formed groups can help to promote altruistic behaviors \cite{wilson1975theory,wilson1977structured,fontanari2024dynamics}, models with fixed group structure featuring group-level fission or fusion events \cite{traulsen2005stochastic,traulsen2006evolution,simon2010dynamical,simon2012numerical,simon2013towards,simon2016group,simon2024evolutionary,simon2024fission}, and spatially explicit models in which group formation and group-level competition emerges via spatial pattern formation \cite{savill1997self,hogeweg2003multilevel,hogeweg2012toward,hermsen2022emergent,doekes2024multiscale} or other aggregation processes \cite{tarnita2013evolutionary,ackermann2023role}. The mathematical approaches used to describe cross-scale evolutionary dynamics range from individual-based stochastic models used to describe the fixation or persistence of cooperation due to group-level competition \cite{traulsen2005stochastic,traulsen2006evolution,traulsen2008analytical,mcloone2018stochasticity,bottcher2016promotion} to a variety of PDE models describing multilevel selection that incorporate individual-level evolutionary forces like migration, mutation, and genetic drift \cite{ogura1987stationary,ogura1987stationary2,velleret2024two} or which incorporate detailed group-level events including fission, fusion, and collective extinction of groups \cite{simon2010dynamical,simon2012numerical,simon2013towards,simon2016group,simon2024fission,simon2024evolutionary,lerch2024flexible}.
 Similar nested stochastic models have been further explored to study a range of biological phenomena including host-pathogen coevolution \cite{pokalyuk2019diversity,pokalyuk2020maintenance}, the evolution of cooperative or complementary genetic replicators in protocells \cite{fontanari2006coexistence,fontanari2013solvable,fontanari2014effect,fontanari2014nonlinear,markvoort2014computer}, and the origin of chromosomes \cite{smith1993origin}.

Luo and couathors recently introduced a stochastic framework for describing evolutionary dynamics featuring individual and collective birth-death competition in group-structured populations \cite{luo2014unifying,luo2017scaling,van2014simple}, modeling finite population dynamics through a nested Moran process and deriving a PDE describing the dynamics of multilevel selection in the limit of large population size. Luo and Mattingly considered the case of two types of individuals in which one type had a fixed advantage under individual-level replication and the other type conferred a collective advantage  to their group, showing that beneficial group-level outcomes could be achieved in the long-time behavior a PDE model of multilevel selection when competition among groups was sufficiently strong \cite{luo2017scaling}. Subsequent extensions of these two-level birth-models have explored fixation probabilities in finite populations \cite{mcloone2018stochasticity}, the existence of quasi-stationary distributions in a diffusive PDE scaling limit of the two-level stochastic process \cite{velleret2020individual,velleret2024two}, and the formulation of individual-level and group-level replication rates based on two-player, two-strategy social dilemma games played within each group \cite{cooney2019replicator,cooney2020analysis}. The resulting hyperbolic PDE models for multilevel selection have been further generalized to study multilevel dynamics with individual-level and group-level replication rates described by arbitrary functions of the fraction of cooperators within each group \cite{cooney2022long}, and results for these generalized models have been applied to explore synergistic effects of group-level competition and within-group mechanisms for promoting the evolution of cooperation \cite{cooney2022assortment,cooney2023evolutionary} and to study models of protocell evolution and the origin of chromosomes \cite{cooney2022pde}. 

For these two-level replicator equation models, it was possible to use the method of characteristics to determine the long-time behavior for these models of multilevel selection. A particularly interesting feature of the two-level replicator model was a phenomenon described as a ``shadow of lower-level selection'', in which the long-time group-level replication rate could not exceed the replication rate of the all-cooperator group, even for scenarios in which the replication rate of groups was maximized by intermediate levels of cooperation \cite{cooney2019replicator,cooney2020analysis,cooney2022long}. One question of interest for this paper is whether this behavior is limited to the case of previously studied two-level replicator equations, or whether this long shadow cast by lower-level selection can hold for a broader class of PDE models of multilevel selection that incorporate frequency-dependent competition at the group level. While existing work on generalizations of the Luo-Mattingly PDE model of multilevel selection typically assume that group-level replication events occur at rates that depend only on the strategic composition of the replicating group, it is also possible to incorporate frequency-dependent competition at the level of groups by assuming that group-level selection occurs through interactions between competing groups.

Questions related to multilevel selection with group-level interactions have often been considered in the evolutionary anthropology literature, arising in models of cultural evolution of behaviors that spread through both individual-level transmission and pairwise conflict between groups. Simulation studies of cultural group selection have been used to study the coevolution of cooperative behaviors and within-group social norms for punishment of defectors or rewarding of cooperators \cite{boyd2003evolution,janssen2014effect,odouard2023polarize}. The simulation model introduced by Boyd and coauthors described competition between groups through a series of pairwise conflicts between groups, in which the probability of group-level victory depends on the strategic composition of the competing groups, with the victorious group producing a copy of itself and replacing the group that lost the pairwise conflict. This form of pairwise between-group competition differs from the models based on the two-level birth-death process proposed by Luo and coauthors, in which the replication rate of groups depended only on the strategic composition of the replicating group. By allowing pairwise conflicts to determine group-level replication events, we are now incorporating frequency dependence in our model of group-level competition, allowing for a more general description of the evolutionary dynamics of group-structured competition featuring selection within and between groups. Such models of pairwise group-level conflict have also recently been applied to study nested models of multilevel selection with density-dependent within-group dynamics \cite{lerch2024flexible}, and prior theoretical work on intergroup conflict in animal populations \cite{rusch2020logic} suggests substantial room for formulating multilevel selection models with intergroup competition for resources in the presence of individual-level competition within groups. 

Our goal in this paper is to understand how pairwise group-level competition impacts the dynamics of multilevel selection for various scenarios arising from evolutionary games. We look to explore how the incorporation of group-level frequency dependence impacts both the qualitative behavior and mathematical details of PDE models of multilevel selection relative to existing work on two-level replicator equations that feature group-level replication rates depending only on the strategic composition of the replicating group. This analysis allows us to explore how the tug-of-war between individual level and group incentives plays out in the case of pairwise group conflicts, broadening the scope of analytically tractable models of multilevel selection and highlighting which previously studied behaviors of two-level evolutionary dynamics may be robust to the specific assumptions made by mathematical modelers when formulating a stochastic or PDE model of multilevel selection.

In this paper, we study the dynamical and steady-state behavior of a PDE model of multilevel selection that incorporates pairwise between-group competition, expanding on the recent class of PDE models that have assumed frequency-independent group-level replication rates. We provide a measure-valued formulation for the PDE model, using the method of characteristics and a contraction mapping argument to show well-posedness of measure-valued solutions and obtain an implicit representation formula to solutions of the multilevel dynamics. As a first application of this measure-valued representation of solutions, we study the infimum and supremum H{\"o}lder exponent of solutions near the all-cooperator equilibrium, which are two properties of the measure-valued solution that can be helpful to characterize the tail-behavior of the strategic distribution of the group-structured population \cite{cooney2022long}. We are able to show that the infimum and supremum H{\"o}lder exponents near the all-cooperator composition is preserved in time for our model, suggesting that these quantities may play a similar role for multilevel dynamics with pairwise group conflicts as they have been shown to do in the case of two-level replicator equations \cite{cooney2022long}. 

We then explore the possible long-time outcomes for the multilevel dynamics for the cases of within-group and group-level competition based on generalizations of the Prisoners' Dilemma (PD) and Hawk-Dove games (HD). For the PD game, we use the representation formula for measure-valued solutions to show that the population will converge to a delta-function at the all-defector equilibrium when between-group selection is sufficiently weak relative to competition within groups. We also explore the possibility of existence of density steady states for the multilevel dynamics for the PD and HD scenarios, deriving necessary conditions for the existence of steady states with given behavior near the all-cooperator equilibrium. These necessary conditions provide us with an expression for the average group-level victory probability such a steady-state population would achieve in pairwise competition with the all-cooperator group, and motivate a conjectured formula for the threshold strength of between-group competition required to achieve a steady state supporting positive levels of cooperation under the dynamics of multilevel selection with pairwise group-level competition. We further explore numerical simulations for our PDE model, observing good agreement between the behavior of numerical solutions and the conjectured expressions for the threshold selection strength and collective success of possible steady-state populations. Notably, we see that the numerical solutions and conjectured analytical formulas suggest that the population may be limited by the collective success of the all-cooperator group even in the limit of infinitely strong between-group competition, suggesting that the shadow of lower-level selection seen in two-level replicator models may also generalize to our PDE models for multilevel selection featuring frequency-dependent group-level competition. 

We also consider a class of individual and group-level replication rates that generalize the multilevel dynamics for Stag-Hunt (SH) games, with individual-level selection featuring bistability of the all-defector and all-cooperator groups and in which group-level selection most favors the all-cooperator composition. We show that, in the presence of any pairwise group-level competition, the population will converge upon a delta-function at the all-cooperator outcome, so cooperation will achieve long-time fixation in the population when an all-cooperator group is locally stable under individual-level dynamics and is favored under group-level competition. This result provides an analogue to the result of Boyd and Richerson on group selection between alternative stable within-group equilibria  \cite{boyd1990group}, which was initially proposed in the context of finite population dynamics with a separation of time-scales between within-group and group-level competition. This result suggests that achieving local stability of a cooperative equilibrium will be sufficient to achieve full-cooperation under our model of multilevel selection with pairwise group conflict, which can be useful for future work exploring synergistic effects between pairwise group-level competition and within-group mechanisms that help to promote and stabilize cooperation with groups.

In Section \ref{sec:mainmodel}, we introduce our PDE model for multilevel selection with pairwise group-level competition and we discuss the game-theoretical background used to generate assumptions on the individual and group replication rates. In Section \ref{sec:measureproperties}, we present a measure-valued formulation of our PDE model for multilevel selection, providing a characterization of well-posedness and preservation of the tail behavior for measure-valued solutions. We then study dynamical and steady state properties of our model for the case of generalized PD games in Section \ref{sec:PDbehavior}, and provide similar analysis for the generalizations of the HD and SH games in Section \ref{sec:HDandSHbehavior}. We then study numerical solutions to the multilevel dynamics in Section \ref{sec:numerics}, studying dynamical behavior for an example group-level victory probability based on the Fermi update rule and providing evidence for results and conjectures provided in the previous two sections. We present a discussion of our results and an outlook for future work in Section \ref{sec:discussion}, and we provide additional proofs of analytical results and information on our numerical simulations in the appendix.

\section{Baseline Model of Multilevel Selection with Pairwise Group-Level Competition}
\label{sec:mainmodel}

In this section, we will formulate our baseline model for multilevel selection with pairwise group-level competition and explore the underlying game-theoretic assumptions that will motivate our assumptions for the individual-level and group-level replication events in our model. We discuss our general PDE model in Section \ref{sec:generalPDEmodel} and describe the game theory underlying our modeling assumptions in Section \ref{sec:gametheory}.

\subsection{General Approach for Modeling Multilevel Selection with Pairwise Group-Level Competition}
\label{sec:generalPDEmodel}

In this section, we formulate our PDE model for multilevel selection with pairwise group-level competition. A heuristic derivation of this PDE model was provided in the context of describing multilevel selection with pairwise group-level competition in the presence of altruistic punishment within groups \cite{cooney2024exploring}, so we will focus on describing the main assumptions we make regarding individual-level and group-level replication rates. For our model, the dynamics of within-group competition use the same assumptions that were made in the derivation of two-level replicator equation models with replication rates based on the payoff of evolutionary games \cite{cooney2019replicator,cooney2022long}, while the assumptions will differ from previous models by considering group-level replication events that result from pairwise group conflict. 

We consider a population structured into $m$ groups, with each group containing $n$ individuals. There are two possible types of individuals, cooperators ($C$) and defectors ($D$), and we describe the strategic composition of an $n$-member group by the number of cooperators $i$ in the group. We then formulate our model of multilevel selection by specifying the individual-level and group-level replicator rates based on the payoffs achieved in a group featuring $i$ cooperators and $n-i$ defectors. Our individual-level replication events will depend on the payoffs $\pi_C\left(\frac{i}{n}\right)$ and $\pi_D\left(\frac{i}{n}\right)$ received by cooperators and defectors in an $i$-cooperator group. We assume that cooperators and defectors respectively replicate at rates $1 + w_I \pi_C\left(\frac{i}{n}\right)$ and $1 + w_I \pi_D\left( \frac{i}{n} \right)$, where $w_I$ describes the sensitivity of individual-level replication events to the payoffs received by the replicating individuals. We then assume that the offspring of the replicating individual replaces a randomly chosen individual within the same group, which allows the group size to remain constant and to allow us to describe the state of the group through the fraction of cooperator $\frac{i}{n}$ (as the total number of group members remains $n$ for all time). 

Unlike previously studied models in which group replication depends only on the strategic composition of the replicating group, we can introduce frequency-dependent group-level competition by assuming that the group-level replication rate of an $i$-cooperator group will depend on the fraction of groups $f^{m,n}_j(t)$ of groups featuring $j$ cooperators for $j \in \{0,1,\cdots,n\}$. We assume that each group engages in pairwise conflicts with other groups at rate $\Lambda$, and their opponent group in the pairwise conflict is drawn uniformly from the population of groups. This means that a given $i$-cooperator group interacts with a $j$-cooperator group at rate $\Lambda f_j^{m,n}(t)$, and the total rate of interaction between $i$-cooperator and $j$-cooperator groups occurs at rate $\Lambda f^{m,n}_i(t) f^{m,n}_j(t)$. We assume that an $i$-cooperator group defeats a $j$-cooperator group in a pairwise conflict with probability $\rho\left( \frac{i}{n},\frac{j}{n} \right)$, while the $i$-cooperator group loses this conflict with the complementary probability $\rho\left( \frac{j}{n} , \frac{i}{n} \right) = 1 - \rho\left( \frac{i}{n},\frac{j}{n} \right)$. We assume that the winning group in the pairwise conflict makes a copy of itself, with the offspring group replacing the group that loses the conflict. Therefore, the rate at which $i$-cooperator groups are produced in pairwise competition is given by $\sum_{j=0}^n \Lambda f^{m,n}_i(t) f^{m,n}_j(t) \rho\left( \frac{i}{n},\frac{j}{n} \right)$, while the rate at which $i$-cooperator groups are lost due to pairwise competition is given by $\sum_{j=0}^n \Lambda f^{m,n}_i(t) f^{m,n}_j(t) \rho\left( \frac{j}{n},\frac{i}{n} \right)$.

We can then use these rates for individual-level and group-level replication events to describe a large-population limit for this two-level birth-death process. Under the assumption that the group size $n$ and the number of groups $m$ tend to infinity while all other parameters in our model are fixed, we can derive a hyperbolic partial differential equation for $f(t,x)$, the probability density of groups featuring $x$-fraction of cooperators. This PDE takes the following form
\begin{equation} \label{eq:multilevelPDEfirstform}
\begin{aligned}
    \dsdel{f(t,x)}{t} &= \dsdel{}{x} \left[x (1-x) \pi(x) f(t,x) \right] + \lambda f(t,x) \left[  \int_0^1 \rho(x,u) f(t,u) du -  \int_0^1 \rho(u,x) f(t,u) du \right],
    \end{aligned}
\end{equation}
where $\rho(x,y)$ describes the probability that a group with $x$ cooperators defeats a group with $y$ cooperators in a pairwise competition between groups, $\pi(x) := \pi_D(x) - \pi_C(x)$ describes the individual-level replication advantage of defectors, and the parameter $\lambda := \frac{\Lambda}{w_I}$ describes the relative strength of group-level competition. With the interpretation that $\rho(x,y)$ is the probability of an $x$-cooperator group winning a conflict against a $y$-cooperator group, we will always assume that the function $\rho(x,y)$ has the following properties:
\begin{subequations}
\begin{align}
0 &\leq \rho(x,y) \leq 1 \\
\rho(y,x) &= 1 - \rho(x,y)
\end{align}
\end{subequations}
for all $(x,y) \in [0,1]^2$.For details on a heuristic derivation of this PDE model with pairwise group-level competition, we refer the reader to \cite[Section D.1]{cooney2024exploring}.

Equation \eqref{eq:multilevelPDEfirstform} is a hyperbolic PDE, whose characteristic curves are given by the ODE
\begin{equation} \label{eq:characteristics}
\dsddt{x(t)} = - x (1-x) \pi(x) = x(1-x) \left[ \pi_C(x) - \pi_D(x) \right],
\end{equation}
which is the replicator equation for individual-level selection. Within groups, the fraction of cooperators will increase in our PDE model at group compositions $x$ at which cooperators receive a higher payoff than defectors, while the density of $x$-cooperator groups will increase due to group-level competition when $x$-cooperator groups are expected to win more group-level conflicts than they are expected to lose when paired against groups drawn uniformly from the current distribution of groups. 

Using the fact that $\rho(x,u) = 1 - \rho(u,x)$, we can rearrange the nonlocal term of Equation \eqref{eq:multilevelPDEfirstform} to see that our hyperbolic PDE can be 
\begin{equation} \label{eq:multilevelPDEtworho}
    \dsdel{f(t,x)}{t} = \dsdel{}{x} \left[x (1-x) \pi(x) f(t,x) \right] + \lambda f(t,x) \left[ 2 \int_0^1 \rho(x,u) f(t,u) du - 1  \right].
\end{equation}
This form of the equation is typically more convenient for studying well-posedness of solutions and studying the dynamical behavior of the PDE. 

Unlike previously studied PDE models for multilevel selection \cite{luo2017scaling,cooney2019replicator,cooney2022long}, the group-level replication rate $\int_0^1 \rho(x,y) f(t,y) dy$ for an $x$-cooperator group depends on the entire distribution of groups $f(t,\cdot)$. However, with a choice of additively separable pairwise group-level victory probability taking the form
\begin{equation}
\rho(x,y) = \frac{1}{2} \left( 1 + \mc{G}(x) - \mc{G}(y) \right)
\end{equation}
for a function $\mc{G}(x)$ satisfying $|\mc{G}(x)| \leq 1$, we can note that
\begin{equation}
2 \rho(x,y) - 1 = \mc{G}(x) - \mc{G}(y),
\end{equation}
and the PDE from Equation \eqref{eq:multilevelPDEtworho} takes the form
\begin{equation} \label{eq:twolevelreplicator}
    \dsdel{f(t,x)}{t} = \dsdel{}{x} \left[x (1-x) \left( \pi_D(x) - \pi_C(x) \right) f(t,x) \right]  + \lambda f(t,x) \left[ \mc{G}(x) - \int_0^1 \mc{G}(y) f(t,y) dy \right].
\end{equation}
This PDE is of the general form previously studied for multilevel selection with arbitrary continuously-differentiable group-level replication rate $\lambda \mc{G}(x)$, and the long-time behavior for measure-valued solutions of this  PDE model has been fully characterized in prior work \cite{cooney2022long}. Our emphasis in this paper is on group-level victory probabilities $\rho(x,y)$ that are not additively separable, for which the group-level replication term is nonlinear and we cannot study the dynamics of our nonlinear PDE model by studying the behavior of a simpler associated linear PDE.

We can also describe a weak formulation of this model where the population composition is described by the probability measure $\mu_t(dx)$ for groups featuring $x$ cooperators and $1-x$ defectors. For a given test function $v(x) \in C^1([0,1])$, the measure-valued version of our hyperbolic PDE model takes the following form:
\begin{equation} \label{eq:PDEmeasure}
\begin{aligned}
\dsddt{} \int_0^1 v(x) \mu_t(dx) &= -\int_0^1 v'(x) x(1-x) \pi(x) \mu_t(dx) \\
&+ \lambda \int_0^1 v(x) \left[ 2 \left(\int_0^1 \rho(x,y) \mu_t(dy)\right)  - 1 \right] \mu_t(dx).
\end{aligned}
\end{equation}
We will use this characterization of measure-valued solutions to obtain a representation formula for solutions, and our dynamical results on the preservation of tail behavior of our solution near $x=1$ and on the dominance of defection for multilevel Prisoners' Dilemma scenarios with weak group-level competition. An equivalent measure-valued formula was previously used to fully characterize the long-time behavior of PDE models for multilevel selection based on two-level replicator equations, and incorporating population states that can contain delta-measures allows us describe monomorphic steady states in which every group in the population has concentrated upon an equilibrium point for individual-level dynamics.

\subsection{Game-Theoretic Motivation} \label{sec:gametheory}

We now focus on example individual-level and group-level replication rates that depend on the payoffs of two-player, two-strategy games played with the groups. These examples will motivate some of the assumptions we will make on the individual-level replication function $\pi(x)$ and the group-level victory probability $\rho(x,y)$ for our analytical exploration in Sections \ref{sec:PDbehavior} and \ref{sec:HDandSHbehavior}, and the replication rates we we formulate in this section will also provide concrete examples that we use for numerical simulations of the long-time behavior for the PDE models in Section \ref{sec:numerics}.

We will draw the assumptions for our model of multilevel selection from symmetric games with the following payoff matrix
\begin{equation} \label{eq:payoffmatrix}
\begin{blockarray}{ccc}
& C & D \\
\begin{block}{c(cc)}
C & R & S \\
D & T & P \\
\end{block}
\end{blockarray},
\end{equation}
for which the entries correspond to a reward for mutual cooperator ($R$), a punishment for mutual defection ($P$), a temptation payoff to defect against a cooperator ($T$), and a sucker payoff for cooperating with a defector ($S$). Four games of the form that we will discuss in this paper are the Prisoners' Dilemma (PD), the Hawk-Dove (HD) game, the Stag Hunt (SH) game, and the Prisoners' Delight (PDel), which are characterized by the following ranking of the entries of the payoff matrix
\begin{subequations} \label{eq:payoffrankings}
\begin{align}
    \mathrm{PD} &: T > R > P > S \label{eq:PDpayoffs} \\
    \mathrm{HD} &: T > R > S > P \label{eq:HDpayoffs} \\
    \mathrm{SH} &: R > T > P > S. \label{eq:SHpayoffs} \\
    \mathrm{PDel} &: R > T > S > P \label{eq:PDelpayoffs}
\end{align}
\end{subequations}
For each of these games, $R > P$, so players receive a higher payoff if they both cooperate than if they both defect, indicating a collective incentive for a pair of individuals to achieve mutual cooperation. In the PD game, the rankings $T > R$ and $P > S$ imply that an individual receives a higher payoff by defecting than cooperating regardless of the action taken by their opponent. Under the HD game, $T > R$ and $S > P$, so individuals receive a payoff advantage by playing the opposite strategy from the strategy of their opponent.  For the SH game, $R > T$ and $P > S$, so individuals receive a payoff advantage when they play the same strategy as their opponent. Finally, for the PDel game, the rankings $R > T$ and $P > S$ correspond to cooperation yielding a higher payoff than defection regardless of the strategy of one's opponent.

If we assume that each individual plays the game against all members of their own group and the group is infinitely large, then the average payoff obtained by a cooperator and defector in a group composed of a fraction $x$ cooperators and a fraction $1-x$ defectors can be calculated as 
\begin{subequations}
\begin{align}
\pi_C(x) &= x R + (1-x) S \\
\pi_D(x) &= x T + (1-x) P
\end{align}
\end{subequations}
We may also use these payoff functions to define the payoff advantage that defectors have over cooperators in an $x$-cooperator groups as
\begin{equation}
\pi(x) := \pi_D(x) - \pi_C(x) = P - S - \left(R - S - T + P \right) x,
\end{equation}
and we may define the average payoff of the members of an $x$-cooperator group as
\begin{equation}
G(x) := x \pi_C(x) + (1-x) \pi_D(x) = P + \left( S + T - 2P \right) x + \left( R - S - T + P \right) x^2.
\end{equation}

For convenience, we introduce the shorthand notation $\alpha = R - S - T + P$, $\beta = S - P$, and $\gamma = S + T  -2 P$ to describe the roles different values play in the multilevel dynamics. With this notation, we may rewrite the net individual-level advantage of defectors in an $x$-cooperator group as
\begin{equation}
\pi(x) = - \left( \beta + \alpha x \right),
\end{equation}
and we may write the average payoff of an $x$-cooperator group as
\begin{equation}
G(x) = P + \gamma x + \alpha x^2.
\end{equation}

We can then use properties of the individual-level payoff functions $\pi_C(x)$ and $\pi_D(x)$ and the average payoff function $G(x)$ to characterize features of individual-level and group-level selection imposed by each of the four games under consideration. For the dynamics of individual-level selection, we can study the replicator equation
\begin{equation}
\dsddt{x} = - x (1-x) \pi(x) = x (1-x) \left( \beta + \alpha x \right),
\end{equation}
which has equilibrium points $x = 0$ and $x = 1$ for all games. For the HD and SH games, there is a third equilibrium $x = x_{eq}$ given by the interior fraction of cooperation
\begin{equation}
x_{eq} = \frac{\beta}{-\alpha} = \frac{S - P}{S - P + T - R} \in (0,1).
\end{equation}

The within-group dynamics for each of the four games span the possible stability scenarios for the endpoint and interior equilibria \cite{nowak2006evolutionary}. Starting from any interior initial condition, The PD game features convergence to the all-defector equilibrium $x = 0$ (with $\pi(x) > 0$ for $x \in [0,1]$) and the PDel game features convergence to the all-cooperator equilibrium ($\pi(x) < 0$ for $x \in [0,1]$). The HD game features convergence of the population upon the interior equilibrium $x_{eq}$ for interior initial conditions (with $\pi(x) < 0$ for $x \in [0,x_{eq})$ and $\pi(x) > 0$ for $x \in (x_{eq},1]$), while the SH game features bistability between the all-defector and all-cooperator equilibria with the basins of attraction separated by the interior equilibrium $x_{eq}$ (with $\pi(x) > 0$ for $x \in [0,x_{eq})$ and $\pi(x) < 0$ for $x \in (x_{eq},1]$). 

The dynamics of average payoff $G(x)$ and the corresponding group-level dynamics for these games have a slightly more subtle dependence on the parameters of the payoff matrix, with different possible qualitative behaviors displayed within the same classes of games. 
In Lemma \ref{lem:payoffproperties}, we discuss known properties and $G(x)$ arising from the payoffs of two-player and two-strategy games, and we will look to generalize these properties in our subsequent analysis of multilevel dynamics for a broader class of group-level replication rates and approaches to modeling pairwise group-level conflict.

\begin{lemma}[Characterization of Parameters and Average Payoffs for Example Games (Originally studied in \cite{cooney2020analysis})]
\label{lem:payoffproperties}
For the PD game, $\gamma$ and $\alpha$ can take either sign. There are several different possible behaviors for the average payoff $G(x)$ based on the values of $\gamma$ and $\alpha$:
\begin{itemize}
    \item If $\gamma > 0$, then the fraction of cooperation $x^*$ maximizing average payoff $G(x)$ has the following piecewise characterization
    \begin{equation}
  x^* = \left\{
    \begin{array}{lr}
      1 & : \gamma + 2 \alpha \leq 0\\
      \ds\frac{\gamma}{-2\alpha} & : \gamma + 2 \alpha < 0
    \end{array}
  \right. .
    \end{equation}
We therefore see that average payoff is maximized by $x^* = 1$ when $2R > T + S$ (when interaction between two cooperators contributes more to the group's total payoff than the payoff generated by a cooperator and defector). Average payoff is maximized by an intermediate level of cooperation 
\item If $\gamma < 0 < \alpha$, then the average payoff $G(x)$
initially decreases for values of $x$ close to $0$ before reaching a local minimum at $x = \frac{\gamma}{-2 \alpha}$, and then increases until reaching its maximal value at $x^* = 1$. \end{itemize}
For the HD game, the payoff ranking always results in $\gamma > 0 > \alpha$, so the only possibilities are that $G(x)$ is increasing for all $x \in [0,1]$ if $\gamma + 2 \alpha > 0$, or that $G(x)$ has a unique maximum at the point $x^* = \frac{\gamma}{-2\alpha} \in (0,1)$ if $\gamma + 2 \alpha < 0$. In addition, we have the following ranking of payoffs satisfied at the equilibrium points of the within-group dynamics for the HD game:
\begin{equation}
G(1) > G(x_{eq}) > G(0)
\end{equation}
For the SH game, we have that $\alpha > 0$ and $\beta < 0$, but the sign of $\gamma$ cannot be conclusively determined from the SH payoff rankings. This means that $G(x)$ can either be an increasing or decreasing function when the fraction of cooperators $x$ is close enough to $0$, but $G(x)$ is always maximized by the all-cooperator composition $x^* = 1$ and $G(x)$ is always increasing for $x$ close enough to $1$. Furthermore, we have the following ranking of payoffs for the equilibrium points of the within-group dynamics:
\begin{equation}
G(0) < G(x_{eq}) < G(1).
\end{equation}
\end{lemma}

\subsubsection{Group-Level Victory Probabilities}
We now define a family of group-level victory probabilities, which allow us to study the different possible ways that pairwise group competition can depend on the strategic composition of groups or differences in the average payoff of competing groups. Each of these group-level victory probabilities we discuss has previously been explored for this PDE model for multilevel selection with pairwise group-level competition when within-group dynamics feature competition between defectors and altruistic punishers \cite{cooney2024exploring}.

There are multiple ways in which we can map differences in average payoff of group members to the probability of the victory of an $x$-cooperator group over a $y$-cooperator group in a pairwise conflict. Boyd and coauthors considered a group-level victory probability of the form
\begin{equation}
\rho(x,y) = \frac{1}{2} + \frac{1}{2} \left( x - y \right),
\end{equation}
assuming that the advantage in group-level competition should be proportional to the difference in the fraction of cooperative individuals between the two competing groups \cite{boyd2003evolution}.

We can also consider group level victory probabilities that depend on the difference of average payoffs $G(x)$ and $G(y)$ achieved by members of $x$-cooperator and $y$-cooperator groups. A first example of such a group-level victory probability $\rho(x,y)$ is given by

\begin{equation}
\rho(x,y) = \frac{1}{2} \left[ 1 + \frac{G(x) - G(y)}{G^* - G_{*}}\right],
\end{equation}
where $G^* = \max_{x \in [0,1]} G(x)$ and $G_{*} = \min_{x \in [0,1]} G(x)$. This group-level victory probability is based on the local update rule for individual-level learning \cite{traulsen2005coevolutionary}. This is an example of an addtively separable group-level victory probability with net group-level replication rate $\mc{G}(x)= \frac{G(x)}{G^*-G_*}$, so the long-time behavior of the multilevel dynamics for this victory probability can be established using existing results for two-level replicator equations \cite{cooney2022long,cooney2024exploring}. 

As a first example of a group-level victory probability that is not additively separable, we can consider a group-level version of the Fermi update rule for individual-level selection \cite{traulsen2006stochastic}. This constitutes of a victory probability of the form
\begin{equation}
\rho(x,y) = \frac{1}{2} \left(1 + \tanh\left(s \left[ G(x) - G(y)  \right] \right) \right),
\end{equation}
where $s \geq 0$ describes the sensitivity of group-level competition to differences in payoffs between the competing groups. We will use this group-level victory probability for our simulations in the main text, exploring how multilevel dynamics with a nonlinear group-level replication rate can help to promote the long-time evolution of cooperation.  

In the appendix, we will also consider two other group-level victory probabilities that are not additively separable. The first of these victory probabilities is based on models of pairwise competition for individual-level selection \cite{morgan2003pairwise,schluter2016robustness}, and is given by
\begin{equation}
\rho(x,y) = \frac{1}{2} + \frac{1}{2} \left[ \frac{G(x) - G(y)}{|G(x)| + |G(y)|} \right].
\end{equation}
This is a modified version of the group-level local update rule that determines victory based on the difference in average payoff of the competing groups normalized by the absolute values of the payoffs for the two competing groups. The second victory probability we consider in the appendix is
\begin{equation}
\rho(x,y) = \frac{\left(G(x) - G_*\right)^{1/a}}{\left(G(x) - G_*\right)^{1/a} + \left(G(y) - G_*\right)^{1/a}\textbf{}},
\end{equation}
where the parameter $a$ describes a sensitivity of group-level victory probabilities on average payoff of group members. This victory probability is based on the Tullock contest function used in the economics and animal behavior literature to describe the ability to win a contest over resources based on the input effort by competing individuals or groups \cite{tullock2008efficient,rusch2020logic,tverskoi2021dynamics,lerch2024flexible}.

\begin{remark}
The formulation of group-level victory probabilities with parameters like $s$ or $a$ describing sensitivity to payoff differences allows us to separate the effects of the relative speed of group-level competition and the strength of group-level selection in our PDE models of multilevel selection. In such models, the relative speed of group-level selection is quantified by the group level conflict rate $\lambda$, while sensitivity parameters like $s$ and $a$ quantify the strength of group-level selection. This stands in contrast to the case of two-level replicator equations in the form of Equation \eqref{eq:twolevelreplicator}, in which the single parameter $\lambda$ captures the combined effects of the relative rates and strengths of group-level selection events. 
\end{remark}

Having provided a variety of example group-level victory probabilities that depend on the average payoffs of competing groups, we can now use our understanding of the payoff structure for two-player, two-strategy social dilemmas to formulate general assumptions we will use to describe victory probabilities for generalized PD, HD, and SH scenarios. In particular, we would like to describe families of group-level victory probabilities that can generalize the properties of group-level competition that are collected in Lemma \ref{lem:payoffproperties}.

For all of the social dilemma games under consideration, we have that the average payoff function $G(x)$ satisfies $G(1) > G(0)$, so we will assume that the probabilities of victory in group-level competition will satisfy
\begin{equation}
\rho(1,0) > \frac{1}{2} > \rho(0,1).
\end{equation}
For the cases of Hawk-Dove and Stag-Hunt games, we know that the average payoffs satisfy $G(1) > G(x_{eq})$, so we will further assume that
\begin{equation}
\begin{aligned}
\rho(1,x_{eq}) &> \frac{1}{2} > \rho(x_{eq},1) 
\end{aligned}
\end{equation}
For the SH and PDel games, we know that average payoff is an increasing function of $x$ when the fraction of cooperators is close enough to $1$, so we will further impose the assumptions that there is a point $z_{min}$ such that
\begin{equation}
\rho(x,y) > \rho(y,x)
\end{equation}
for any $x > z_{min} > y$ and that
\begin{equation}
\rho(x,u) > \rho(y,u)
\end{equation}
for any $u \in [0,1]$ if $x > z_{min} > u$. These two assumptions captures the fact that there is a level of cooperation $z_{min}$ for which groups featuring $x > z_{min}$ cooperators will have a collective advantage in pairwise competition over all groups featuring a fraction of cooperators less than $z_{min}$. These assumptions reflect the property described in Lemma \ref{lem:payoffproperties} that there is a value $z_{min} > x_{eq}$ such that, for the SH game, the  average payoff $G(x)$ is increasing $x > z_{min}$ and that $G(x)$ exceeds the average payoff $G(y)$ for all lower fractions of cooperation $y < x$ if $x > z_{min}$. We will use these assumptions on the group-level victory probability $\rho(x,y)$ to analyze the long-time behavior for multilevel selection for a generalized SH scenario in Section \ref{sec:SHdelta}.

\section{Measure-Valued Formulation of the PDE Model and Well-Posedness in the Space of Measures}
\label{sec:measureproperties}

In this section, we address general properties of measure-valued solutions to our PDE model for multilevel selection with pairwise group-level competition. We first prove well-posedness of the weak formulation of our PDE model given in Equation \eqref{eq:PDEmeasure} in Section \ref{sec:wellposednessresults}. Our approach will be to consider an associated linear PDE and to use the well-posedness of solutions for the linear PDE and a fixed point argument to show the existence of unique solutions to address well-posedness of Equation \eqref{eq:multilevelPDEtworho}.

We will then apply our implicit representation formula for solutions to Equation \eqref{eq:PDEmeasure}, showing in Section \ref{sec:holderpreserve} that the infimuma and supremum H{\"o}lder exponents of the strategic composition of the population near $x=1$ are preserved under the dynamics of our model of multilevel selection. We will further use this implicit representation formula in Section \ref{sec:PDbehavior} to demonstrate that defectors will take over the entire population when between-group competition is sufficiently weak. 

\subsection{Main Results on Well-Posedness in Space of Measures}
\label{sec:wellposednessresults}

We are looking to establish the existence and uniqueness of a flow of measures $\mu := \{ \mu_t(dx)\}_{t \in [0,T]}$ that satisfies the following weak formulation of our PDE model for multilevel selection
\begin{equation}
\begin{aligned}
\dsdel{}{t} \int_0^1 v(x) \mu_t(dx) &= -\int_0^1 v'(x) x (1-x) \pi(x) \mu_t(dx) + \lambda \int_0^1 v(x) \left[ 2 \int_0^1 \rho(x,y) \mu_t(dy) - 1 \right] \mu_t(dx) 
\end{aligned}
\end{equation}
for any test function $v(x) \in C^1\left([0,1]\right)$ and whose initial distribution is given by a prescribed initial measure $\mu_0(dx)$. 

To further explore this, we now introduce an associated linear PDE that we can use to demonstrate the existence of measure-valued solutions to Equation \eqref{eq:multilevelPDEtworho} through an interation scheme. Given an arbitrary %
flow of measures $\nu := \{\nu_t\}_{t \in [0,T]} \in C\left([0,T];\mc{M}([0,1]\right)$, we can define the following linear PDE 
\begin{equation} \label{eq:PDEhlinear}
\begin{aligned}
\dsddt{} \int_0^1 v(x) \mu_t^{\nu}(dx) &= -\int_0^1 v'(x) x(1-x) \pi(x) \mu_t^{\nu}(dx) + \lambda \int_0^1 v(x) \left[ 2 \int_0^1 \rho(x,y) \nu_t(dy)  - 1 \right] \mu_t^{\nu}(dx) \\
\mu_0^{\nu}(dx) &= \mu_0(dx)
\end{aligned}
\end{equation}
that has the same prescribed initial measure as our full nonlinear PDE model. This is a linear hyperbolic PDE whose characteristic curves are satisfy Equation \eqref{eq:characteristics}, which is the replicator equation for individual-level selection. We denote by $\phi_t(x_0)$ the solution to this ODE starting from the initial point $x_0$, and we denote by $\phi_t^{-1}(x)$ the solutions of the characteristic ODE backward in time from a given point $(t,x)$, meaning that $\phi_t^{-1}(x)$ satisfies

\begin{equation} \label{eq:backwardschar}
\begin{aligned}
\dsddt{\phi_t^{-1}(x)} &= \phi_t^{-1} \left( 1 - \phi_t^{-1}(x) \right) \pi\left(\phi_t^{-1}(x)\right) \\
\phi_0^{-1}(x) &= x. 
\end{aligned}
\end{equation}

To study existence of measure-valued weak solutions to the hyperbolic model, we can adapt an approach previously used to study the two-level replicator equation with a frequency-independent group-level reproduction term \cite{cooney2020analysis}.  Here we are considering continuity of the trajectories $\mu = \{\mu_t\}_{t \geq 0}$ with respect to the weak-$*$ topology on $\mc{M}([0,1])$, meaning that we consider trajectories for which the family of integrals
\begin{equation}
\langle v , \mu_t \rangle := \int_0^1 v(x) \mu_t(dx)
\end{equation}
will be a continuous function of $t$ for all test functions $v(x) \in C([0,1])$. Accordingly, we consider convergence of a flow of measures $\mu = \{\mu_t\}_{t \geq 0}$ to a limit $\mu_{\infty}$ provided that, for each test function $v(x) \in C([0,1])$,
\begin{equation}
\lim_{t \to t_1} \int_0^1 v(x) \mu_t(dx) = \int_0^1 v(x) \mu_{\infty}dx
\end{equation}
for each test function $v(x) \in C\left([0,1]\right)$.

This assumption on the space of solutions is in line with existing work on measure-valued solutions for PDEs arising in evolutionary games featuring mixed strategies or continuous strategy sets \cite{cleveland2013evolutionary,martin2024asymptotic}, models of the dynamics of structured populations \cite{ackleh2005rate,canizo2013measure,ackleh2016population,ackleh2020well,gabriel2022periodic}, and in models motivated by epidemiology \cite{ackleh2024multiple} and neuroscience \cite{canizo2019asymptotic}. 

In this paper, we will consider solutions in the space of flows of measures $C([0,T]; \mc{M}[0,1])$, where $\mc{M}([0,1])$ denote finite signed Borel measures on $[0,1]$. For pairs of measures $\mu_t, \nu_t \in \mc{M}[0,1]$ at a fixed time $t$, we will consider the total variation distance
\begin{equation}
|| \mu_t - \nu_t ||_{TV} := \sup_{\substack{v \in L^{\infty}([0,1]) \\ ||v||_{\infty} = 1}} | \langle v,\mu_t - \nu_t \rangle | = \sup_{\substack{v \in L^{\infty}([0,1]) \\ ||v||_{\infty} = 1}} \bigg| \int_0^1 v(x) \mu_t(dx) - \int_0^1 v(x) \nu_t(dx) \bigg|.
\end{equation}
Because we will consider solutions that are continuous in the weak-* topology, we will also describe the distance between measures in terms of the bounded Lipschitz norm, which is defined as
\begin{equation}
|| \mu_t - \nu_t ||_{BL} := \sup_{\substack{v \in W^{1,\infty}([0,1]) \\  ||v||_{1,\infty} \leq 1}} \bigg| \int_0^1 v(x) \mu_t(dx) - \int_0^1 v(x) \nu_t(dx) \bigg|,
\end{equation}
where $W^{1,\infty}\left([0,1]\right)$ denotes the set of weakly differentiable functions equipped with the norm
\begin{equation}
|| v ||_{1,\infty} := ||v||_{\infty} + ||v'||_{\infty}.
\end{equation}
We can then consider the distance between the flows of measure $\mu = \{\mu_t\}_{t \in[0,T]}$ and $\nu = \{\nu_t\}_{t \in[0,T]}$ using the norm
\begin{equation}
|| \mu - \nu ||_{C([0,T]; \mc{M}[0,1])} := \sup_{t \in [0,T]} || \mu_t - \nu_t ||_{TV}.
\end{equation}

To establish the well-posedness of our full nonlinear PDE model through a fixed point argument, we would like to construct an iteration scheme by defining a new flow of measures $H(\{\nu\})$ using our auxiliary linear PDE. We can define the iteration map as
\begin{equation}
H(\nu) = \{H(\nu)_t \}_{t \in [0,T]}  := \{\mu_t^{\nu}\}_{t \in [0,T]},
\end{equation}
where $\mu^{\nu} := \{\mu_t^{\nu}\}_{t \in [0,T]}$ is the flow of measures that solves Equation \eqref{eq:PDEhlinear} for $t \in [0,T]$ starting from the initial measure $\mu_0$. In particular, we can use the representation formula for the solution to the linear PDE for $\mu^{\nu}$ to write this mapping as a trajectory of the form
\begin{equation}
\begin{aligned}
\int_0^1 v(x) H(\nu)_t(dx) &= \int_0^1 v(x) \mu_t^{\nu}(dx) \\
&= \int_0^1 v(\phi_t(x)) \exp\left( \lambda \left[ \int_0^t \left\{ 2 \int_0^1 \rho(\phi_s(x),y) \nu_s(dy) \right\} - t\right]\right) \mu_0(dx).
\end{aligned}
\end{equation}
\sloppy{In Lemma \ref{lem:linearwellposedness}, we characterize well-posedness of the the associated linear problem by showing that there exists a unique a flow of measures $\mu^{\nu} = \{\mu_t^{\nu}\}_{t \in [0,T]} \in C\left([0,T];\mc{M}([0,1]\right) \cap C^1\left([0,T];\left(C^1\right)^*\right)$ that satisfies the linear PDE of Equation \eqref{eq:PDEhlinear}, so the mapping $H(\{\nu_t\})$ is well-defined. This approach also allows us to obtain an explicit representation formula for the solution to Equation \eqref{eq:PDEhlinear} using the method of characteristics.}

\begin{lemma} \label{lem:linearwellposedness}
\sloppy{Suppose that the group-level victory probability is given by $\rho(x,y) \in C^1\left([0,1]^2\right)$ %
and that the individual-level replication advantage for defectors is given by  $\pi(x) \in C^1\left([0,1]\right)$. Given $T > 0$, an arbitrary initial probability measure $\mu_0(dx)$, and a flow of measures $\nu := \{ \nu_{t} \}_{t \in [0,T]} \in C\left([0,T];\mc{M}([0,1]\right)$, there exists a unique flow of measures $\mu_t^{\nu}(dx) \in C\left([0,T];\mc{M}([0,1]\right) \cap C^1\left([0,T];\left(C^1\right)^*\right)$  that satisfies Equation \eqref{eq:PDEhlinear} for each test function $v(x) \in C^1([0.1])$ and for all $t \in [0,T]$. Furthermore, we can obtain the following representation formula for the flow of measures $\mu_t^{\nu}(dx)$ solving Equation \eqref{eq:PDEhlinear} for all $t \in [0,T]$:}
\begin{equation} \label{eq:hlinearrepresentation}
\int_0^1 v(x) \mu_t^{\nu}(dx) = \int_0^1 v(\phi_t(x)) \exp\left( \lambda \left[ \int_0^t \left\{ 2 \int_0^1 \rho(\phi_s(x),y) \nu_s(dy) \right\} - t\right]\right) \mu_0(dx).
\end{equation}
\end{lemma}

Next, we use a fixed point theorem to show how we can use solutions to the linear PDE of Equation \eqref{eq:PDEhlinear} to demonstrate the existence and uniqueness of a solution to the full nonlinear of Equation \eqref{eq:PDEmeasure} for the multilevel dynamics with pairwise between-group competition. The approach we use builds off of prior work on establishing well-posedness of measure-valued solutions for hyperbolic PDE models in mathematical biology \cite{canizo2013measure,canizo2019asymptotic,ackleh2020well}, and similar approaches have also been applied for related PDE models with more regular classes of solutions \cite{dawidowicz1986existence,eftimie2009weakly,qu2023modeling}. This fixed point argument allows us to obtain an implicit representation formula for the solution $\mu_t$ to the full nonlinear model from Equation \eqref{eq:PDEmeasure}, which we can use to study dynamical properties of solutions.  

\begin{theorem} \label{thm:nonlinearwellposedness}
Suppose that the group-level victory probability %
is given by $\rho(x,y) \in C^1\left([0,1]^2\right)$ for all $(x,y) \in [0,1]^2$,
and that the individual-level replication advantage for defectors is given by  $\pi(x) \in C^1\left([0,1]\right)$. For any $T > 0$ and any initial probability measure $\mu_0(dx)$, there exists a unique flow of measures ${\mu_t(dx)}_{t \in [0,T]}  \in C\left([0,T];\mc{M}([0,1])\right) \cap C^1\left([0,T];\left(C^1\right)^*\right)$  that satisfies Equation \eqref{eq:PDEmeasure} for each test function $v(x) \in C^1([0.1])$ and for all $t \in [0,T]$.The solution satisfies the implicit representation formula
\begin{equation} \label{eq:mutimplicit}
\int_0^1 v(x) \mu_t(dx) = \int_0^1 v(\phi_t(x)) \exp\left( 2 \lambda\int_0^t  \int_0^1 \rho(x,y) \mu_s(dy) ds - \lambda t \right) \mu_0(dx)
\end{equation}
for each $t \in [0,T]$. Furthermore, the measure $\mu_t$ for the  solution $\mu = \{ \mu_t \}_{t \in [0,T]}$ is a probability measure for each $t \in [0,T]$. 
\end{theorem}

\subsection{Preservation of Infimum and Supremum H{\"o}lder Exponent} \label{sec:holderpreserve}

In this section, we describe how the tail behavior of the measure $\mu_t(dx)$ near the all-cooperator equilibrium $x=1$ is preserved in time under the multilevel dynamics of Equation \eqref{eq:PDEmeasure} with pairwise group-level competition. To do this, we use the quantities known as the infimum and supremum H{\"o}lder exponent near $x=1$, which have previously been used to characterize the long-time behavior of PDE models for multilevel selection with frequency-independent group-level competition \cite{cooney2022long}.

\begin{definition} \label{def:infimumHolder}
The infimum \holder exponent $\overline{\theta}$ near $x=1$ satisfies \begin{equation} \label{eq:infholderexponent} \overline{\theta} := \sup \left\{ \Theta \geq 0: \ds\liminf_{x \to 0}\frac{\mu_t([1-x,1])}{x^{\Theta}} = 0 \right\},\end{equation} 
and the associated infimum H{\"o}lder constant $C_{\overline{\theta}}$ is given by
\begin{equation}
    \ds\liminf_{x \to 0} \frac{\mu_t([1-x,1])}{x^{\overline{\theta}}} = C_{\overline{\theta}}.
\end{equation}
\end{definition}
\begin{definition} \label{def:supremumHolder}
The supremum \holder exponent $\underline{\theta}$ near $x=1$ satisfies \begin{equation} \label{eq:supholderexponent} \underline{\theta} := \sup \left\{ \Theta \geq 0: \ds\limsup_{x \to 0}\frac{\mu_t([1-x,1])}{x^{\Theta}} = 0 \right\}, \end{equation}
and the associated supremum H{\"o}lder constant $C_{\underline{\theta}}$ is given by
\begin{equation}
    \ds\limsup_{x \to 0} \frac{\mu_t([1-x,1])}{x^{\underline{\theta}}} = C_{\underline{\theta}}.
\end{equation}
\end{definition}
Intuitively, these two quantities provide upper and lower bounds for the extent to which the strategic distribution of groups $\mu(dx)$ features groups near the all-cooperator composition at $x = 1$. For the two-level replicator equation, it was shown that the supremum H{\"o}lder exponent of the initial population distribution $\mu_0(dx)$ near $x=1$ appeared in the threshold selection strength required to sustain long-term cooperation, while having an initial distribution whose infimum and supremum H{\"o}lder exponents disagreed could result in long-time oscillatory behavior of the multilevel dynamics \cite{cooney2022long}. 

If we further have that a measure $\mu(dx)$ has an infimum and supremum H{\"o}lder exponents and constants near $x=1$ that agree, then we can say that the initial measure has H{\"o}lder exponent $\theta = \overline{\theta} = \underline{\theta}$ and associated H{\"o}lder constant $C_{\theta} := C_{\overline{\theta}} = C_{\underline{\theta}}$ near $x=1$. The H{"o}lder exponent of a measure $\mu(dx)$ near $x=1$ can be described with the following characterization.

\begin{definition}
The measure $\mu(dx)$ has H{\"o}lder exponent $\theta$ with associated H{\"o}lder constant $C_{\theta} \in \RR \cup {\infty}$ near $x=1$ if $\mu(dx)$ has the following tail behavior
\begin{equation}
\lim_{y \to 1} \frac{\mu\left( \left[1-y,1\right]\right)}{y^{\Theta}} = \left\{
    \begin{array}{lr}
      0 & : \Theta < \theta \\
      C_{\theta} & \Theta = \theta \\
      \infty & \Theta > \theta. 
     \end{array}
  \right.
\end{equation}
\end{definition}

We can check from the definition that the family of measures of the form $\mu^{\theta}(dx) = \theta (1-x)^{\theta -1} dx$ have H{\"o}lder exponent $\theta$ with associated H{\"o}lder constant $C_{\theta} = 1$. The measure in this family with $\theta = 1$ is the uniform measure $\mu^1(dx) = dx$, which we will use as our initial population distribution for our numerical simulations of the multilevel dynamics. In prior work on two-level replicator equations, it was shown that the family of density steady states could be parameterized by the H{\"o}lder exponent near $x = 1$ \cite{luo2017scaling,cooney2019replicator,cooney2020analysis,cooney2022long}, so in Section \ref{sec:steadytheta} of the appendix we will also look to characterize necessary conditions for having a steady-state density with H{\"o}lder exponent $\theta$ near $x=1$ for our model of pairwise group-level competition and a generalized multilevel PD scenario. 

In Proposition \ref{prop:Holderpreserve}, we show that the infimum and supremum H{\"o}lder exponents of the population distribution $\mu_t(dx)$ are preserved under the multilevel dynamics described by Equation \eqref{eq:PDEmeasure}. This result suggests that it may be possible to approach understanding the dynamics of our model of multilevel selection with pairwise group conflict in a similar manner to existing results for two-level replicator equations that rely on characterizing the behavior of solutions in terms of the infimum or supremum H{\"o}lder exponents of the initial measure $\mu_0$ near $x=1$. We present the proof of this proposition in Section \ref{sec:infsupHolderpreserved} of the appendix.

\begin{proposition} \label{prop:Holderpreserve} Suppose the initial population is given by measure $\mu_0(dx)$ with infimum H{\"o}lder exponent $\overline{\theta}$ and supremum H{\"o}lder exponent $\underline{\theta}$ near $x = 1$. Then, for all times $t \geq 0$, the measure $\mu_t(dx)$ solving Equation \eqref{eq:PDEmeasure} has infimum and supremum H{\"o}lder exponents near $x = 1$ given by $\overline{\theta}_t = \overline{\theta}$ and $\underline{\theta}_t = \underline{\theta}$. \end{proposition}

\begin{remark}
The infimum and supremum H{\"o}lder exponents are also preserved in time for the solution to the two-level replicator equation models that have been used to study multilevel selection with frequency-independent group-level competition \cite{cooney2020analysis,cooney2022long}. For the two-level replicator equation, it is also possible to show that the H{\"o}lder exponent itself is preserved in time, which served as an initial suggestion that, out of the infinitely many density steady state solutions for the two-level replicator equation, the long-time behavior of multilevel selection would converge to the unique steady state with the H{\"o}lder exponent for the initial strategic composition of the population. The proof of the preservation of the infimum and supremum H{\"o}lder exponents under the dynamics of Equation \eqref{eq:PDEmeasure} relies only on the fact that $\rho(x,y)$ is a probability, so a more careful study of our representation formula for measure-valued solutions may reveal more precise estimates that will allow us to show that a population with initial H{\"o}lder exponent $\theta$ near $x=1$ will continue to have the same well-defined H{\"o}lder exponent for all positive time. 
\end{remark}

The proof of Proposition \ref{prop:Holderpreserve} relies on the following lemma providing an expression for the impact of within-group selection dynamics on the trajectories of the characteristic curves. This calculation was originally performed for a two-level replicator equation in \cite[Lemma 3]{cooney2022long}, but the result carries through directly for our model with pairwise group-level competition because the advection term is identical for the two models. 

 \begin{lemma}[Originally {\cite[Lemma 3]{cooney2022long}}]%
 \label{lem:backwardcharacteristics}
Suppose that $\pi(x) \in C^1[0,1]$ and $\pi(x) > 0$ for $x \in [0,1]$, and let $\phi_t^{-1}$ be the backward characteristic curve solving \eqref{eq:backwardschar}. For $0<x\leq 1$, we have that
\begin{equation}\label{zest}
\begin{aligned}
\exp(\pi(1)t)(1-\phi_t^{-1}(x))&=(1-x)\exp\left(\int_x^{\phi_t^{-1}(x)} \frac{Q(s)ds}{s\pi(s)}\right),\\
Q(s)&=\frac{\pi(1)-\pi(s)}{1-s}+\pi(s).
\end{aligned}
\end{equation}
Furthermore, we note from our assumptions that $\pi(x) \in C^1([0,1])$ and $\pi(x) > 0$ for $x \in [0,1]$ that the integrand $\tfrac{Q(s)}{s \pi(s)}$ is bounded for $ \in [x,1]$ for any $x > 0$.
\end{lemma}

\section{Dynamical Behavior of the PDE Model for the Generalized PD Scenario}
\label{sec:PDbehavior}

In this section, we illustrate some of the dynamical behaviors exhibited by our PDE model for multilevel selection with pairwise between-group competition when the replication rates at each level reflect a generalized Prisoners' Dilemma scenarios. In Section \ref{sec:PDdelta}, we show for multilevel PD scenarios that the population can converge to a delta-function concentrated at the all-defector equation $x=0$ when the strength of between-group competition is sufficiently weak (corresponding to $\lambda$ close enough to $0$). In Section \ref{sec:PDsteadygroupsuccess}, we provide necessarry conditions for the existence of a steady-state density supporting positive levels of cooperation in the population and develop a measure of the collective success of such steady-state populations when engaged in group-level conflict with all-cooperator groups. We complement these analytical explorations of the PDE model with numerical simulations for example scenarios of Prisoners' Dilemma, Hawk-Dove, and Stag-Hunt games, which we present in Section \ref{sec:numerics}.

\subsection{Convergence to Delta-Function at All-Defector Composition} \label{sec:PDdelta}

In Proposition \ref{prop:convergencetoalldefection}, we show that defectors take over the entire population when between-group competition is sufficiently weak. Mathematically, this corresponds to showing that the measure $\mu_t(dx)$ converges to a delta-function $\delta(x)$ concentrated at the all-defector equilibrium when $\lambda$ is sufficiently close to $0$. In particular, we show this concentration upon the all-defector equilibrium in the sense of weak convergence, showing that, for any test function $v(x) \in C\left([0,1]\right)$, the solution $\mu_t(dx)$ to the measure-valued multilevel dynamics of Equation \eqref{eq:PDEmeasure} satisfy $\int_0^1 v(x) \mu_t(dx) \to \int_0^1 v(x) \delta(x) = v(0)$ as $t \to \infty$, and we denote this weak convergence to a delta-measure at $x =0$ by $\mu_t(dx) \rightharpoonup \delta(x)$ as $t \to \infty$. 

\begin{proposition} \label{prop:convergencetoalldefection}
Consider a within-group replication rate satisfying $\pi(x) > 0$ for $x \in [0,1]$ and a group-level victory probability $\rho(x,y) \in C^1([0,1]^2)$. Suppose the initial measure $\mu_0(dx)$ has supremum H{\"o}lder exponent $\overline{\theta} > 0$ near $x = 1$. If $\lambda \left( 2 ||\rho||_{L^{\infty}([0,1]^2)} - 1\right) < \overline{\theta} \pi(1)$, then the solution $\mu_t(dx)$ to Equation \eqref{eq:PDEmeasure} satisfies $\mu_t(dx) \rightharpoonup \delta(x)$ as $t \to \infty$. 
\end{proposition}

\begin{remark}
The condition $\lambda \left( 2 ||\rho||_{L^{\infty}([0,1]^2)} - 1 \right) < \overline{\theta} \pi(1)$ in Proposition \ref{prop:convergencetoalldefection} can be rewritten to see that we require that the between-group selection strength $\lambda$ to satisfy
\begin{equation}
\lambda < \frac{\overline{\theta} \pi(1)}{2 ||\rho||_{L^{\infty}([0,1]^2)} - 1}.
\end{equation}
\begin{equation} \label{eq:deltathresholdtugofwar}
\lambda < \frac{\overline{\theta} \pi(1)}{\max_{(x,y) \in [0,1]^2} \left[\rho(x,y) - \rho(y,x) \right]}
\end{equation}
We expect this that threshold for convergence to a delta-function at the all-defector is not optimal. Using calculations of properties of potential density steady states in Section \ref{sec:PDsteadygroupsuccess} and numerical simulations in Section \ref{sec:numerics}, we identify a conjectured threshold quantity for the maximal between-group selection strength $\lambda$ for which cooperation will go extinct under the dynamics of multilevel selection. Namely, we expect that a population whose initial measure has supremum H{\"o}lder exponent $\overline{\theta}$ near $x = 1$ will converge to a delta-function $\delta(x)$ at the all-defector equilibrium provided that
\begin{equation}
\lambda \leq \lambda^*_{PD}(\overline{\theta}) := \frac{\overline{\theta} \pi(1)}{\rho(1,0)- \rho(0,1)}.
\end{equation}
This coincides with the sufficient condition we found in Proposition \ref{prop:convergencetoalldefection} when $\rho(x,y)$ is increasing in $\rho(1,0) = \max_{(x,y) \in [0,1]^2} \rho(x,y)$, but we expect this result to hold for a much more general class of group-level victory probability functions $\rho(x,y)$. 
\end{remark}

To prove Proposition \ref{prop:convergencetoalldefection}, we adopt the strategy used in the proof of \cite[Theorem 3]{cooney2022long} for the case of convergence to the all-defector state in a model of multilevel selection with a linear between-group replication term. 

\begin{proof}[Proof of Proposition \ref{prop:convergencetoalldefection}]
To show weak convergence to a delta function $\delta(x)$ at the all-defector equilibrium, we consider an arbitrary test function $v(x) \in C([0,1])$ and look to show that $\int_0^1 \mu_t(dx) \to \int_0^1 v(x) \delta(x) = v(0)$. Using the assumption that $v(x)$ is continuous and the fact that the solution $\mu_t(dx)$ to Equation \eqref{eq:multilevelPDEtworho} is a probability measure, we can see that, for any $\epsilon > 0$, there is a $\delta > 0$ such that
\begin{equation}
\begin{aligned}
\bigg| \int_0^1 v(x) \mu_t(dx) - v(0) \bigg| &= \bigg| \int_0^1 \left( v(x) - v(0) \right) \mu_t(dx) \bigg|  \\ &\leq \int_0^1 |v(x) - v(0)| \mu_t(dx) \\ & \leq \int_{0}^{\delta} |v(x) - v(0) | \mu_t(dx) + \int_{\delta}^1 | v(x) - v(0)| \mu_t(dx) \\  & \leq \epsilon + 2 ||v||_{\infty} \int_{\delta}^1\mu_t(dx)
\end{aligned}
\end{equation}
Using the implicit representation formula for $\mu_t(dx)$, we can further note that
\begin{equation}
\begin{aligned}
\int_{\delta}^1 \mu_t(dx) &= \int_{\phi_t^{-1}(\delta)}^1 \exp\left( 2 \int_0^t  \int_0^1 \rho(x,y) \mu_s(dy) ds - t \right) \mu_0(dx) \\
& \leq \exp\left( \lambda \left[ 2 ||\rho||_{L^{\infty}([0,1]^2)} - 1  \right] t\right) \int_{\phi_t^{-1}(\delta)}^1 \mu_0(dx) \\
&= \exp\left( \lambda \left[ 2 ||\rho||_{L^{\infty}([0,1]^2)} - 1  \right] t\right) \mu_0\left( \left[\phi_t^{-1}(\delta),1\right] \right)
\end{aligned}
\end{equation}
Furthermore, we can use Lemma \ref{lem:backwardcharacteristics} to say there is a positive constant $C < \infty$ such that
\begin{equation}
\begin{aligned}
\mu_0\left([\phi_t^{-1}(\delta),1] \right) &= \mu_0\left(\left[1 - \left( 1 - \phi_t^{-1}(\delta) \right), 1 \right]  \right) \\
&= \mu_0\left( \left[ 1 - e^{- \pi(1) t} \left(1-\delta\right) \exp\left(\int_{\delta}^{\phi_t^{-1}(\delta)} \frac{Q(s)ds}{s\pi(s)}\right) \right]\right).
\end{aligned}
\end{equation}
Using the assumption that the initial measure $\mu_0(dx)$ has a supremum H{\"o}lder exponent $\overline{\theta} > 0$ near $x=1$, we know that, for any $\Theta < \overline{\theta}$, there is a positive constant $C_{\Theta}$ such that 
\begin{equation}
\begin{aligned}
\mu_0\left([\phi_t^{-1}(\delta),1] \right) & \leq C_{\Theta} \exp\left( - \Theta \pi(1) t \right) \underbrace{\left(1 - \delta\right)^{\Theta} \exp\left(\Theta \int_{\delta}^{\phi_t^{-1}(\delta)} \frac{Q(s)ds}{s\pi(s)}\right)}_{:= B(\Theta)}.
\end{aligned}
\end{equation}
We can combine our previous bounds to deduce that
\begin{equation}
\bigg| \int_0^1 v(x) \mu_t(dx) - v(0) \bigg| \leq \epsilon + 2 ||v||_{\infty} C_{\Theta} B(\Theta) \exp\left( \left\{ \lambda \left[2 ||\rho||_{L^{\infty}([0,1]^2)} - 1\right]  - \Theta \pi(1) \right\} t \right).
\end{equation}
Because we have assumed the strict inequality $ \lambda \left[2 ||\rho||_{L^{\infty}([0,1]^2)} - 1\right] < \overline{\theta} \pi(1)$, we know that we can choose $\Theta < \overline{\theta}$ sufficiently close to $\overline{\theta}$ to satisfy 
\[ \lambda \left[2 ||\rho||_{L^{\infty}([0,1]^2)} - 1\right] < \Theta \pi(1)  \]
and 
\[2 ||v||_{\infty} C_{\Theta} B(\Theta) \exp\left( \left\{ \lambda \left[2 ||\rho||_{L^{\infty}([0,1]^2)} - 1\right]  - \Theta \pi(1) \right\} t \right) \to 0 \: \: \mathrm{as} \: \: t \to \infty.  \]
This allows us to conclude that there is sufficiently large $T$ such that, for any test function $v(x) \in C(0,1])$ and any $t > T$, 
\begin{equation}
\bigg| \int_0^1 v(x) \mu_t(dx) - v(0) \bigg| \leq 2 \epsilon,
\end{equation}
and therefore the solution $\mu_t(dx) \rightharpoonup \delta(x)$. 
\end{proof}

\subsection{Necessary Conditions for Steady-State Date Density Supporting Cooperation}
\label{sec:PDsteadygroupsuccess}

We now consider the properties of possible density steady state solutions to Equation \eqref{eq:multilevelPDEtworho} for the case of a PD game. We study the behavior of the collective group-level success for any differentiable steady state density that places nonzero density on the all-cooperator composition $f(x)$ (corresponding to case of a steady state population with H{\"o}lder exponent $\theta = 1$ near $x=1$). This choice of H{\"o}lder exponent is motivated by the comparisons we will make to finite volume numerical simulations in Section \ref{sec:numerics}, in which we represent our numerical solutions as piecewise constant densities that remain nonzero in the limit as $x \to 1$. We also briefly mention generalizations of these formulas on the collective success and threshold selection strength to the case of density steady states with arbitrary H{\"o}lder exponent $\theta > 0$ near $x=1$, delaying the details for these calculations to Section \ref{sec:steadytheta}.

In Proposition \ref{prop:PDsteadybounded}, we characterize the average group-level victory probability of $\int_0^1 \rho(y,1) f(y) dy$ against the all-cooperator group for any density steady state $f(x)$ of the multilevel dynamics that is continuously differentiable on $[0,1]$ and satisfies $f(1) \ne 0$. We then compare this expression to the group-level victory probability $\rho(0,1)$ of an all-defector group when engaged in group-level competition with an all-cooperator group, allowing us to find a threshold strength of between-group selection $\lambda$ for which a density steady state produces a better collective outcome than the all-defector group.
}
\begin{proposition} \label{prop:PDsteadybounded}
Consider a group-level victory probability $\rho(x,y) \in C^1\left(\left[0,1\right]^2 \right)$ and a relative individual-level advantage of defectors $\pi(x) \in C^1([0,1])$ satisfying $\pi(x) > 0$ for all $x \in [0,1]$. Suppose that Equation \eqref{eq:multilevelPDEtworho} has a steady state solution $f(x)$ that is a probability density $f(x) \in C([0,1])$ whose behavior near $x=1$ satisfies $\lim_{x \to 1} f(x) = L_1 \ne 0$. Then the average group-level victory probability of the steady state population in pairwise competition with the all-cooperator group must satisfy
\begin{equation} \label{eq:rhoy1statement}
\int_0^1 \rho(y,1) f(y) dy = \frac{1}{2} - \frac{\pi(1)}{2 \lambda}. 
\end{equation}
Furthermore, this group-level success $\int_0^1 \rho(y,1) f(y) dy$ for the steady state $f(x)$ will exceed the group-level victory probability $\rho(0,1)$ of the all-defector group against the all-cooperator group when the $\lambda$ exceeds the threshold value
\begin{equation}
\lambda > \lambda^*_{PD} = \frac{\pi(1)}{\rho(1,0) - \rho(0,1)}.
\end{equation}
\end{proposition}

\begin{remark}
By taking the limit of the expression in Equation \eqref{eq:rhoy1statement} for average group-level success as $\lambda \to \infty$, we see that see that
\begin{equation}
\lim_{\lambda \to \infty} \rho(y,1) f(y) dy = \frac{1}{2},
\end{equation}
so the groups sampled from the steady-state density have a fifty-fifty chance of defeating an all-cooperator group in a pairwise group conflict in the limit of infinitely strong between-group competition. This is the direct analogue of the shadow of lower-level selection for the case of pairwise between-group competition, as the steady-state population in the limit of strong between-group competition only consists of groups that are equal to the all-cooperator group under pairwise competition. 
\end{remark}

\begin{proof}[Proof of Proposition \ref{prop:PDsteadybounded}]
A density $f(x)$ is a steady state solution to Equation \eqref{eq:multilevelPDEtworho} must satisfy
\begin{equation}
- \dsddx{}{x} \left( x (1-x) \pi(x) f(x) \right) = \lambda f(x) \left[ 2 \int_0^1 \rho(x,y) f(y) dy - 1 \right].
\end{equation}
This ODE for the steady state density can be rewritten as
\begin{equation} \label{eq:steadyfprime}
f'(x) = - \left(\frac{f(x)}{x (1-x) \pi(x)} \right) \left[ \lambda - 2 \lambda  \int_0^1 \rho(x,y) f(y) dy - (1-2x) \pi(x) \right] - \left(\frac{\pi'(x)}{\pi(x)} \right) f(x)
\end{equation}
For the PD case of the multilevel dynamics, we assume that $\pi(x) > 0$, so the second term is bounded as $x \to 1$. If we want that the steady state density remains bounded as $x \to 1$, then we want $f'(x)$ to approach a finite limit as $x \to 1$. We can then note that the limit of the first term on the righthand side of Equation \eqref{eq:steadyfprime} will be finite and nonzero only if
\begin{equation}
\lim_{x \to 1} \left[ \lambda - 2 \lambda  \int_0^1 \rho(x,y) f(y) dy - (1-2x) \pi(x) \right]  = 0.
\end{equation}
We can evaluate the limit on the lefthand side to see that $f(x)$ can be a steady state density of the multilevel dynamics that remains bounded in the limit as $x \to 1$ only if
\begin{equation}
\lambda - 2 \lambda \int_0^1 \rho(1,y) f(y) dy + \pi(1) = 0.    
\end{equation}
This can be rearranged to see that the average group-level reproduction rate of the all-cooperator must satisfy
\begin{equation} \label{eq:rho1xsteadystate}
\int_0^1 \rho(1,y) f(y) dy = \frac{1}{2} + \frac{\pi(1)}{2 \lambda}. 
\end{equation}
Using the fact that $\rho(1,y) = 1 - \rho(y,1)$, we may also express that the average group-level victory probability $\int_0^1 \rho(y,1) f(y) dy$ for a population at steady state $f(x)$ against an all-cooperator group may be written as
\begin{equation} \label{eq:rhoy1steadystate}
\int_0^1 \rho(y,1) f(y) dy = \frac{1}{2} - \frac{\pi(1)}{2 \lambda}.
\end{equation}

Furthermore, if $f(y)$ is a steady state probability density of Equation \eqref{eq:multilevelPDEtworho} that supports a positive level of cooperation, we would expect that the group-level success for groups in the steady-population should exceed the collective success of the all-defector group (as the within-group dynamics acting alone would result in the extinction of cooperation). We can try to quantify this relative success of the steady state population by comparing the average group-level victory probabilities $\int_0^1 \rho(y,1) f(y) dy$ and $\rho(0,1)$ for both the steady state population and the all-defector group when engaged in a pairwise competition with the all-cooperator group. In particular, we look to determine the conditions under which $\int_0^1 \rho(y,1) f(y) dy > \rho(0,1)$, allowing us to see the conditions under which pairwise between-group competition may help to promote the evolution of cooperation.

We can use Equation \eqref{eq:rhoy1steadystate} to write the condition $\int_0^1 \rho(y,1) f(y) dy > \rho(0,1)$ can be written in the form
\begin{equation}
\int_0^1 \rho(y,1) f(y) dy =  \frac{1}{2} - \frac{\pi(1)}{2 \lambda} >  \rho(0,1),
\end{equation}
which will only be satisfied when
\begin{equation}
\lambda > \frac{\pi(1)}{1 - 2 \rho(0,1)}. 
\end{equation}
Noting that $\rho(1,0) = 1 - \rho(0,1)$, we may further write this condition as
\begin{equation}
\lambda > \frac{\pi(1)}{\rho(1,0) - \rho(0,1)},
\end{equation}
and we have established our desired necessary conditions for the existence of a steady-state density $f(x)$ that remains bounded in the limit as $x \to 1$ and for which the collective victory probability of the steady-state population exceeds that of an all-defector group when pitted in group-level conflict with an all-cooperator group. 
\end{proof}

We further show in the appendix that, if there exists a steady state density $f(x)$ with H{\"o}lder exponent $\theta$ near $x=1$, then the average group-level victory probability at steady state against an all-cooperator group will be given by
\begin{equation}
\int_0^1 \rho(y,1) f(y) dy = \frac{1}{2} - \frac{\theta \pi(1)}{2 \lambda},
\end{equation}
and we have that the such a density steady state will have an expected group-level success against an all-cooperator group that exceeds the corresponding group-level success of an all-defector group $\rho(0,1)$ provided that $\lambda$ exceeds the following threshold quantity 
\begin{equation} \label{eq:lambdastarPDtheta}
\lambda > \lambda^*_{PD}(\theta) := \frac{\theta \pi(1)}{\rho(1,0) - \rho(0,1)}.
\end{equation}
We conjecture that this quantity $\lambda^*_{PD}(\theta)$ serves as a threshold determining the existence of integrable steady state densities with H{\"o}lder exponent $\theta$ near $x=1$, as well as a threshold quantity for determining the long-time survival of cooperation under the dynamics of Equation \eqref{eq:multilevelPDEtworho} starting from an initial population with H{\"o}lder exponent $\theta$ near $x=1$. These conjectures are based on the similarities seen between existing work on PDE models with multilevel selection with frequency-independent between-group competition \cite{luo2017scaling,cooney2019replicator,cooney2022long} and the threshold behavior we see when varying $\lambda$ in our numerical simulations in Section \ref{sec:numerics}.

\begin{remark} \label{rem:PDadditive}
For the case of an additively separable group-level victory probability of the form
\begin{equation}
\rho(x,y) = \frac{1}{2} + \frac{1}{2} \left[ \mc{G}(x) - \mc{G}(y) \right],
\end{equation}
the threshold between-group selection strength $\lambda^*_{PD}$ from Equation \eqref{eq:lambdastarPDtheta} reduces to 
\begin{equation}
\lambda^*_{PD}(\theta) = \frac{\theta \pi(1)}{\mc{G}(1) - \mc{G}(0)},
\end{equation}
which is the threshold selection strength for the survival of cooperation for the PDE model of multilevel selection with when frequency-independent group-level reproduction takes place at rate $\mc{G}(x)$ \cite{cooney2022long}. 
\end{remark}

\section{Dynamical Behavior for Generalized HD and SH Games}
\label{sec:HDandSHbehavior}

We will now extend our analysis of PDE models of multilevel selection to consider cases beyond the generalized PD scenario. We now consider how multilevel selection with pairwise group-level competition can impact the evolution of cooperative behaviors for cases in which individual-level and group-level replication rates reflect evolutionary tensions inspired by HD or SH games. In Section \ref{sec:HDsteady}, we propose a conjectured threshold selection strength required to sustain a steady state density for a generalized HD scenario. In Section \ref{sec:SHdelta}, we consider the dynamics for a generalization of the Stag-Hunt game, and we show for a wide class of density-valued initial conditions that the population concentrates upon a delta-function at the all-cooperator equilibrium in the presence of any between-group competition (corresponding to $\lambda > 0$). We will then compare these analytical results and conjectures with numerical solutions to the multilevel dynamics for the HD game and SH game in Sections \ref{sec:HDnumerics} and \ref{sec:SHnumerics}, respectively.

\subsection{Properties of Density Steady States for Generalized HD Scenario}

\label{sec:HDsteady}

For the HD game, the within-group dynamics push for convergence towards an intermediate equilibrium level of cooperation $x_{eq}$, while between-group selection supports groups featuring compositions with higher average payoff than that achieved by the group with $x_{eq}$ cooperators. Therefore we could potentially expect that the population will support steady state densities $f(x)$ with cooperation in excess of $x_{eq}$ when the average success of an all-cooperator success at steady state $\int_0^1 \rho(1,y) f(y) dy$ is less than the group-level victory probability $\rho(1,x_{eq})$ of the all-cooperator group over the equilibrium group with $x_{eq}$ cooperators. Using Equation \eqref{eq:rho1xsteadystate}, we see that this condition $\int_0^1 \rho(1,x_{eq}) f(y) dy > \rho(1,x_{eq})$ will occur provided that
\begin{equation}
\int_0^1 \rho(1,y) f(y) dy = \frac{1}{2} + \frac{\pi(1)}{2 \lambda} > \rho(1,x_{eq}),
\end{equation}
which can be rearranged to obtain the condition
\begin{equation}
\lambda >  \frac{\pi(1)}{2 \rho(1,x_{eq}) - 1}.
\end{equation}
Using the fact that $\rho(1,x_{eq}) = 1 - \rho(x_{eq},1)$, we may further write this condition in terms of a critical group-level selection strength $\lambda^*_{HD}$
\begin{equation}
\lambda > \lambda^*_{HD} > \frac{\pi(1)}{\rho(1,x_{eq}) - \rho(x_{eq},1)}.
\end{equation}
Using similar reasoning to the case of the multilevel PD dynamics, we may conjecture that there is a threshold selection strength
\begin{equation} \label{eq:lambdastarHDtheta}
\lambda^*_{HD}(\theta) := \frac{\theta \pi(1)}{\rho(1,x_{eq}) - \rho(x_{eq},1)}
\end{equation}

\begin{remark}
For the case of additively separable group-victory probabilities of the form
\begin{equation}
\frac{1}{2} + \frac{1}{2} \left( \mc{G}(x) - \mc{G}(y) \right),
\end{equation}
we see that this threshold $\lambda^*_{HD}$ takes the form
\begin{equation}
\lambda^*_{HD} = \frac{\pi(1)}{\mc{G}(1) - \mc{G}(x_{eq})},
\end{equation}
which is the threshold selection strength required to obtain cooperation above the level of $x_{eq}$ under multilevel selection with frequency-independent group-level competition for the HD case \cite[Theorem 6]{cooney2022pde}.  
\end{remark}

\subsection{Convergence to Delta-Function at All-Cooperator Equilibrium for SH and PDel Games} \label{sec:SHdelta}

In this section, we study the multilevel dynamics with pairwise between-group competition with replication rates arising from payoffs in the SH and PDel games. Unlike the PD and HD games, the all-cooperator equilibrium is locally stable under within-group competition for these two games. For a class of within-group relative replication rates $\pi(x)$ and group victory probabilities $\rho(x,y)$ generalizing the dynamics of the SH game, we show that the population will end up concentrating upon a delta-function at the all-cooperator composition provided that there is any between-group competition ($\lambda > 0$) and that there is a positive initial probability of groups with a fraction of cooperators above the level seen in the basin of attraction for the all-cooperator equilibrium under the within-group replicator dynamics. We will prove our result using assumptions made to generalize the scenario for multilevel selection with pairwise group conflict motivated by payoffs from the SH game, but analogous assumptions on the group-level victory probability $\rho(x,y)$ and the local stability of the all-cooperator equilibrium are sufficient to prove an analogous concentration result for the PDel game.

\begin{proposition} \label{prop:SHdelta}
\sloppy{Consider a within-group relative reproduction rate $\pi(x) \in C^2([0,1])$ with an equilibrium $x_{eq} \in (0,1)$ such that $\pi(x) < 0$ for $x \in (x_{eq},1]$, and consider a group-level victory function $\rho(x,y) \in C^1([0,1]^2)$ %
such  that there is a fraction of cooperation $z_{min} > x_{eq}$ for which the following two conditions are satisfied:
\begin{itemize}
    \item $\rho(x,y) > \rho(w,y)$ for any $z > z_{min} > w$ and for all $y \in [0,1]$
    \item $\rho(x,y) > \frac{1}{2} > \rho(y,x)$ for any $x,y$ satisfying $x > z_{min} > y$.
\end{itemize} Suppose the initial population $\mu_0(dx) = f_0(x) dx$ has density $f_0(x)$ satisfying $\int_{x_{eq}}^1 f_0(x) > 0$. If $\lambda > 0$, then $\mu_t(dx) = f(t,x) dx \rightharpoonup \delta(x-1)$ as $t \to \infty$.} 
\end{proposition}

To prove this, we adopt the strategy used to prove convergence of density-valued solutions to a delta-function at the all-cooperator equilibrium for multilevel selection in SH games with frequency-independent group-level competition \cite[Proposition 8]{cooney2020analysis}. The main changes in our approach are adapting our estimates to incorporate the role of pairwise between-group competition and formulating the minimal assumptions on the group-level reproduction rate $\rho(x,y)$ to reflect a generalization of the dynamics of an SH game that will result in sufficient group-level success of the all-cooperator group. 

\begin{proof}[Proof of Proposition \ref{prop:SHdelta}]
To show concentration of our solution $f(t,x)$ to the multilevel dynamics of Equation \eqref{eq:multilevelPDEtworho}, we will look to characterize the dynamics of the probability of groups contained in the interval $\mc{I}_z := [z,1]$. We then can define this probability as the following function of time
\begin{equation}
P_{\mathcal{I}_z}(t) := \int_z^1 f(t,x) dx,
\end{equation}
and long-time convergence to a delta-function $\delta(x-1)$ at the all-cooperator equilibrium would correspond to having $P_{\mathcal{I}_z}(t) \to 1$ as $t \to \infty$ for all $z$ close enough to $1$. 

For any $z > \min(z_{min},x_{eq})$, we can use our assumption that $\pi(z) < 0$ to compute that
\begin{equation}
\begin{aligned}
\dsdel{}{t} \int_z^1 f(t,x) dx &= x(1-x) \pi(x) f(t,x) \bigg|_{x=z}^{x=1} + \lambda \int_z^1 f(t,x) \left[ \int_0^1 \left( \rho(x,y) - \rho(y,x) \right) f(t,y) dy \right]dx \\
&= -z(1-z) \pi(z) f(t,z) + \lambda \int_z^1 f(t,x) \left[ \int_0^1 \left( \rho(x,y) - \rho(y,x) \right) f(t,y) dy \right]dx  \\
& \geq \lambda \int_z^1 f(t,x) \left[ \int_0^1 \left( \rho(x,y) - \rho(y,x) \right) f(t,y) dy \right]dx  \\
& = \lambda \int_z^1 \int_z^1 \left[ \rho(x,y) - \rho(y,x) \right] f(t,x) f(t,y) dx dy \\ &+ \lambda \int_z^1 \int_0^z \left[ \rho(x,y) - \rho(y,x) \right] f(t,x) f(t,y) dx dy.
\end{aligned}
\end{equation}
Noting that $\rho(x,y) - \rho(y,x)$ is anti-symmetric about the line $y = x$, we can deduce that the integral 
\begin{equation}
\lambda \int_z^1 \int_z^1 \left[ \rho(x,y) - \rho(y,x) \right] f(t,x) f(t,y) dx dy = 0 
\end{equation}
because the domain $[z,1] \times [z,1]$ is symmetric about this line. We can combine this with the assumption that $\rho(x,y) > \rho(y,x)$ for $x > z > y$ to deduce that 
\begin{equation}
\dsdel{}{t} \int_z^1 f(t,x) dx \geq \lambda \int_z^1 \int_0^z \left[ \rho(x,y) - \rho(y,x) \right] f(t,x) f(t,y) dx dy \geq 0,
\end{equation}
where the last inequality only holds strictly if $\int_z^1 f(t,x) dx = 0$ or $\int_z^1 f(t,x) dx = 1$. We have therefore shown that $P_{\mathcal{I}_z}(t)$ is a non-decreasing sequence. In addition, we can use the fact that $f(t,x)$ is a probability density to deduce that $P_{\mathcal{I}_z}(t) \leq 1$, so $P_{\mathcal{I}_z}(t)$ is a non-decreasing sequence that is bounded above, and we can conclude that there is a limit $P^*_{\mc{I}_z}$ such that $P_{\mathcal{I}_z}(t) \to P^*_{\mc{I}_z}$ as $t \to \infty$. 

Next, we need to show that this limit satisfies $ P^*_{\mc{I}_z} = 1$ for any $z \in [0,1)$. We assume for contradiction that there is a $z > \max(z_{min},x_{eq})$ such that $\lim_{t \to \infty} P_{\mathcal{I}_z}(t) \to P^*_{\mathcal{I}_z} < 1$ as $t \to \infty$. Then we see that, for any $z' > z$, that $P_{\mc{I}_{z'}}$ satisfies
\begin{equation}
\begin{aligned}
\dsddt{} P_{\mc{I}_{z'}}(t) &\geq \lambda \int_{z'}^1 \int_{0}^{z'} \left[ \rho(x,y) - \rho(y,x) \right] f(t,y) f(t,x) dy dx \\
&= \lambda \int_{z'}^1 \int_{0}^{z} \left[ \rho(x,y) - \rho(y,x) \right] f(t,y) f(t,x) dy dx  \\ &+ \lambda \int_{z'}^{1} \int_{z}^{z'} \left[ \rho(x,y) - \rho(y,x) \right] f(t,y) f(t,x) dy dx 
\end{aligned}
\end{equation}
From our assumption that $\rho(x,y) > \rho(y,x)$ for $x > y > z_{min}$, we can further deduce that 
\begin{equation}
\dsddt{} P_{\mc{I}_{z'}}(t) \geq \lambda \int_{z'}^1 \int_{0}^{z} \left[ \rho(x,y) - \rho(y,x) \right] f(t,y) f(t,x) dy dx. 
\end{equation}
From the assumption that $\rho(x,y)$ is increasing in $x$ for $x > z_{min}$ and that $\rho(x,y) > \rho(y,x)$ for $x > z_{min} > y$, there is a constant $C > 0$ such that $\rho(x,y) - \rho(y,x) > C$ for all $(x,y) \in [z',1] \times [0,z]$, and we can deduce that
\begin{equation}
\begin{aligned}
\dsddt{} P_{\mc{I}_{z'}}(t) &\geq \lambda C \int_{z'}^1 \int_{0}^{z}  f(t,y) f(t,x) dy dx \\
& \geq \int_{z'}^1 f(t,x) \left( \int_0^z f(t,y) dy \right) dx \\
& \geq \lambda C \left( 1 -  P_{\mc{I}_{z}}(t) \right) \int_{z'}^1 f(t,x) dx,
\end{aligned}
\end{equation}
and we can conclude that $ P_{\mc{I}_{z'}}(t)$ satisfies the differential inequality
\begin{equation}
 P_{\mc{I}_{z'}}(t) \geq \lambda C P_{\mc{I}_{z'}}(t) \left( 1 - P_{\mc{I}_{z}}(t) \right) 
\end{equation}
As we have assumed that $P_{\mc{I}_{z}}(t) \to P_{\mc{I}_{z}}^* < 1$ as $t \to \infty$, we know that there exists $\epsilon > 0$ such that $1 - P_{\mc{I}_{z}}(t) > \epsilon$ for all $t \geq 0$, so we further have the bound
\begin{equation}
\dsddt{} P_{\mc{I}_{z'}}(t) \geq \lambda \epsilon C P_{\mc{I}_{z'}}(t),
\end{equation}
and we can solve this differential inequality to see that
\begin{equation}
P_{\mc{I}_{z'}}(t) \geq P_{\mc{I}_{z'}}(0) e^{\lambda \epsilon C t}.
\end{equation}
We therefore see that $P_{\mc{I}_{z'}}(t) \to \infty$ as $t \to \infty$, contradicting the definition of $P_{\mc{I}_{z'}}(t)$ as a probability. Therefore we can conclude that $P_{\mc{I}_z} \to 1$ as $t \to \infty$ for all $z \in [0,1]$. 
\end{proof}

\section{Numerical Simulations of Long-Time Behavior} \label{sec:numerics}

In this section, we further investigate the long-time behavior of solutions to Equation \eqref{eq:multilevelPDEtworho} by exploring numerical simulations starting from uniform initial strategy distributions $\mu_t(dx) = dx$. All simulations described in this section are numerical solutions to upwind finite-volume for the game-theoretic examples introduced in Section \ref{sec:gametheory}. For these simulations, we use the group-level Fermi victory probability depending on the hyperbolic tangent of the difference in average payoffs between groups 
\begin{equation} \label{eq:Fermireminder}
\rho(x,y) = \frac{1}{2} \left[1 + \tanh \left(s \left\{G(x) - G(y) \right\} \right) \right],
\end{equation}
which means that the net group-level success in pairwise group-level competition for an $x$-cooperator group engaged in conflict with a $y$-cooperator group can be described by
\begin{equation}
\begin{aligned}
\rho(x,y) - \rho(y,x) &=  \left( \frac{1}{2} \left[1 + \tanh \left(s \left\{G(x) - G(y) \right\} \right) \right] \right) - \left(  \frac{1}{2} \left[1 + \tanh \left(s \left\{G(y) - G(x) \right\} \right) \right] \right) \\ &= \tanh\left( s \left[ G(x) - G(y) \right] \right)
\end{aligned}
\end{equation}
Using this formula and our discussion in Sections \ref{sec:PDsteadygroupsuccess} and \ref{sec:HDsteady}, we have conjectured that multilevel competition following the group-level Fermi rule can support steady-state cooperation if the relative strength $\lambda$ exceeds the threshold 
\begin{equation}
\lambda^*_{PD}(\theta) = \frac{\theta \pi(1)}{\rho(1,0) - \rho(0,1)} = \frac{\theta \pi(1)}{\tanh\left( s \left[G(1) - G(0) \right]\right)}
\end{equation}
for a generalized PD scenario, while we conjecture that a steady-state featuring cooperation above the within-group equilibrium level $x_{eq}$ can occur in a generalized HD scenario provided that $\lambda$ exceeds
\begin{equation}
\lambda^*_{HD}(\theta) = \frac{\theta \pi(1)}{\rho(1,x_{eq}) - \rho(x_{eq},1)} =  \frac{\theta \pi(1)}{\tanh\left( s \left[G(1) - G(x_{eq}) \right]\right)}.
\end{equation}

For our simulations, we will consider individual-level and group-level competition using replication rates based on the net individual-level advantage of defectors $\pi(x)$ and the average payoff of group members $G(x)$ given by
\begin{equation} \label{eq:pigame}
\pi(x) = - \left( \beta + \alpha x \right)
\end{equation}
and 
\begin{equation} \label{eq:Ggame}
G(x) = P + \gamma x + \alpha x^2
\end{equation}
as described in Section \ref{sec:gametheory}. We present numerical simulations for the group-level Fermi victory probability in the case of example games with payoff parameters corresponding to the PD game (Section \ref{sec:PDnumerics}), the HD game (Section \ref{sec:HDnumerics}), and the SH game (Section \ref{sec:SHnumerics}).
We describe the scheme used for our finite volume simulations in Section \ref{sec:numericalscheme} of the appendix, and we also present additional simulations of the multilevel dynamics for PD scenarios with different group-level update rules in Section \ref{sec:AdditionalPDSimulations}.

\subsection{PD Game}
\label{sec:PDnumerics}

We first consider numerical solutions for the multilevel dynamics of Equation \eqref{eq:multilevelPDEtworho} for the case of a PD game. For PD games, we can expect that cooperation will survive under the multilevel dynamics with pairwise group-level competition when the relative threshold value $\lambda^*_{PD}(1)$ from Equation \eqref{eq:lambdastarPDtheta} for H{\"o}lder exponent $\theta = 1$ near $x=1$. We can combine our expression for the threshold $\lambda^*_{PD}(1)$ with our expressions from Equations \eqref{eq:Fermireminder}, \eqref{eq:pigame}, and \eqref{eq:Ggame} to see that this can be written in terms our payoff parameters as
\begin{equation}
\lambda > \lambda^*_{PD}(1) := \frac{\pi(1)}{\rho(1,0) - \rho(0,1)} = \frac{- \left( \beta + \alpha \right)}{\tanh\left( s \left[ G(1) - G(0) \right] \right)} = \frac{-(\beta + \alpha)}{\tanh\left( s \left[\gamma + \alpha \right]\right)}.
\end{equation}
We will now consider numerical solutions of our PDE model, studying how the payoff parameters and strength of group-level selection $\lambda$ support the long-time dynamics of multilevel selection with pairwise group conflict competition following the Fermi group-level update rule.

In Figure \ref{fig:PDtrajectories}, we illustrate time-dependent solutions to our PDE model for multilevel selection starting from a uniform initial condition. In Figure \ref{fig:PDtrajectories}(left), we present a case with weak between-group competition characterized by $\lambda < \lambda^*_{PD}(1)$, showing that all groups in the population concentrate upon a the all-defector composition as time progresses in the simulation. In Figure \ref{fig:PDtrajectories} (right), we provide an example of a numerical solution for the case of $\lambda > \lambda^*_{PD}(1)$, showing that the solutions appear to tend towards a steady-state density supporting groups with all possible levels of cooperation. We will now explore the behavior of such steady state densities, examining how the long-time support for cooperation depends on the strength of group-level selection and the shape of the average payoff $G(x)$ for group members. 

\begin{figure}[!ht]
    \centering
    \includegraphics[width = 0.45\textwidth]{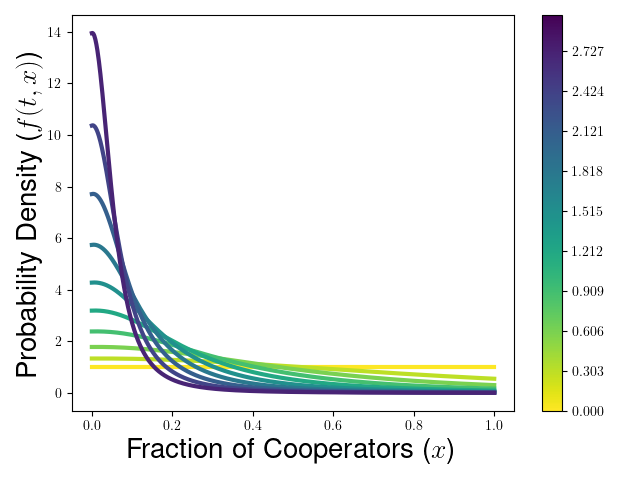}
       \includegraphics[width = 0.45\textwidth]{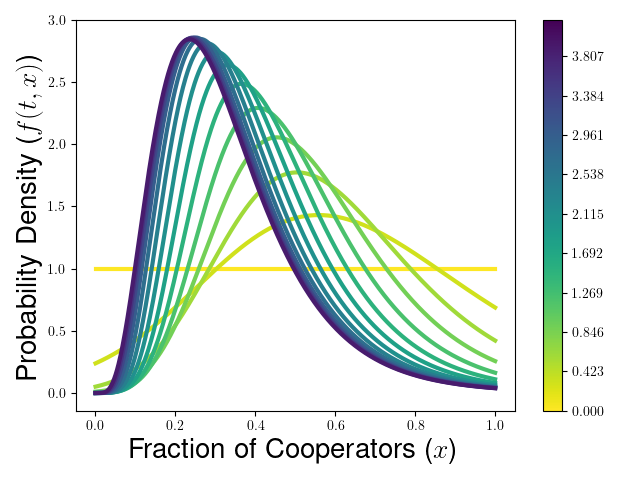}
    \caption{Snapshots in time to solutions for multilevel dynamics for PD game with average group payoff maximized by 75 percent cooperation for relative strength of group-level competition $\lambda = 0.01$ (left) and $\lambda = 14$ (right). The color of the densities corresponds to the time at which the density was achieved as a numerical solution to the PDE, with yellow describing early times and later times plotted in blue (see the color bars for actual time values $t$). The simulations were run for 1000 time-steps (left) and 1400 time-steps (right) with a step-size of $\Delta t = 0.003$, starting from a uniform initial density of group compositions. The game-theoretic parameters were fixed to $\gamma = 1.5$, $\alpha = -1$, $\beta = -1$, and $P = 1$, and the group-level competition followed the Fermi victory probability with payoff sensitivity parameter $s = 1$.}
    \label{fig:PDtrajectories}
\end{figure}

In Figure \ref{fig:PDsteadyghostcompare}, we provide illustrations of the steady state densities achieved by the multilevel PD dynamics for different values of the relative strength of between-group selection $\lambda$. Starting with the same initial condition, we see that the steady state densities feature increasing levels of cooperation as $\lambda$ increases both in the case in which average payoff of group members $G(x)$ is maximized at 100 percent cooperation (Figure \ref{fig:PDsteadyghostcompare}, left) and in a case in which average payoff of group members $G(x)$ is maximized by 75 percent cooperators (\ref{fig:PDsteadyghostcompare}, right). For the case in which full-cooperation is collectively optimal, we see that the steady state densities are capable of supporting a substantial portion of groups with compositions close to the optimal level of cooperation when between-group competition is sufficiently strong. For the case in which an intermediate level of cooperation maximizes average payoff, we see that the steady state densities appear to concentrate upon a composition of 50 percent cooperators, which is notably less than the optimal composition of 75 percent cooperators that maximizes average payoff and the probability of group-level victory. This discrepancy in the steady state behavior between the case in which average payoff is maximized by the all-cooperator group and by a group featuring 75-percent cooperation is analogous to the behavior seen for the long-time behavior for the case of multilevel selection following a two-level replicator equation \cite[Figure 1]{cooney2022long}, suggesting that multilevel dynamics with pairwise group-level competition also feature a long shadow cast by the individual-level incentive to defect.

\begin{figure}[!ht]
    \centering
    \includegraphics[width = 0.45\textwidth]{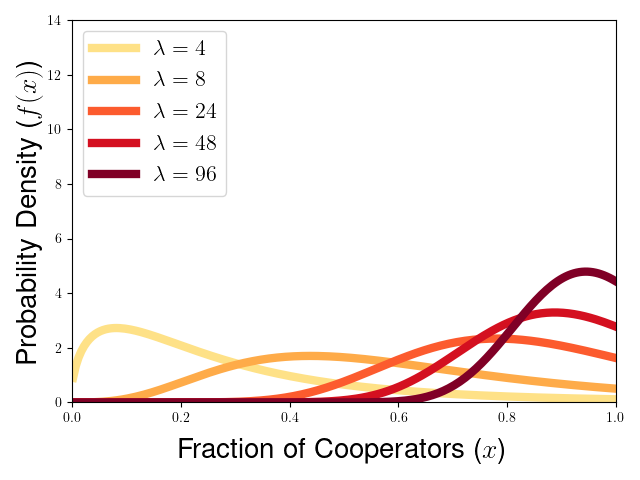}
       \includegraphics[width = 0.45\textwidth]{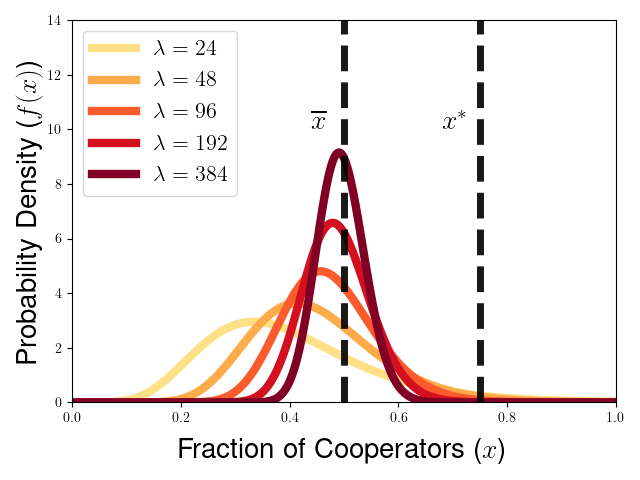}
    \caption{Comparison of densities achieved starting from uniform initial distribution after 9,600 time steps with step-size $\Delta t = 0.003$ for different values of between-group selection strength $\lambda$ for PD games in which average payoff is maximized by the all-cooperator group $x^* = 1$ (left) or is maximized by a group with 75 percent cooperators ($x^* = \frac{3}{4}$, right). In the right panel, the vertical dashed lines corresponds to the fraction of cooperation $x^* = \frac{3}{4}$ that maximizes the average payoff of group members for the PD game and the fraction of cooperation $\ol{x} = 0.5$ achieving the same average payoff as the payoff in an all-cooperator group (i.e. satisfying $G(\ol{x}) = G(1)$). For these simulations, the game-theoretic parameters considered were $\gamma = 2$ (left) and $\gamma = 1.5$ (right), with $\alpha = \beta = -1$ for both panels, and group-level conflict followed the Fermi update rule with sensitivity parameter $s = 1$.}
    \label{fig:PDsteadyghostcompare}
\end{figure}

We can also explore the long-time behavior of our numerical simulations as a function of the strength of between-group selection $\lambda$. In Figure \ref{fig:PDgroupsuccess}, we explore the impact of the relative strength $\lambda$ of between-group selection on the collective outcomes achieved for a PD game in which average payoff is maximized by groups with 75 percent cooperation. We plot the average group-level success at steady state $\int_0^1 x \rho(x,1) f(x,t) dx$ (Figure \ref{fig:PDgroupsuccess}) achieved from the states achieved after 9600 time steps of our finite volume simulations, and we compare this to the conjectured collective success we would expect in the long-time outcome based on our exploration in Section \ref{sec:PDsteadygroupsuccess}. In particular, we can use our conjecture that the population converges to a delta-function at the all-defector equilibrium if $\lambda \leq \lambda^*_{PD}(1)$ and that the population converges to a steady-state density $f^{\lambda}_{1}(x)$ with collective success $\int_0^1 \rho(1,x) f^{\lambda}_{1}(x) dx$ given by Equation \eqref{eq:rho1xsteadystate} to conjecture that the long-time collective success against the all-cooperator group is given by 
\begin{equation}\label{eq:PDlongtimesuccessconjecture}
\begin{aligned}
  \lim_{t \to \infty} \int_0^1 \rho(x,1) f(t,x) dx &= \left\{
    \begin{array}{cr}
      \rho(0,1) & : \lambda < \lambda^*_{PD}(1) \vspace{3mm} \\
      \ds\frac{1}{2} - \ds\frac{\pi(1)}{2 \lambda} & : \lambda \geq \lambda^*_{PD}(1)
    \end{array}
  \right.\\ \\
  &=  \left\{\begin{array}{cr}
      \frac{1}{2} + \frac{1}{2} \tanh\left( s \left[G(0) - G(1) \right] \right) & : \lambda < \lambda^*_{PD}(1) \vspace{3mm} \\
      \ds\frac{1}{2} + \ds\frac{\alpha + \beta}{2 \lambda} & : \lambda \geq \lambda^*_{PD}(1)
    \end{array}
    \right.
  .
  \end{aligned}
\end{equation}
From Figure \ref{fig:PDgroupsuccess}, we see that there is good agreement between the conjectured formula from Equation \eqref{eq:PDlongtimesuccessconjecture} for the long-time collective outcome and the numerically computed collective outcome after 9,600 time-steps. In particular, we see that the numerical solution appears to suggest that $\lambda^*_{PD}(1)$ is a critical value at which the long-time population can outperform the success achieved by an all-defector group, while the long-time population's average victory probability against the all-cooperator group approaches $\frac{1}{2}$ in the limit of large $\lambda$.

\begin{figure}[!ht]
    \centering
    \includegraphics[width = 0.6\textwidth]{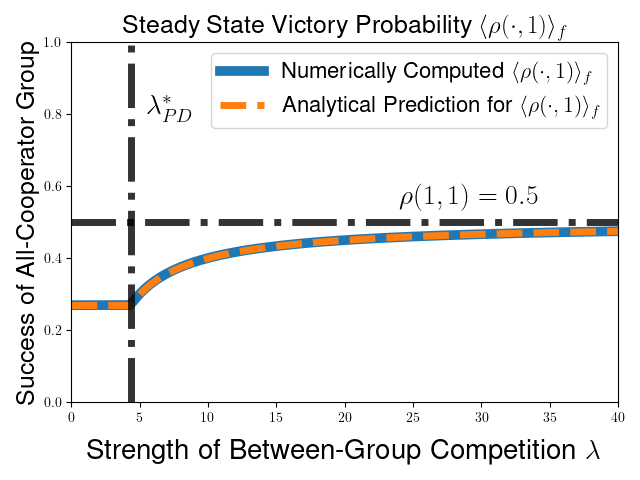}
    \caption{Comparison of the group-level success measured from numerical simulations and the predicted value for a density steady state, plotted as a function of $\lambda$ for a PD game in which average group payoff is maximized by a group fraction $x^* = \frac{3}{4}$ cooperators. We plot both the numerical computation $\int_0^1 \rho(y,1) f(t,y) dt$ of the population against the all-cooperator group (solid blue curve) with the predicted formula from Equation \eqref{eq:rhoy1steadystate} of the collective success $\int_0^1 \rho(y,1) f(y) dy$ for a bounded density steady state (dashes orange curve). The black vertical dash-dotted line gives the predicted threshold selection strength $\lambda_{PD}(1)$ from Equation \eqref{eq:lambdastarPDtheta} for the case of H{\"o}lder exponent $\theta = 1$. The payoff parameters were $\gamma = 1.5$, $\alpha = \beta = -1$, and $P = 1$, and group-level competition followed the Fermi update rule with sensitivity parameter $s = 1$.}
    \label{fig:PDgroupsuccess}
\end{figure}

In Figure \ref{fig:PDaveragecooperators}, we plot the average level of cooperation $\int_0^1 x f(t,x) dx$ achieved after 9,600 time-steps for various strengths of the between-group competition $\lambda$ for the case of a PD game in which average payoff is maximized by the intermediate level of cooperation $x^* = 0.75$ and such that the composition $\ol{x} = 0.5$ satisfies $G(\ol{x}) = G(1)$, featuring the same average payoff as the all-cooperator group. We see that the average level of cooperation starts out at $0$ for $\lambda$ sufficiently close to $0$, and that the level of cooperation first takes positive values when $\lambda$ increases past the conjectured threshold selection strength $\lambda^*_{PD}(1)$. The average level of cooperation then increases with $\lambda$, approaching the value of $\ol{x} < x^*$ for the case of strong relative levels of group-level competition. This provides further illustration of a shadow cast by lower-level selection for our PDE model with pairwise group-level competition, as the long-time behavior produces less cooperation than is optimal for the group, and the population concentrates upon the level of cooperation that has the same payoff as the all-cooperator group (with corresponding group-level victory probability satisfying $\rho(x^*,1) = \frac{1}{2}$).  

\begin{figure}[!ht]
    \centering
    \includegraphics[width = 0.6\textwidth]{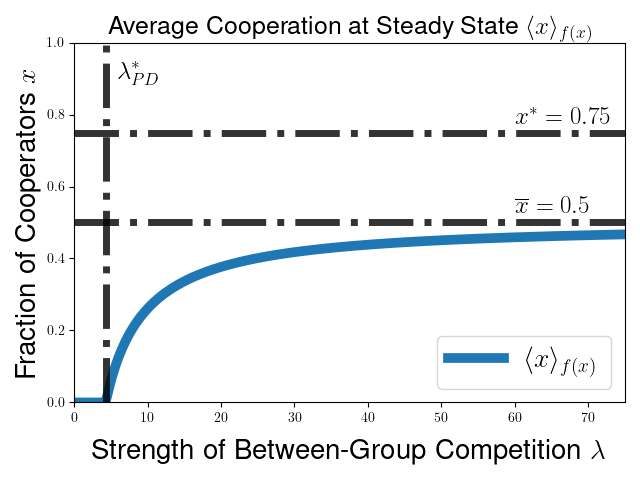}
    \caption{Plot of the average fraction of cooperators $\int_0^1 x f(t,x) dx$ for the numerical simulation after 9,600 time-steps of step-size 0.01 (blue curve). The vertical dash-dotted line corresponds to the predicted threshold strength of between-group selection $\lambda^*_{PD}(1)$ from Equation \eqref{eq:lambdastarPDtheta}. The upper horizontal dash-dotted line corresponds to the level of cooperation $x^* = 0.75$ that maximizes average payoff of group members, while the lower horizontal dash-dotted line corresponds to the level of cooperation $\overline{x} = 0.5$ at which the average payoff of group members is equal to the average payoff $G(1)$ of the all-cooperator group. The game-theoretic parameters were fixed at $\gamma = 1.5$, $\beta = -1$, $\alpha = - 1$, and $P = 1$ for all simulations, and the sensitivity parameter for the group-level Fermi update rule was set to $s = 1$. }
    \label{fig:PDaveragecooperators}
\end{figure}

\subsection{HD Game}
\label{sec:HDnumerics}

Next, we present the results of similar numerical simulations of the PDE model when the within-group and between-group competition depends on payoffs from an HD game. From our discussion from Section \ref{sec:HDsteady}, we expect that the population will converge to a delta-function at the within-group Hawk-Dove equilibrium $x_{eq} = \tfrac{\beta}{-\alpha}$ for low strength $\lambda$ of group-level competition, while we anticipate that our simulations multilevel dynamics can support a long-time steady state density featuring increased levels of cooperation when $\lambda$ exceeds the conjectured threshold quantity $\lambda^*_{HD}(1)$. 
In particular, we can use our expressions for $\lambda^*_{HD}(1)$, $G(x)$, $\pi(x)$, and $x_{eq}$ to see that we expect a density steady-state featuring groups with fractions of cooperation exceeding $x_{eq}$ provided that
\begin{equation}
\lambda > \lambda^*_{HD}(1) := \frac{\pi(1)}{2 \rho(1,0) - 1} = \frac{-\left(\beta + \alpha\right)}{ \tanh\left(s \left[ G(1) - G(x_{eq}) \right] \right)} = \frac{-\left(\beta + \alpha\right)}{\tanh\left( \frac{s \left( \gamma + \alpha - \beta\right) \left(  \beta + \alpha \right)}{\alpha} \right)}.
\end{equation}

We now present numerical trajectories for solutions to the multilevel dynamics of Equation \eqref{eq:multilevelPDEfirstform} for an HD game starting from an initial uniform distribution of group compositions and two different relative strengths of between-group selection $\lambda$. For the case of $\lambda = 0.1$, we see that the population concentrates around the equilibrium value $x_{eq} = \frac{1}{2}$ for the within-group dynamics for the HD game (Figure \ref{fig:HDtrajectories}, left), while we see that the numerical solutions appear to approach a steady state density featuring levels of cooperation exceeding $x_{eq}$ for the case of $\lambda = 15$ (Figure \ref{fig:HDtrajectories}, right). For the case of convergence to a steady state supporting additional cooperation, we see that there appear to be no groups featuring less than a fraction $x_{eq}$ of cooperators after 9,600 time steps, but the groups present at this time feature a mix of different fractions of cooperation ranging between $x = x_{eq}$ and $x = 1$. This behavior is consistent with the dynamical behavioral seen for the HD case of the previously studied PDE model of multilevel selection with frequency-independent group-level competition (see \cite[Proposition 4.2]{cooney2019replicator} and \cite[Theorem 6]{cooney2022long}). 
\begin{figure}[!ht]
    \centering
    \includegraphics[width = 0.45\textwidth]{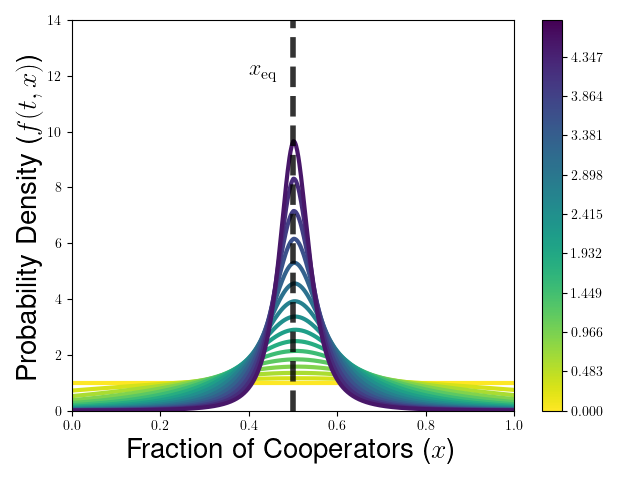}
       \includegraphics[width = 0.45\textwidth]{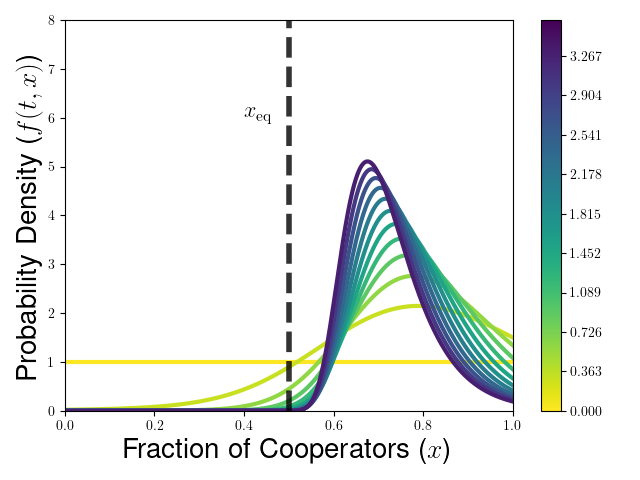}
    \caption{Snapshots in time to solutions for multilevel dynamics for HD game starting from a uniform initial distribution of group compositions (with the initial condition plotted in yellow). For the HD game under consideration, the equilibrium to the within-group dynamics was $x_{eq} = 0.5$ and the fraction of cooperation maximizing the average payoff of group members was given by $x^* = 0.875$.  Simulations were run for $1600$ (left) and $1200$ (right) time steps with step size of $\Delta t = 0.003$ and the colors of the densities plotted describe times as displayed in the colorbar. The game-theoretic parameters were set to $\gamma = 3.5$, $\alpha = -2$, $\beta = 1$, and $P = 1$ for each panel, and the group-level Fermi update rule was considered with sensitivity $s = 1$, and the relative strength of between-group selection was $\lambda = 0.1$ (left) and $\lambda = 15$ (right) for the two panels.}
    \label{fig:HDtrajectories}
\end{figure}

In Figure \ref{fig:HDsteadyghostcompare}, we plot the solutions for multilevel HD dynamics obtained after many time-steps for different values of the between-group selection strength $\lambda$. We compare a case in which the average payoff of group members $G(x)$ is maximized by the all-cooperator group $x^* = 1$ (Figure \ref{fig:HDsteadyghostcompare}) and a case in which average payoff $G(x)$ is maximized by a group featuring 87.5\% cooperation. For the case in which full-cooperation is collectively optimal, we see that the long-time behavior features high fractions of cooperation and hasfull-cooperation as its modal outcome even for between-group selection strengths of $\lambda = 16$. For the case of the intermediate collective outcome, we see that the long-time distribution of groups concentrates upon a level of cooperation lower than the fraction of cooperation $x^* = 0.875$ that optimizes average payoff $G(x)$, even for a relatively large strength of between group-selection $\lambda = 384$. This suggests that the multilevel dynamics with pairwise group-level competition also features a shadow of lower-level selection, similar to the behavior seen for the case of the Hawk-Dove game under the two-level replicator equation model \cite{cooney2019replicator,cooney2020analysis,cooney2022pde}. 

\begin{figure}[!ht]
    \centering
    \includegraphics[width = 0.45\textwidth]{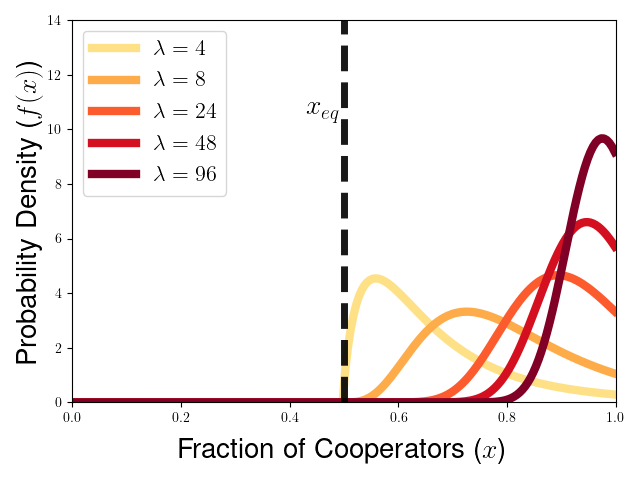}
       \includegraphics[width = 0.45\textwidth]{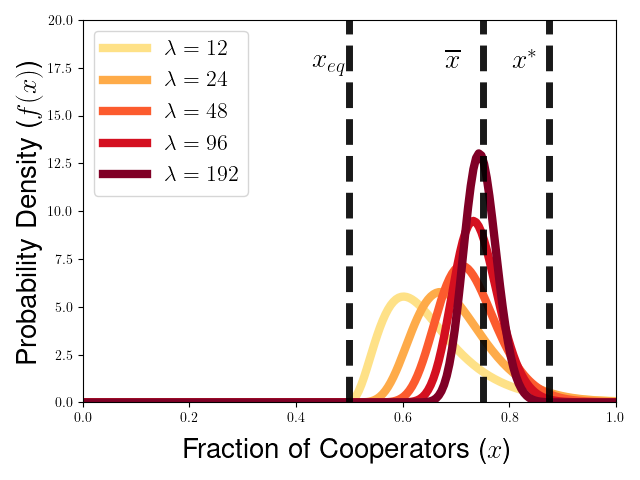}
    \caption{Comparison of densities achieved after 9,600 time steps of step-size $\Delta t = 0.003$ starting from uniform initial distribution for different values of between-group selection strength $\lambda$ for HD games in which average payoff is maximized by the all-cooperator group $x^* = 1$ (left) or is maximized by a group with 87.5 percent cooperators ($x^* = \frac{7}{8}$, right). The vertical dashed lines correspond to the equilibrium fraction of cooperators $x_{eq} = 0.5$ under individual-level selection alone, the fraction of cooperators $x^*$ maximizing average payoff, and the fraction of cooperators $\ol{x}$ with the same average payoff as the all-cooperator groups (i.e. satisfying $G(\ol{x}) = G(1)$). The game-theoretic parameters were given by $\gamma = 4$ (left) and $\gamma = 3.5$, $\alpha = -2$, $\beta = 1$, and $P = 1$, and the simulations were run with the group-level Fermi update rule with sensitivity parameter $s = 1$.}
    \label{fig:HDsteadyghostcompare}
\end{figure}

We can then compare our analytical predictions for the average group-level replication rate at steady state and the threshold selection strength $\lambda^*_{HD}$ to the values of these quantities computed from numerical solutions of the PDE model starting from a uniform initial distribution of strategies. We show the average group-level success $\int_0^1 \rho(y,1) f(t,y) dy$ relative to the all-cooperator group (Figure \ref{fig:HDquantitieslambda}) and the average fraction of cooperation $\int_0^1 y f(t,y) dy$ (Figure \ref{fig:HDaveragecooperators}) for the numerical solutions to the multilevel dynamics, which are plotted as a function of $\lambda$.
We compare the numerically computed values of the collective success against the all-cooperator group with the following conjectured formula obtained from our discussion in Section \ref{sec:HDsteady}:
\begin{equation}\label{eq:HDlongtimesuccessconjecture}
\begin{aligned}
  \lim_{t \to \infty} \int_0^1 \rho(x,1) f(t,x) dx &= \left\{
    \begin{array}{cr}
      \rho(x_{eq},1) & : \lambda < \lambda^*_{HD}(1) \vspace{3mm} \\
      \ds\frac{1}{2} -  \ds\frac{\pi(1)}{2 \lambda} & : \lambda \geq \lambda^*_{HD}(1)
    \end{array}
  \right.\\ \\
  &=  \left\{\begin{array}{cr}
     \frac{1}{2} + \frac{1}{2} \tanh\left( s \left[G\left( \frac{\beta}{-\alpha} - G(1)\right) \right] \right) & : \lambda < \lambda^*_{HD}(1) \vspace{3mm} \\
      \ds\frac{1}{2} + \ds\frac{\beta + \alpha}{2 \lambda} & : \lambda \geq \lambda^*_{HD}(1)
    \end{array}
    \right.
  .
  \end{aligned}
\end{equation}
We see that there is good agreement with the numerically calculated value of the average success of the population $\int_0^1 \rho(t,x)dx$ after 9,600 time-steps and our analytical prediction for the collective success that is necessary to achieve a bounded steady state for the multilevel dynamics. We also see that the average fraction of cooperators starts out a the equilibrium level $x_{eq}$ for $\lambda$ sufficiently small, before increasing once $\lambda$ exceeds our conjectured threshold value $\lambda^*_{HD}(1)$. This average fraction of cooperation approaches the values $\overline{x} = 0.75$ as $\lambda$ becomes large, which is less than the fraction of cooperation $x^*$ that maximizes average payoff for the HD game under consideration. The average level of cooperation instead converges to the level of cooperation that achieves the same collective payoff $G(\ol{x}) = G(1)$ of the all-cooperator group, providing more numerical support for the idea that our PDE model with pairwise group-level competition appears to feature a shadow of lower-level selection for the multilevel HD scenario.

\begin{figure}[!ht]
    \centering
    \includegraphics[width = 0.6\textwidth]{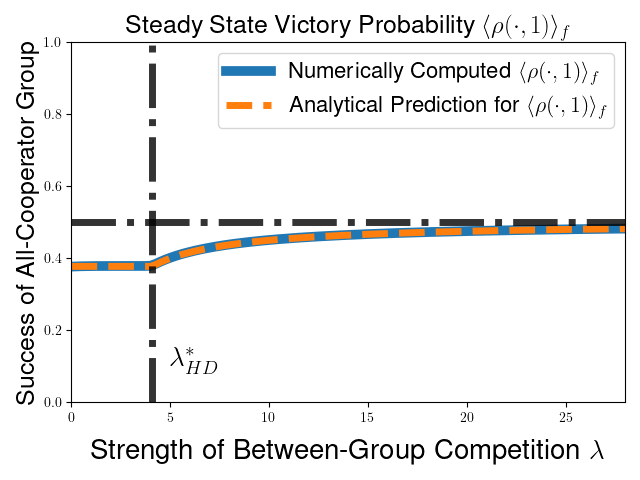}
    \caption{Comparison of the group-level success $\int_0^1 \rho(y,1) f(t,y) dt$ of the population against the all-cooperator group (solid blue curve) with the predicted formula from Equation \eqref{eq:rhoy1steadystate} of the collective success $\int_0^1 \rho(y,1) f(y) dy$ in conflict with an all-cooperator group for a bounded density steady state (dashes orange curve). The black vertical dash-dotted line gives the predicted threshold selection strength $\lambda^*_{HD}(1)$ from Equation \eqref{eq:lambdastarHDtheta} for the case of H{\"o}lder exponent $\theta = 1$. The game-theoretic parameters were set to $\gamma = 3.5$, $\alpha = -2$ $\beta = 1$, and $P = 1$ for each panel, and the group-level Fermi update rule was considered with sensitivity $s = 1$.}
    \label{fig:HDquantitieslambda}
\end{figure}

\begin{figure}[!ht]
    \centering
     \includegraphics[width = 0.6\textwidth]{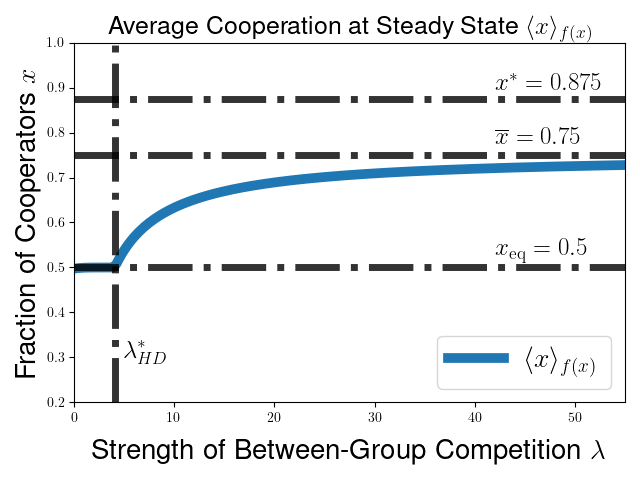}
    \caption{Plot of the average fraction of cooperators $\int_0^1 x f(t,x) dx$ for the numerical simulation after 9,600 time-steps of step-size 0.01 (blue curve), presented as a function of the relative strength of group-level competition $\lambda$. The vertical dash-dotted line corresponds to the predicted threshold strength of between-group selection $\lambda^*_{PD}(1)$ from Equation \eqref{eq:lambdastarPDtheta}. The upper horizontal dash-dotted line corresponds to the level of cooperation $x^* = 0.875$ that maximizes average payoff of group members, the middle horizontal dash-dotted line corresponds to the level of cooperation $\overline{x} = 0.75$ at which the average payoff of group members is equal to the average payoff $G(1)$ of the all-cooperator group, and the lower horizontal dash-dotted line corresponds to the equilibrium point $x_{eq} = 0.5$ for the within-group dynamics. The payoff parameters were given by $\gamma = 3.5$, $\beta = 1$, and $\alpha = -2$, and group-level conflict used the Fermi update rule with sensitivity parameter $s = 1$.  }
    \label{fig:HDaveragecooperators}
\end{figure}

\subsection{SH Game}
\label{sec:SHnumerics}

We can also run finite volume simulations with replication rates arising from the payoffs of the SH game. In Figure \ref{fig:SHtrajectories}, we illustrate numerically computed solutions to Equation \eqref{eq:PDEmeasure} starting from a uniform initial distribution of strategies both in the absence of between-group competition when $\lambda = 0$ (Figure \ref{fig:SHtrajectories}, left) and in a case with a positive strength of between-group competition with $\lambda = 1$ (Figure \ref{fig:SHtrajectories}). For the case of no between-group competition, we see from the left panel of Figure \ref{fig:SHtrajectories} that each group moves towards the closest stable equilibrium of the within-group dynamics, producing a long-time behavior with half of the groups each ending up at the all-cooperator and all-defector composition. For the case of positive strength of between-group competition $\lambda > 0$, we see that portions of the population initially move towards each of the two stable within-group equilibria, before concentrating upon the collectively beneficial all-cooperator equilibrium after several thousand time-steps. This provides a numerical illustration of the result from Proposition \ref{prop:SHdelta} predicting that the population will concentrate at the all-cooperator equilibrium in the presence of group-level competition and any groups located above the within-group equilibrium $x_{eq}$, and is reminiscent of the tug-of-war between stable within-group equilibria suggested in the work of Richerson and Boyd on group selection between alternative stable evolutionarily stable strategies \cite{boyd1990group}.

\begin{figure}[!ht]
    \centering
    \includegraphics[width = 0.45\textwidth]{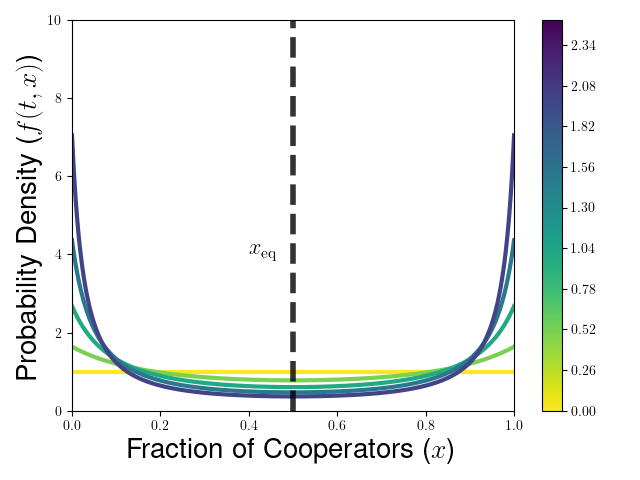}
       \includegraphics[width = 0.45\textwidth]{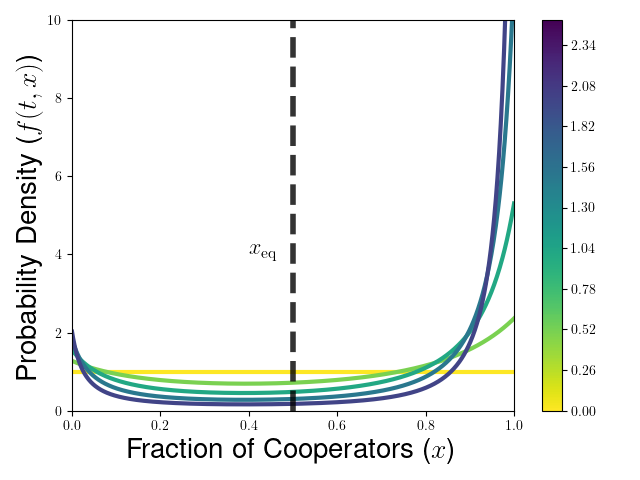}
    \caption{Snapshots in time to solutions for multilevel dynamics for the SH game for cases for between-group selection strength $\lambda = 0$ (left) and $\lambda = 1$ (right), starting with a uniform distribution of strategies in the group-structured population plotted in yellow. The color of each density corresponds to the time displayed in the color bar. The vertical dashed line represents the within-group equilibrium $x_{eq} = 0.5$. The simulations were run for 250 time-steps with step-size of $\Delta t = 0.01$, the payoff parameters were set to $gamma = 0$, $\alpha = 2$, $\beta = 0$, and $P = 2$ for each panel, and group-level competition used the Fermi update rule with sensitivity parameter $s = 1$. }
    \label{fig:SHtrajectories}
\end{figure}

\section{Discussion}
\label{sec:discussion}

In this paper, we explored dynamical and steady-state behavior of a PDE model of multilevel selection as a two-level birth-death process with pairwise between-group competition. By establishing well-posedness and obtaining an implicit representation formula for measure-valued solutions for our model, we provide a first step in the process of characterizing the dynamics induced by multilevel dynamics with frequency-dependent group-level conflict. We have applied the implicit representation formula to demonstrate the preservation of the tail-behavior of solutions of the multilevel dynamics, and to demonstrate that the population can converge to a delta-function at the all-defector equilibrium for multilevel PD dynamics when individual-level selection dominates pairwise group-level competition. Furthermore, our characterization of possible steady state densities for multilevel PD and HD scenarios provides us with conjectured analytical expressions for steady-state collective success and a threshold selection strength to sustain long-time cooperation. Our numerical simulations of our PDE model provide support for these conjectured analytical expressions, and the combination of analytical and numerical results provide a preliminary analysis motivating future exploration for modeling multilevel selection with pairwise group-level competition. 

We used both numerical simulations and a heuristic calculation for the group-level victory probability for the all-cooperator group to suggest that the collective outcome is limited by the collective success of the all-cooperator group in our model of multilevel selection. This extends the behavior called the ``shadow of lower-level selection'' seen in previous PDE models of multilevel selection to our model with pairwise between-group selection, indicating that this behavior may hold for a broader class of models with different forms of competition taking place between groups. The form of the conjectured threshold level of selection strength $\lambda^*_{PD}(\theta)$ required to sustain steady-state cooperation suggests that the ability to sustain long-time cooperation depends on a tug-of-war between the individual incentive to defect $\pi(1) = \pi_D(1) - \pi_C(1)$ in an all-cooperator group and the collective incentive $\rho(1,0) - \rho(0,1)$ to achieve full-cooperation over full-defection. In the multilevel HD scenario, our conjectured formula for the threshold selection strength $\lambda^*_{HD}$ similarly reflected a conflict between the individual incentive to defect $\pi(1)$ and the collective advantage $\rho(1,x_{eq}) - \rho(x_{eq},1)$ of the all-cooperator group over a group at the within-group HD equilibrium $x_{eq}$ under pairwise group conflict.  As noted in Remark \ref{rem:PDadditive}, these threshold formulas directly generalize the formulas obtained from the case of additively separable group-level victory probabilities which produce multilevel dynamics following two-level replicator equations, suggesting that this tug-of-war between individual and group incentives at within-group equilibrium points may potentially be a common feature of deterministic models of multilevel selection. 

In  addition, we have shown in the case of a generalized Stag-Hunt scenario that multilevel selection with pairwise between-group competition will result in convergence to a population featuring only cooperators provided that there is any positive strength of between-group selection. This result provides an analysis of the behavior suggested by Richerson and Boyd for multilevel selection for scenarios in which within-group dynamics feature alternative stable strategies \cite{boyd1990group}. While Richerson and Boyd focus on the case of a separation of time-scales between the levels of selection (in which within-group dynamics quickly concentrate upon one of two stable strategies and then group-level selection operates upon nearly homogeneous groups on a slower time-scale), our result on the Stag-Hunt game shows that group-level competition will promote coordination on the cooperative equilibrium even when within-group and between-group competition occur on an overlapping time-scale. This suggests that pairwise group-level competition is relatively robust in its ability to promote the selection of the socially optimal within-group equilibria in the case of bistable individual-level dynamics, which may indicate that this form of group-level competition is particular effective when working in concert with mechanisms like reciprocal altruism \cite{trivers1971evolution,nowak1993strategy}, altruistic punishment \cite{boyd2003evolution,bowles2012punishment}, and social norms of indirect reciprocity \cite{boyd1989evolution,nowak1998dynamics,brandt2004logic,ohtsuki2006leading,pacheco2006stern} for which punishment of defectors can help to stabilize the all-cooperator equilibrium under individual-level dynamics. 

By introducing a PDE model of multilevel selection with frequency-dependent group-level competition, we have incorporated a greater degree of nonlinearity in the term describing group-level replication. Unlike the case of prior PDE models of multilevel selection, it does not appear to characterize the long-time behavior of solutions by using a strategy of obtaining explicit solutions to a simplified linear PDE by considering a growing population of groups by ignoring the replacement of groups due to group-level conflict \cite{cooney2022long}. By characterizing necessary conditions for the existence of density steady states supporting cooperation, we have a first step in understanding the extent to which pairwise group-level competition can allow for the support of the evolution of cooperation in multilevel selection when individual and group replication depends on the payoffs of Prisoner Dilemma and Hawk-Dove games. Our conjectured threshold selection strengths $\lambda^*_{PD}$ and $\lambda^*_{HD}$ for the persistence of density steady state generalizes the case of additively separable group-level victory probabilities that can be studied using a two-level replicator equation, and provides us with a basis for understanding of how the assumptions about group-level competition can impact the long-time support for collectively beneficial outcomes.

In addition, by studying the behavior of our PDE model of multilevel selection for sample two-player, two-strategy games, we have set baseline expectations for the role that pairwise group-level competition can play when working in concert with other within-group mechanisms that can help to support cooperative behavior or otherwise alter the strategies that can be supported by within-group selection alone. Synergistic effects between within-group mechanisms and group-level competition have already been explored in the context of two-level replicator equations \cite{cooney2022assortment,cooney2023evolutionary}, and recent experimental and simulation work has demonstrated that group selection and direct reciprocity can together promote cooperation for game-theoretic scenarios in which neither mechanism could sustain cooperation on its own \cite{efferson2024super}. Exploration of stochastic models have already explored how pairwise group-level conflicts can help to promote cooperative behavior in concert with costly punishment of defection \cite{boyd1990group} or social institutions that help to mediate individual-level advantages to defect \cite{bowles2003co}, so a natural next step is to study the extent to which such synergies between group-level conflict and within-group mechanisms can be seen in PDE models of multilevel selection. A first step in this direction has been explored in the context of altruistic punishment, in which it was shown that numerical simulations of two-level replicator equations and our PDE model with pairwise group competition displayed qualitative differences in how long-time average payoff depended on the cost of punishing defectors \cite{cooney2024exploring}.

While we have extended the existing modeling framework for PDE models of multilevel selection to consider pairwise group-level competition, our analysis in this paper has been limited to the assumption that victory in pairwise group-level conflict has been a function of the average payoffs achieved in the two competing groups. However, it is also natural to explore multilevel selection in which group-level competition depends on quantities other than the average individual-level success of group members \cite{simpson2024levels}, and the notion of pairwise conflict between groups could potentially motivate group-level victory probabilities that are more general functions of the strategic compositions of the competing groups. Other researchers have used alternative frameworks for studying multilevel selection with group-level competition in which individual strategies consist of a relative emphasis on investment in within-group cooperation and participation in group-level conflict \cite{henriques2019acculturation,choi2007coevolution,reeve2007emergence,boza2010beneficial,tverskoi2021dynamics}, and the literature on hierachical social dilemmas and multilevel selection with group-level game-theoretic interactions provides further motivation for considering how different forms of group-level frequency dependence can shape cooperative behavior via multilevel selection \cite{simon2010dynamical,fujimoto2017hierarchical,simon2024evolutionary}. This related work suggests that it may be helpful to consider versions of our PDE model with individual-level net replication rates $\pi(x)$ and group-level victory probabilities $\rho(x,y)$ motivated from a more general class of individual and group incentives for cooperation, as well as  to consider generalizations of our PDE model to incorporate group-level replication rates depending on the payoffs of games played between groups.

 While our model of pairwise between-group competition incorporates a feature of some existing simulation models of multilevel selection with frequency-dependent group-level competition, there are still many simplifying assumptions that we have made in order to formulate and analyze our PDE model of multilevel selection. In particular, we have assumed that the group size and number of groups remains constant, and we assume that group-level replication events consist of the winning group in a pairwise competition producing an exact copy of itself to replace the losing group in the conflict. Recent work on other  incorporated group-level competition based on differences in fractions of cooperative individuals has been used in PDE models of multilevel selection with non-constant group size, density-dependent within-group dynamics, and group-level fission events \cite{simon2010dynamical,markvoort2014computer,simon2016group,simon2024evolutionary,lerch2024flexible}, so a natural follow-up question is to explore whether behaviors observed in this paper, such as the shadow of lower-level selection, would also occur if we included more detailed group-level events in our PDE model. In particular, existing work on PDE models of multilevel selection featuring more realistic group-level events typically considers within-group games or group-level competition featuring all-cooperation as the collectively optimal strategy distribution, so it would be interesting to study such models with pairwise group-level competition based on more general classes of games that most favor intermediate levels of cooperation at the individual or group levels.

Finally, there also many future directions for further study of the PDE model we have explored in this paper. In particular, our derivation of conjectured formulas for the collective success at steady state and the threshold selection strength were calculated based on necessary conditions for the existence of steady state densities with given H{\"o}lder exponent near the all-cooperator equilibrium, but we have not yet established the existence of such steady states. While numerical simulations appear to suggest that the population can converge to a steady distribution of group compositions whose collective success matches our conjectured analytical expression, our understanding of the multilevel dynamics with pairwise group-level competition would be improved if we could demonstrate existence of steady state densities and prove convergence to a steady state density in the way that has been established for two-level replicator equations \cite{luo2017scaling,cooney2019replicator,cooney2022long}. Further efforts to demonstrate conditions for extinction or persistence of cooperation depending on the initial condition and strength of group-level competition would allow a greater sense of the tug-of-war between the individual-level incentive to defect and the collective incentive to achieve group-level victory. With these goals for future work, we see that there are still many mathematical questions to explore related to more general PDE models of multilevel selection featuring frequency-dependent competition at the group level.

\renewcommand{\abstractname}{Acknowledgments}
\begin{abstract} 
DBC was supported as a Simons Postdoctoral Fellow in Mathematical Biology through a Math + X Grant awarded to the University of Pennsylvania.  The authors would like to thank Yoichiro Mori,  Joshua Plotkin, and Erol Akçay for many helpful discussions.

\end{abstract}

\renewcommand{\abstractname}{Statement on Code Availability}
\begin{abstract} 
\sloppy{All code and simulation outputs used to generate figures are archived on Github (\href{https://github.com/dbcooney/Multilevel-Pairwise-Group-Paper-Code}{https://github.com/dbcooney/Multilevel-Pairwise-Group-Paper-Code}) and licensed for reuse, with appropriate
attribution/citation, under a BSD 3-Clause Revised License.}
\end{abstract}

\bibliographystyle{unsrt}
\bibliography{references}

\appendix
\appendixpage
\addtocontents{toc}{\protect\setcounter{tocdepth}{1}}

In the appendix, we provide additional proofs of results stated in the main text and additional information on the numerical simulations discussed in the main text. In Section \ref{sec:steadytheta}, we provide a derivation of our conjectured threshold selection strength to sustain a density steady state featuring cooperation. In Section \ref{sec:numericsappendix}, we present additional information about our finite volume simulations for our PDE model of multilevel selection, deriving the finite volume approximation for the effect of pairwise group-level competition (Section \ref{sec:numericalscheme}) and presenting additional numerical simulations of the multilevel dynamics for PD scenarios with the pairwise-local and Tullock group-level victory probabillities (\ref{sec:AdditionalPDSimulations}). We provide the proof of well-posedness for the measure-valued formulation of our PDE model in Section \ref{sec:proofwellposedness}, and we prove the preservation of the infimum and supremum H{\"o}lder exponents near $x = 1$ for our measure-valued solutions in Section \ref{sec:infsupHolderpreserved}. 

\section{Deriving Conjectured Threshold Selection Strength and Collective Success Formulas for Density Steady States}
\label{sec:steadytheta}

In this section, we will use heuristic calculations to obtain formulas for the group-level victory probability for steady state densities of our PDE model with pairwise between-group competition. We will derive an implicit formula that steady state solutions to Equation \eqref{eq:multilevelPDEtworho} must satisfy, and we will use this implicit representation to determine the average group-level success that will achieved by a steady state density with given H{\"o}lder exponent $\theta$ near $x=1$. 

\begin{proposition}
Suppose that the group-level victory probability satisfies $\rho(x,y) \in C^1\left([0,1]^2\right)$ and the individual-level level advantage to defect satisfies $\rho(x) \in C^1\left([0,1]\right)$ and $\rho(x) > 0$ for $x \in [0,1]$. If Equation \eqref{eq:multilevelPDEtworho} has a density steady state solution $f(x)$ with H{\"o}lder exponent $\theta$ near $x = 1$, then the average group-level victory probability achieved by the steady state population in pairwise competition satisfies
\begin{equation}
\int_0^1 \rho(y,1) f(y) dy = \frac{1}{2} - \frac{\theta \pi(1)}{2 \lambda}.
\end{equation}
\end{proposition}

\begin{proof}

We seek a probability density $f(x)$ that is a steady state solution to the multilevel dynamics of Equation \eqref{eq:multilevelPDEtworho}. Such a steady state solution must satisfy the following ODE
\begin{equation}
0 = - \dsddx{}{x} \left[ x(1-x) \pi(x) f(x) \right] + \lambda f(x) \left[ 2 \int_0^1 \rho(x,y) f(y) dy - 1 \right],
\end{equation}
and we can use separation of variables to rewrite the steady state ODE as
\begin{equation}
\frac{1}{f(x)} f'(x) =  \left[ \frac{- \lambda \left(2 \int_0^1 \rho(x,y) f(y) dy - 1 \right)}{x(1-x) \pi(x)} \right] -  \left[ \frac{\dsddx{}{x} \left(x (1-x) \pi(x) \right)}{x(1-x) \pi(x)} \right].
\end{equation}
We can then perform a partial fraction expansion to rewrite the righthand side as
\begin{equation} \label{eq:fprimesteady}
\begin{aligned}
\frac{1}{f(x)} f'(x) = -\frac{\frac{\lambda}{\pi(0)} \left[ 2 \int_0^1 \rho(0,y) f(y) dy  - 1 \right]}{x} - \frac{\frac{\lambda}{\pi(1)} \left[ 2 \int_0^1 \rho(1,y) f(y) dy - 1 \right]}{1-x} \\
+ \frac{\lambda C(x)}{\pi(x)} - \dsddx{}{x} \log\left[ x (1-x) \pi(x) \right],
\end{aligned}
\end{equation}
where $C(x)$ is given by
\begin{equation}
\begin{aligned}
C(x) = \frac{1}{x(1-x)} \left[ 1 - 2 \int_0^1 \rho(x,y) f(y) dy \right. & \left. + \left(\frac{(1-x) \pi(x)}{\pi(0)}\right) \left( 2 \int_0^1 \rho(0,y) f(y) dy - 1 \right) \right. \\ &\left. + \left( \frac{x \pi(x)}{\pi(1)} \right) \left( 2 \int_0^1 \rho(1,y) f(y) dy - 1 \right) \right]
\end{aligned}
\end{equation}
Using the shorthand notation 
\begin{equation}
H_f(x) := 2  \int_0^1 \rho(x,y) f(y) dy - 1,
\end{equation}
we can rewrite our expression for $C(x)$ as
\begin{equation} \label{eq:Cofxexpanded}
\begin{aligned}
C(x) &= \frac{1}{x(1-x)} \left[ -H(x) + \frac{(1-x) \pi(x) H(0)}{\pi(0)} + \frac{x \pi(x) H(1)}{\pi(1)} \right] \\
&= \left( \frac{H(1)}{\pi(1)} \right) \left( \frac{\pi(x) - \pi(1)}{1-x} \right) + \left( \frac{H(0)}{\pi(0)} \right) \left( \frac{\pi(x) - \pi(0)}{x} \right) \\
&+ \frac{H(1) - H(x)}{1-x} + \frac{H(0) - H(x)}{x}.
\end{aligned}
\end{equation}

Under the assumptions that $\rho(x,y)$ is continuously differentiable with respect to $x$ for $x \in [0,1]$ and that $f(x)$ is a probability density, we may use the dominated convergence theorem to deduce that $H(x)$ is a continuously differentiable function for $x \in [0,1]$ whose derivative is given by
\begin{equation}
H_f'(x) = 2 \int_0^1 \dsdel{\rho(x,y)}{x} f(y) dy.
\end{equation}
We can then use the fact that $H(x) \in C^1([0,1])$ and the expression for $C(x)$ from Equation \eqref{eq:Cofxexpanded} to deduce that $C(x)$ is a bounded function for $x \in [0,1]$.

With this information, we may take the antiderivative of both sides of Equation \eqref{eq:fprimesteady} to obtain the expression
that
\begin{equation}
\begin{aligned}
\log\left( f(x) \right) &= - \left( \frac{\lambda}{\pi(0)} \right) \left[2 \int_0^1 \rho(0,y) f(y) dy - 1 \right] \log(x) + \left( \frac{\lambda}{\pi(1)} \right) \left[ 2 \int_0^1 \rho(1,y) f(y) dy - 1 \right] \log(1-x) \\
&- \int_x^1 \frac{\lambda C(s)}{\pi(s)} ds - \log\left( x(1-x) \pi(x) \right) + c
\end{aligned}
\end{equation}
for an arbitrary constant $c$, and we can further exponentiate both sides to obtain the following implicit formula for a potential steady-state density $f(x)$
\begin{equation}
f(x) = \frac{1}{Z_f} x^{- \left( \frac{\lambda}{\pi(0)} \right) H(0) - 1} \left( 1 - x \right)^{\left( \frac{\lambda}{\pi(1)} \right) H(1) - 1}  \pi(x)^{-1} \exp\left( - \lambda \int_x^1 \frac{C(s)}{\pi(s)} ds \right),
\end{equation}
where $Z_f$ is a normalizing constant that guarantees that $f(x)$ is a probability density. We may may then seek to find the H{\"o}lder exponent of $f(x) dx$ near $x = 1$ by computing the following limit
\begin{equation}
\begin{aligned}
 & \lim_{x \to 0} \frac{\int_{1-x}^1 f(y) dy}{x^{\Theta}} \\  &= \lim_{x \to 0} \frac{f(1-x)}{\Theta x^{\Theta - 1}} \\
 &= \frac{1}{\Theta Z_f} \lim_{x \to 0} \left[ x^{\left( \frac{\lambda}{\pi(1)} \right) H(1)  - \Theta - 1}  \left( 1 - x \right)^{- \left( \frac{\lambda}{\pi(0)} \right) H(0) - 1} \pi(1-x)^{-1} \exp\left( - \lambda \int_{1-x}^1 \frac{C(s)}{\pi(s)} ds \right) \right],
 \end{aligned}
\end{equation}
and we can use the fact that $\frac{C(x)}{\pi(x)}$ is bounded for $x \in [0,1]$ to deduce that
\begin{equation}
 \lim_{x \to 1} \frac{\int_{1-x}^1 f(y) dy}{x^{\Theta}} = \left\{
    \begin{array}{cr}
      0 & : \Theta <  \left( \frac{\lambda}{\pi(1)} \right) H(1) \\
     \ds\frac{1}{\lambda Z_f H(1)} & \Theta =  \left( \frac{\lambda}{\pi(1)} \right) H(1)  \\
      \infty & \Theta >  \left( \frac{\lambda}{\pi(1)} \right) H(0).
     \end{array}
  \right.
\end{equation}
From the definition of the H{\"o}lder exponent near $x=1$, we see that a steady state $f(x)$ satisfying this implicit representation formula will have H{\"o}lder exponent $\theta$ near $x=1$ that is given by
\begin{equation}
\theta =  \left( \frac{\lambda}{\pi(1)} \right) H(1),
\end{equation}
and we can use the definition of $H(x)$ to further write our H{\"o}lder exponent as
\begin{equation}
\theta = \frac{\lambda}{\pi(1)} \left( 2 \int_0^1 \rho(1,y) f(y) dy - 1 \right).
\end{equation}
We can rearrange this expression write the average group-level victory probability for the all-cooperator group against a population at steady state $f(x)$ in terms of the H{\"o}lder exponent as
\begin{equation}
\int_0^1 \rho(1,y) f(y) dy = \frac{1}{2} + \frac{\theta \pi}{\lambda}.
\end{equation}
Similarly, we can use the relation that $\rho(x,1) = 1 - \rho(1,x)$ to write the average group-level victory probability for groups at steady state in pairwise competition with an all-cooperator group as
\begin{equation} \label{eq:rhoy1theta}
\int_0^1 \rho(y,1) f(y) dy = \frac{1}{2} - \frac{\theta \pi(1)}{\lambda}. 
\end{equation}
Taking the limit as $\lambda \to \infty$, we see that
\begin{equation}
\lim_{\lambda \to \infty} \int_0^1 \rho(y,1) f(y) dy = \frac{1}{2}, 
\end{equation}
so the steady state population with density $f(x)$ will have, on average, an equal chance of defeating an all-cooperator group in pairwise conflict in the limit of infinitely-strong between-group competition. 
\end{proof}

Returning to the expression from Equation \eqref{eq:rhoy1theta}, we may ask when we expect the average group-level success $\int_0^1 \rho(y,1) f(y) dy$ against the all-cooperator group at a density steady state $f(x)$ to outperform the probability $\rho(0,1)$ that the all-defector group defeats the all-cooperator group in pairwise competition. If $\int_0^1 \rho(y,1) f(y) dy > \rho(0,1)$ for a steady-state density $f(x)$, we may expect that this population distribution would result in a better collective outcome at steady state than the outcome achieved by a population concentrated upon a delta-function $\delta(x)$ at the all-defector group. Using Equation \eqref{eq:rhoy1theta}, we see that this inequality is satisfied provided that
\begin{equation}
\int_0^1 \rho(y,1) f(y) dy = \frac{1}{2} - \frac{\theta \pi(1)}{2 \lambda} > \rho(0,1),
\end{equation}
and we can rearrange this condition to see that this can be achieved when the strength of between-group selection exceeds the threshold quantity
\begin{equation}
\lambda > \lambda^*_{PD}(\theta) := \frac{\theta \pi(1)}{1 - 2 \rho(0,1)}.
\end{equation}
Using the fact that $\rho(1,0) = 1 - \rho(0,1)$ and our assumption that $\pi(x) := \pi_D(x) - \pi_C(x)$ corresponds to a net individaul-level advantage for defections, we can also rewrite this threshold $\lambda^*_{PD}$ for the steady-state population to outperform the all-defector group in pairwise competition with the all-cooperator group as
\begin{equation}
\lambda^*_{PD}(\theta) = \frac{\theta \pi(1)}{\rho(1,0) - \rho(0,1)} = \frac{\theta \overbrace{\left(\pi_D(1) - \pi_C(1)\right)}^{\substack{\textnormal{Individual incentive} \\ \textnormal{to defect}}}}{\underbrace{\rho(1,0) - \rho(0,1)}_{\substack{\textnormal{Collective incentive} \\  \textnormal{to cooperate}}}},
\end{equation}
suggesting that our necessary condition for the existence of such a steady state featuring cooperation is determined by balancing the tug-of-war between the individual-level incentive to defect in an all-cooperator group ($\pi_D(1) - \pi_C(1)$) and a collective advantage of an all-cooperator group when engaged in pairwise conflict with an all-defector group ($\rho(1,0) - \rho(0,1)$). 

\section{More Details on Numerical Simulations of PDE Model with Pairwise Group-Level Competition}
\label{sec:numericsappendix}

We now provide further discussion of numerical simulations for our PDE model of with pairwise group-level conflict. In Section \ref{sec:numericalscheme}, we present the upwind finite volume scheme that was used for the simulations in Section \ref{sec:numerics}. In Section \ref{sec:AdditionalPDSimulations}, we provide additional simulations for our model of multilevel selection for the cases of group-level conflict following the pairwise local update rule and the Tullock contest function. These simulations suggest that that the good agreement between the conjectured analytical formulas and numerical simulations generalizes from the Fermi group-level update rule studied in the main text to a broader class of possible ways to model frequency-dependent group-level competition. 

\subsection{Upwind Finite Volume Scheme for PDE Model with Pairwise Between-Group Competition}
\label{sec:numericalscheme}

We now summarize the upwind finite volume scheme that we used in Section \ref{sec:numerics} to provide numerical solutions to our PDE model of multilevel selection. This scheme was previously derived to perform numerical simulations for a model with pairwise group-level competition for the case of game-theoretic dynamics featuring altruistic punishment \cite[Section E.1]{cooney2024exploring}.

Our numerical scheme will consider $N+1$ gridpoints $\{x_i\}_{i \in \{0,1,\cdots,N\}} = \{ \frac{i}{N} \}_{i \in \{0,1,\cdots,N\}}$ and we aim to describe the solution $f(t,x)$ to our PDE model for multilevel selection in terms of the average values 
\begin{equation}
f_i(t) := \frac{1}{x_{i+1} - x_i} \int_{x_i}^{x_{i+1}} f(t,y) dy = N \int_{x_i}^{x_{i+1}} f(t,y) dy 
\end{equation}
on the intervals interval $[x_i,x_{i+1})$ 
as a piecewise constant function that takes a constant value on each interval between two gridpoints. and we will represent the numerical solution $f(t,x)$ to the multilevel dynamics of Equation \eqref{eq:multilevelPDEfirstform} as a piecewise constant function of the form
\begin{equation}
f(t,x) = \sum_{i=0}^{N-1} f_i(t) 1_{x \in [x_j,x_{j+1}]}.
\end{equation}
We then study how the constants $f_i(t)$ evolve in time according to the following system of ODEs
\begin{equation}
\begin{aligned}
\dsddt{f_i(t)} &= N x_{i+1} \left( 1 - x_{i+1} \right) \left( \pi_C(x_{i+1}) - \pi_D(x_{i+1}) \right) f_{i+1}^{UW}(t) \\
&- N x_{i} \left( 1 - x_i \right) \left[ \pi_C(x_i) - \pi_D(x_i) \right] f_{i}^{UW}(t) + \lambda f_i(t) \left( \frac{1}{N} \sum_{i = 0}^{N-1} \left[ \rho_{ij} - \rho_{ji} \right] f_j(t) \right).
\end{aligned}
\end{equation}
where 
\begin{equation}
\rho_{ij} := \int_{x_i}^{x_{i+1}} \int_{x_j}^{x_{j+1}} \rho(x,y) dy dx
\end{equation}
describes a discretization of the group-level victory probability $\rho(x,y)$ and
\begin{equation}
f_i^{UW}(t) := \left\{
    \begin{array}{lr}
      f_i(t) & :x_i \left( 1 - x_i \right) \left[ \pi_C(x_i) - \pi_D(x_i) \right] \geq 0\\
      f_{i-1}(t) & : x_i \left( 1 - x_i \right) \left[ \pi_C(x_i) - \pi_D(x_i) \right] < 0
    \end{array}
  \right.
\end{equation}
implements the upwinding convention for describing the advection term based on the sign of the within-group dynamics at the boundary between grid volumes. To run numerical simulations for this PDE, we use the trapezoidal rule to obtain the discretized victory probability $\rho_{ij}$ for a given group-level victory function $\rho(x,y)$, and we start our simulations from the uniform initial conditions given by $f_i(t) = 1$ for each $i \in \{0,1,\cdots,N-1\}$. 

\subsection{Numerical Simulations of Multilevel PD Dynamics with Local and Tullock Victory-Probabilities}
\label{sec:AdditionalPDSimulations}

We now present results for numerical solutions of the multilevel dynamics for local group comparison and Tullock group-level victory probabilities for the case of tension between levels of selection corresponding to the PD scenario. As in Section \ref{sec:PDnumerics}, we consider replication events whose rates depend on the individual advantage $\pi(x)$ for defectors and the average payoff of group members $G(x)$ that arise from playing a PD game within each group, and we highlight similarities seen in our simulations with the behavior shown in Section \ref{sec:PDnumerics} for the case of the group-level Fermi update rule. Before providing simulations of the dynamical behavior of our PDE model for these two group-level update rules, we will illustrate the properties of our conjectured value on the threshold selection strength $\lambda^*_{PD}(\theta)$ for these two ways of modeling group-level competition. 

For the pairwise normalized local group comparison rule, we have a group-level victory probability given by
\begin{equation}
\rho(x,y) = \frac{1}{2} + \frac{1}{2} \left[ \frac{G(x) - G(y)}{|G(x)| + |G(y)|} \right],
\end{equation}
so we expect the multilevel dynamics for the PD scenario to support long-time cooperation when the relatives strength $\lambda$ of between-group competition exceeds the threshold value
\begin{align}
\lambda > \lambda^*_{PD}(\theta) := \frac{\theta \pi(1)}{2 \rho(1.0) - 1} = \frac{-\left(\beta + \alpha\right) \theta}{\left( \frac{G(1) - G(0)}{|G(1)| + |G(0)|} \right)} = - \frac{\left( |P + \gamma + \alpha| + |P| \right)\left(\beta + \alpha \right) \theta}{\gamma + \alpha}.
\end{align}
We note that we can rewrite this threshold quantity in terms of the entries of the payoff matrix for our PD game by noting that $\alpha = R - S - T + P$, $\beta = S-P$, and $\gamma = S + T  - 2P$, allowing us to see that
\begin{equation}
\lambda^*_{PD} = \left( |R| + |P| \right) \left( \frac{T-R}{R - P} \right) \theta
\end{equation}
for the pairwise local update under the PD game. Similarly, we consider the Tullock contest function with group-level victory probability given by
\begin{equation}
\rho(x,y) = \frac{\left(G(x) - G_*\right)^{1/a}}{\left(G(x) - G_*\right)^{1/a} + \left(G(y) - G_*\right)^{1/a}},
\end{equation}
where $G_{*} := \min_{x \in [0,1]} G(x)$. 

We can further see from Lemma \ref{lem:payoffproperties} that the minimum possible average payoff for a PD game has the following characterization
\begin{equation}
\begin{aligned}
G_* &= \left\{
    \begin{array}{cr}
      G(0) & : \gamma \geq 0\\
      G\left( \ds\frac{\gamma}{-2\alpha} \right) & : \gamma < 0
    \end{array}
  \right. \\
  &= \left\{
    \begin{array}{cr}
      P & : \gamma \geq 0\\
      P - \ds\frac{\gamma^2}{4 \alpha} & : \gamma < 0
    \end{array}
  \right. ,
  \end{aligned} 
\end{equation}
and we can therefore use the expressions $\pi(x) = - \beta - \alpha x$ and $G(x) = P + \gamma x + \alpha x^2$ to see that the group-level victory probability $\rho(1,0)$ of an all-cooperator group over an all-defector group is given by
\begin{equation}
\begin{aligned}
\rho(1,0) &= \frac{\left(G(1) - G_*\right)^{1/a}}{\left(G(1) - G_*\right)^{1/a} + \left(G(0) - G_*\right)^{1/a}} =  \left\{
    \begin{array}{cr}
      1 & : \gamma \geq 0\\
      \frac{\left( \gamma + \alpha - \frac{\gamma^2}{4 \alpha}\right)^{1/a}}{\left( \gamma + \alpha - \frac{\gamma^2}{4 \alpha}\right)^{1/a} + \left(\frac{\gamma^2}{4 \alpha} \right)^{1/a}} & : \gamma < 0
    \end{array}
  \right. ,
  \end{aligned}
\end{equation}
and we can consequently deduce that the conjectured threshold selection $\lambda^*_{PD}$ takes the following form for the Tullock group-level victory probability
\begin{equation}
\lambda^*_{PD}(\theta) = \frac{\theta \pi(1)}{2 \rho(1,0) - 1} = 
\left\{
\begin{array}{cr}
      - \left( \beta + \alpha \right) \theta & : \gamma \geq 0\\
        - \left( \beta + \alpha \right)  \left( \frac{\left( \gamma + \alpha - \frac{\gamma^2}{4 \alpha}\right)^{1/a} + \left(\frac{\gamma^2}{4 \alpha} \right)^{1/a}}{\left( \gamma + \alpha - \frac{\gamma^2}{4 \alpha}\right)^{1/a} - \left(\frac{\gamma^2}{4 \alpha} \right)^{1/a}} \right)  \theta & : \gamma < 0. 
    \end{array}
  \right.
\end{equation}

Having established the values of our conjectured selection strength $\lambda^*_{PD}(\theta)$ for the pairwise local and Tullock group-level victory probabilities, we will now explore numerical simulations of our multilevel model for these update rules using the upwind finite volume scheme described in Section \ref{sec:numericalscheme}. In Figure \ref{fig:PDsteadyOtherVictory}, we plot the numerical solutions for the strategic composition of groups after 9,600 time-steps of step size $\Delta t = 0.01$, comparing the densities achieved by these two group-level victory probabilities for various strengths of between-group competition $\lambda$. For both sets of simulations, we consider a PD game for which the average payoff of group members is maximized by a group consisting of 75 percent cooperators, while a group featuring 50 percent cooperators receives the same average payoff as an all-cooperator group. We see that, for both group-level victory probabilities, the steady state densities appear to concentrate upon a level $\ol{x} = 0.5$ that is much lower than the optimal group composition $x^* = 0.75$, even for relatively large strengths of group-level competition. This behavior agrees with the numerical results seen in the case of the Fermi group-level update rule, and suggests that the shadow of lower-level selection may also hold for multilevel dynamics with pairwise group-level competition following the group-level local update rule and the Tullock group-level victory probability. 

\begin{figure}[!ht]
    \centering
    \includegraphics[width=0.48\linewidth]{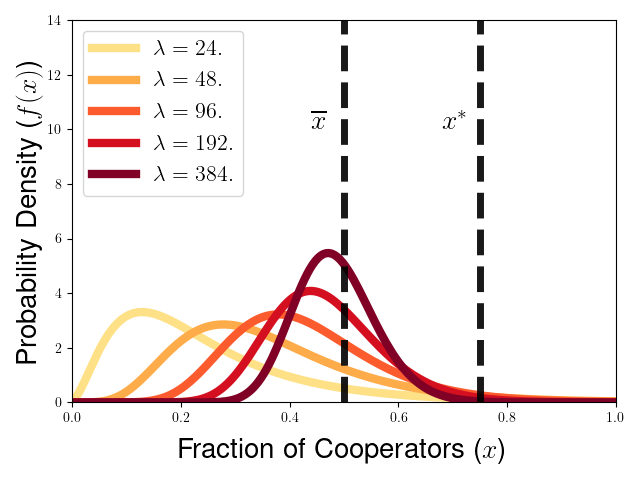}
    \includegraphics[width=0.48\linewidth]{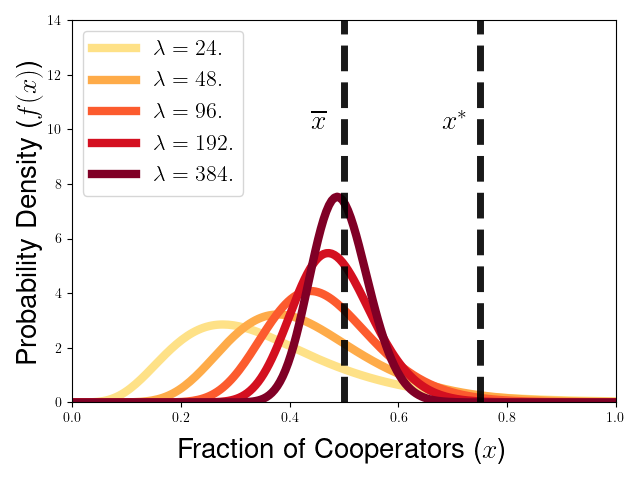}
    \caption{Comparison of densities achieved after 9,600 time steps with step-size $\Delta t = 0.01$ for multilevel PD dynamics with the pairwise normalized local rule (left) and the Tullock group-level victory probability (right) for different relative strengths of group-level competition $\lambda$. The vertical dashed lines describe the fraction of cooperation $\ol{x} = 0.5$ with the same average payoff as the all-cooperator group and the fraction of cooperation $x^* = 0.75$ that maximizes the average payoff $G(x)$ of group members. The simulations were run starting from a uniform initial density, the game-theoretic parameters were fixed at $\gamma = 1.5$, $\alpha = \beta = -1$, and $P = 1$, and the simulations in the right panel used the Tullock group-level victory probability with payoff sensitivity parameter $a = 0.5$.  }
    \label{fig:PDsteadyOtherVictory}
\end{figure}

In Figure \ref{fig:PDgroupsuccessappendix}, we plot the average group-level victory probability against the all-cooperator group after 9,600 time-steps for a PD game with both the local pairwise and Tullock update rules. As in the case of the Fermi group-level victory probability studied in Section \ref{sec:PDnumerics}, we see that there is good agreement between the numerically computed collective success for the population and the piecewise analytical formula $\int_0^1 \rho(y,1) f(y) dy$ obtained for steady state densities in Section \ref{sec:PDsteadygroupsuccess}. In addition, we see that collective success is initially constant and equal to $\rho(0,1)$ for low $\lambda$, before increasing for $\lambda$ exceeding our conjectured threshold selection strength $\lambda^*_{PD}$ for the given group-level victory probability under consideration. We further see that the long-time average collective success against the all-cooperator group approaches $\frac{1}{2}$ in the limit of large $\lambda$, suggesting that the population concentrates upon the composition $\ol{x}$ for which $G(\ol{x}) = G(1)$ as group-level competition becomes arbitrarily strong.

\begin{figure}[!ht]
    \centering
    \includegraphics[width = 0.48\textwidth]{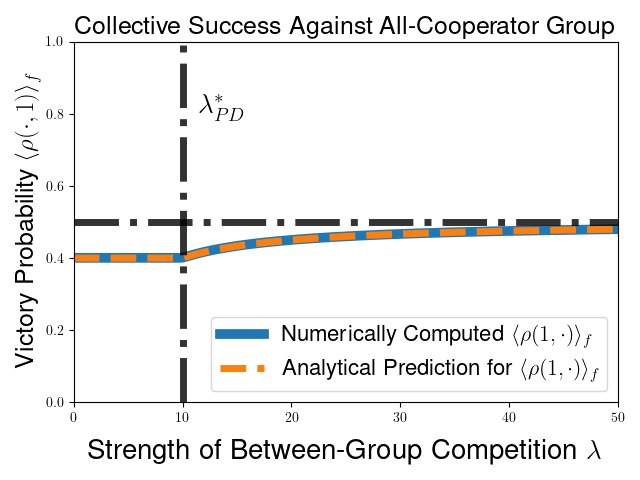}
        \includegraphics[width = 0.48\textwidth]{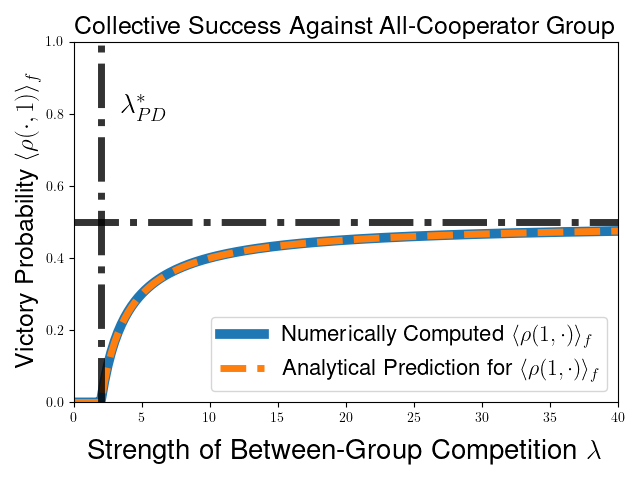}
    \caption{Comparison of the group-level success measured from numerical simulations and the predicted value for a density steady state for the cases of pairwise local (left) and Tullock (right) group-level victory probabilities, plotted as a function of $\lambda$ for a PD game in which average group payoff is maximized by a group fraction $x^* = \frac{3}{4}$ cooperators. We plot both the numerical computation $\int_0^1 \rho(y,1) f(t,y) dt$ of the population against the all-cooperator group (solid blue curve) with the predicted formula from Equation \eqref{eq:rhoy1steadystate} of the collective success $\int_0^1 \rho(y,1) f(y) dy$ for a bounded density steady state (dashes orange curve). The black vertical dash-dotted line gives the predicted threshold selection strength $\lambda_{PD}(1)$ from Equation \eqref{eq:lambdastarPDtheta} for the case of H{\"o}lder exponent $\theta = 1$. Payoff parameters are $\gamma = 1.5$, $\alpha = \beta = -1$, and $P = 1$.}
    \label{fig:PDgroupsuccessappendix}
\end{figure}

Finally, in Figure \ref{fig:PDaveragecooperatorsappendix},  we study the average level of cooperation achieved under the multilevel dynamics with both update rules plotted as a function of the relative strength $\lambda$ of group-level competition. The behavior for the two update rules is similar to that seen for the case of the Fermi group-level update rule, showing that the average fraction of cooperation approaches the composition $\ol{x}$ achieving the same average payoff of the all-cooperator group. This further highlights the fact that multilevel selection with pairwise group-level conflict appears to achieve a level of cooperation substantially less than the collectively optimal outcome even in the limit of infinitely strong between-group competition. This reflects a potential similarity with prior work on PDE models of multilevel selection with frequency-independent group-level competition, which motivates potential future work in exploring the generality of multilevel selection scenarios displaying this limitation upon achieving collectively optimal outcomes.

\begin{figure}[!ht]
    \centering
    \includegraphics[width = 0.48\textwidth]{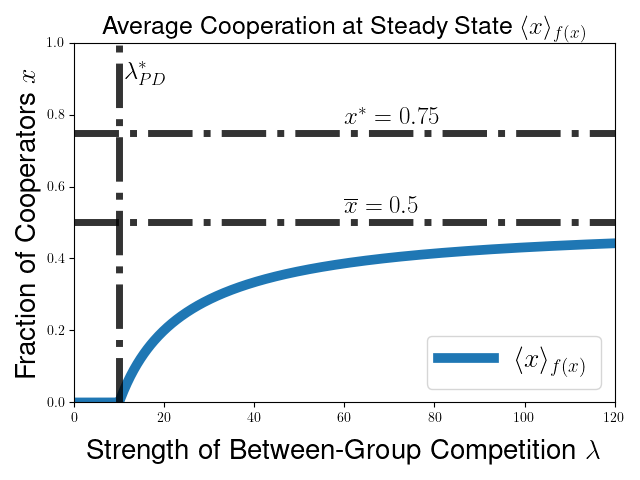}
    \includegraphics[width = 0.48\textwidth]{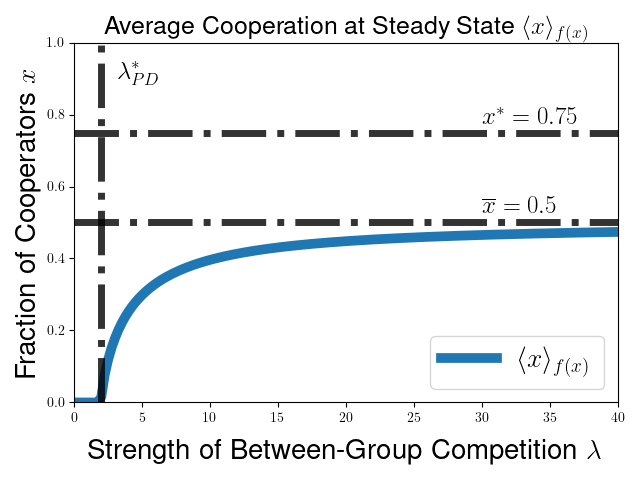}
    \caption{Comparison of the  average payoff $\int_0^1 G(x) f(t,x) dx$ for the numerical simulation after 9,600 time-steps of step-size 0.01 (blue curves) for the case of the pairwise local (left) and Tullock (right) group-level victory probabilities. The vertical dash-dotted line corresponds to the predicted threshold strength of between-group selection $\lambda^*_{PD}(1)$ from Equation \eqref{eq:lambdastarPDtheta}. The upper horizontal dash-dotted line corresponds to the level of cooperation $x^* = 0.75$ that maximizes average payoff of group members, while the lower horizontal dash-dotted line corresponds to the level of cooperation $\overline{x} = 0.5$ at which the average payoff of group members is equal to the average payoff $G(1)$ of the all-cooperator group. The game-theoretic parameters were fixed at $\gamma = 1.5$, $\beta = -1$, and $\alpha = - 1$ for all simulations, and the sensitivity parameter for the Tullock victory probability  was set at $a = 2$. }
    \label{fig:PDaveragecooperatorsappendix}
\end{figure}

\section{Proof of Well-Posedness for Measure-Valued Solutions to Multilevel Dynamics}
\label{sec:proofwellposedness}

In this section, we present the proofs of well-posedness and derivations of representation formulas of solutions for our PDE models for multilevel selection and the associated linear PDEs we use for our fixed point arguments. We present the proof of Lemma describing well-posedness of the associated linear equation in Section \ref{sec:linearwellposednessproof} (Lemma \ref{lem:linearwellposedness}), and we then apply this result and a fixed point argument to prove the well-posedness of the nonlinear model of multilevel selection in Section \ref{sec:nonlinearwellposednessproof} (Theorem \ref{thm:nonlinearwellposedness})

\subsection{Well-Posedness of Measure Solutions for Associated Linear PDE}
\label{sec:linearwellposednessproof}

\begin{proof}[Proof of Lemma \ref{lem:linearwellposedness}]
We start by noting that the flow of measures $\mu_t^{\nu} = \{\mu_t^{\nu}\}_{t \in [0,T]} \in C\left([0,T]; \mc{M}\left([0,1]\right)\right)$. We can confirm this by using our representation formula from Equation \eqref{eq:hlinearrepresentation} for $\mu_{t+h}^{\nu}$ and $\mu_t^{\nu}(dx)$ and considering the difference quotient
\begin{equation}
\begin{aligned}
&\bigg| \int_0^1 v(x) \mu_{t+h}^{\nu}(dx) - \int_0^1 v(x) \mu_t^{\nu}(dx) \bigg| \\ &= \bigg| \int_0^1 v(\phi_t(x)) \left[ e^{\lambda \left[ \int_0^{t+h} \left\{ 2 \int_0^1 \rho(\phi_s(x),y) \nu_s(dy) \right\} ds - t\right]}  - e^{\lambda \left[ \int_0^{t} \left\{ 2 \int_0^1 \rho(\phi_s(x),y) \nu_s(dy) \right\} ds - t\right]}   \right] \mu_0(dx) \bigg|
\end{aligned}
\end{equation}
for $t \in [0,T]$ and any test function $v(x) \in C\left([0,1]\right)$. We can then use our assumption that $\nu = \{ \nu_t\}_{t \in [0,T]} \in C\left( [0,T] ; \mc{M}\left([0,1]\right) \right)$ and the dominated convergence theorem to show that
\begin{equation}
\bigg| \int_0^1 v(x) \mu_{t+h}^{\nu}(dx) - \int_0^1 v(x) \mu_t^{\nu}(dx) \bigg|\to 0 \: \: \mathrm{as} \: \: h \to 0
\end{equation}
for $t \in [0,T]$ and $v(x) \in C\left([0,1]\right)$, allowing us to conclude that $\mu_t^{h} \in C\left([0,T];\mc{M}\left([0,1]\right)\right)$ with continuity in the weak-$*$ topology.

Next, we establish that the flow of measures $\mu = \{\mu_t\}_{t \in [0,T]}$ defined by Equation \eqref{eq:hlinearrepresentation} is a solution to the linear PDE of Equation \eqref{eq:PDEhlinear}. To do this, we consider the representation formula for $\mu_t^{h}$ frpm Equation \eqref{eq:hlinearrepresentation} for any test function $v(x) \in C^1\left([0,1] \right)$, and we apply the dominated convergence theorem to differentiate both sides of Equation \eqref{eq:hlinearrepresentation} with respect to $t$ to see that
\begin{equation}
\begin{aligned}
\dsddt{} \int_0^1 v(x) \mu_t^{\nu}(dx) &= \int_0^1 v'(\phi_t(x)) \dsdel{\phi_t(x)}{t} e^{\lambda \left[\int_0^t \lambda \left\{ 2\int_0^1 \rho(\phi_s(x),y) \nu_s(dy) \right\} ds - t \right]} \mu_0(dx) \\  &+ \lambda \int_0^1 v(\phi_t(x)) \left[ 2 \int_0^1 \rho(\phi_t(x),y) \nu_t(dy) - 1 \right] e^{\lambda \left[\int_0^t \lambda \left\{ 2\int_0^1 \rho(\phi_s(x),y) \nu_s(dy) \right\} ds - t \right]} \mu_0(dx) \\
&=  -\int_0^1 v'(\phi_t(x)) \phi_t(x) \left( 1 - \phi_t(x) \right) \pi(\phi_t(x)) e^{\lambda \left[\int_0^t \lambda \left\{ 2\int_0^1 \rho(\phi_s(x),y) \nu_s(dy) \right\} ds - t \right]} \mu_0(dx) \\
&+ \lambda \int_0^1 v(\phi_t(x)) \left[ 2 \int_0^1 \rho(\phi_t(x),y) \nu_t(dy) - 1 \right]  e^{\lambda \left[\int_0^t \lambda \left\{ 2\int_0^1 \rho(\phi_s(x),y) \nu_s(dy) \right\} ds - t \right]} \mu_0(dx) \\
&= \int_0^1 \left\{-v'(x) \left[x(1-x) \pi(x) \right]  + \lambda v(x) \left[ 2 \int_0^1 \rho(\phi(x,y) \nu_t(dy) - 1 \right]  \right\} \mu_t^{\nu}(dx),
\end{aligned}
\end{equation}
where we obtained the final equality by applying the definition of $\mu_t^h$ as a linear functional acting on the test function defined in curly braces. 

This tells us that the measure from Equation \eqref{eq:hlinearrepresentation} is a weak solution of the linear PDE given by Equation \eqref{eq:PDEhlinear}. 

Now we show that Equation \eqref{eq:PDEhlinear} has a unique weak solution given an initial measure $\mu_0(dx)$ and a given flow of measures $\{\nu_t\}_{t \geq 0}$. To show this, we suppose that there are two solutions $\mu_t^{\nu}(dx)$ and $\tilde{\mu}^{\nu}_t(dx)$  that solve Equation \eqref{eq:PDEhlinear} with initial measure $\mu^{\nu}_0(dx) = \tilde{\mu}_0^{\nu}(dx) = \mu_0(dx)$.

To highlight the fact that the depends of the group-level birth term on both the flow of measures $\{\nu_t\}_{t \geq 0}$ and the fraction of cooperators $x$, we will employ the shorthand notation
\begin{equation} \label{eq:hshorthand}
h_{\nu}(t,x) := 2 \int_0^1 \rho(x,y) \nu_t(dy).
\end{equation}
Using this notation, we may rewrite Equation \eqref{eq:PDEhlinear} in the following form
\begin{equation} \label{eq:PDEhshorthand}
\begin{aligned}
\dsddt{} \int_0^1 v(x) \mu_t^{\nu}(dx) &= -\int_0^1 v'(x) x(1-x) \pi(x) \mu_t^{\nu}(dx) + \lambda \int_0^1 v(x) \left[ h(t,x)  - 1 \right] \mu_t^{\nu}(dx) \\
\mu_0^{\nu}(dx) &= \mu_0(dx). 
\end{aligned}
\end{equation}

Using these assumptions, we can integrate the representation of our linear PDE from Equation \eqref{eq:PDEhshorthand} in time to compute that, for any test function $v(x) \in C^1([0,1])$,
\begin{equation}
\begin{aligned}
\langle v , \mu_t(dx) - \nu_t(dx) \rangle &= \int_0^1 v(x) \mu_t(dx) - \int_0^1 v(x) \nu_t(dx)  \\ 
&-  \cancelto{0}{\int_0^1 v(x) \mu_0(dx) - \int_0^1 v(x) \nu_0(dx)} \\ 
&+ \int_0^t \left[ \int_0^1 v'(x) x (1-x) \pi(x) \mu_s^{\nu}(dx) - \int_0^1 v'(x) x(1-x) \pi(x) \tilde{\mu}^{\nu}_s(dx) \right] ds \\
&+ \lambda \int_0^t \left[\int_0^1  v(x) \left\{ h_{\nu}(t,x) - 1 \right\} \mu_s^{\nu}(dx) - \int_0^1 v(x) \left\{h_{\nu}(t,x) - 1 \right\} \tilde{\mu}^{\nu}_s(dx)  \right] ds \\
&=  \int_0^t \left[ \langle v'(x) x (1-x) \pi(x) , \mu^{\nu}_s \rangle + \lambda \langle v \left[ h_{\nu}(s,x) - 1 \right], \mu^{\nu}_s - \tilde{\mu}^{\nu}_s \rangle  \right] ds \\
\end{aligned}
\end{equation}
This allows us to see that
\begin{equation}
|\langle v , \mu_t^{\nu} - \tilde{\mu}^{\nu}_t \rangle | \leq \int_0^t \left[ \frac{1}{2} ||v'||_{\infty} \pi_{M} + ||v||_{\infty} \left(1 + ||h_{\nu}(s,x)||_{\infty} \right) \right] \sup_{\substack{u \in C(0,1) \\ ||u||_{\infty} = 1 }} \langle u, \mu_s^{\nu} - \tilde{\mu}^{\nu}_s  \rangle ds
\end{equation}
Using Equation \eqref{eq:hshorthand} and our assumption that $\rho(x,y)$ is a probability, we can also estimate that
\begin{equation}
| h(s,x) | = \bigg| 2 \int_0^1 \rho(x,y) \nu_s(dy) \bigg| \leq 2 || \rho(x,y)||_{L^{\infty}([0,1]^2)} \int_0^1 \nu_s(dy) \leq 2 ||\nu_s||_{TV},
\end{equation}
and therefore $||h(s,x)||_{\infty} \leq 2 ||\nu_s||_{TV}$. We can then use this and the fact that
\begin{equation}
\sup_{\substack{u \in C([0,1]) \\ ||u||_{\infty} = 1 }} \langle u , \mu_t - \nu_t \rangle = \sup_{\substack{u \in C^1([0,1])  \\ ||u||_{\infty} = 1 }}  \langle u , \mu_t - \nu_t \rangle = ||\mu_t - \nu_t||_{TV},
\end{equation}
to further deduce that
\begin{equation}
||\mu_t^{\nu} - \tilde{\mu}^{\nu}_t||_{TV} \leq \int_0^t  \left[ \frac{1}{2} ||v'||_{\infty} \pi_{M} + ||v||_{\infty} \left(1 + 2||\nu_s||_{TV} \right) \right] || \mu_s^{\nu} - \tilde{\mu}^{\nu}_s ||_{TV} ds.
\end{equation}
From Gr{\"o}nwall's inequality, we can now deduce that $||\mu_t^{\nu} - \tilde{\mu}^{\nu}_t||_{TV} = 0$ for all $t \geq 0$, and we can therefore conclude that Equation \eqref{eq:PDEhlinear} has a unique weak, measure-valued solution for a given initial measure $\mu_0(dx)$. 
\end{proof}

\subsection{Well-Posedness for Full Nonlinear PDE Model}
\label{sec:nonlinearwellposednessproof}

\begin{proof}[Proof of Theorem \ref{thm:nonlinearwellposedness}]

 It was shown by Gwiadza and coauthors that the space of measures
\begin{equation}
\mc{X}_R := \left\{ \sigma \in \mc{M}\left([0,1]\right) \bigg| ||\sigma||_{TV} \leq R \right\}
\end{equation}
is a complete metric space under the bounded Lipschitz norm \cite{gwiazda2010nonlinear}. This means that the space of flows of measures 
\begin{equation}
\mc{X}_R(T) := \left\{ \mu \in C\left([0,T];\mc{M}([0,1]) \right) \bigg| \mu_0 = \mu^0, \: \:  || \mu_t ||_{TV} \leq R \: \: \textnormal{for all} \: \: t \in [0,T]  \right\},
\end{equation}
is a Banach space under the supremum norm 
\begin{equation}
||\mu||_{C\left([0,T];\mc{M}([0,1]) \right)} := \sup_{t \in [0,T]} || \mu_t ||_{BL}
\end{equation} \cite{ackleh2020well,canizo2013measure}, so we will perform our fixed point argument in $\mc{X}_R(T)$. 

To do this, we now need to show that the map $H(\nu)$ satisfies the requirements to be a contraction mapping on $\mc{X}_R$. We need verify the following two properties of our map $H(\nu)$

\begin{enumerate}[(i)]
\item There exists a sufficiently large $R$ and a sufficiently small time $T_R$ such that $H(\nu): \mc{X}_R \mapsto \mc{X}_R$ for all $t \in [0,T]$. %
\item There exists $\eta \in (0,1)$ such that, for any pair $\nu_t(dx), \tilde{\nu}_t(dx) \in C\left([0,T];\mc{M}([0,1])\right)$,
\begin{equation}
||H(\nu) - H(\tilde{\nu}) ||_{C\left([0,T];\mc{M}([0,1])\right)} \leq \eta ||\nu - \tilde{\nu} ||_{C\left([0,T];\mc{M}([0,1])\right)}.
\end{equation}
\end{enumerate}

We first show that $H(\nu)$ maps $\mc{X}_R$ to itself. To do this, we need to obtain a bound on $||H(\nu)_t||_{TV} = ||\mu^{\nu}_t||_{TV}$ that holds for all $t \in [0,T]$.

Given a test function $v(x) \in C([0,1])$, we use Equation \eqref{eq:hlinearrepresentation} and the assumption that $\mu_0(dx)$ is a probability measure to estimate that
\begin{equation}
\begin{aligned}
\bigg| \int_0^1 v(x) \mu_t^{\nu}(dx) \bigg| &= \bigg| \int_0^1 v(\phi_t(x)) \exp\left(\lambda \int_0^t \left\{ 2 \int_0^1 \rho(\phi_s(x),y) \nu_s(dy)  ds   \right\} ds - t \right)  \mu_0(dx) \bigg| \\
& \leq \int_0^1 \bigg| v(\phi_t(x)) e^{-\lambda t} \bigg| \bigg| \exp\left( 2 \lambda  \int_0^t \int_0^1  \left[\rho(x,y) \nu_s(dy) \right] ds \right) \bigg|  \mu_0(dx) \\
& \leq ||v||_{\infty} \exp\left( 2 \lambda || \rho||_{L^{\infty}\left( [0,1] \right)^2} \int_0^t \int_0^1 \nu_s(dy) ds \right) \mu_0(dx)\\
& \leq ||v||_{\infty} \exp\left( 2 \lambda T \left[\sup_{t \in [0,T]} || \mu_s(dx) ||_{TV} \right] \right) \int_0^1 \mu_0(dx) \\
& \leq ||v||_{\infty} e^{ 2 \lambda R T}.
\end{aligned}
\end{equation}
Because this holds for all test functions $v(x) \in C\left([0,1]\right)$, we may further deduce that
\begin{equation}
|| H(\nu)_t ||_{TV} = ||\mu_t^{\nu}||_{TV} \leq e^{2 \lambda R T}.
\end{equation}
Therefore we see that choosing any $R > 1$ and $T  < \frac{1}{2 \lambda R}$ will allow us to deduce that
\begin{equation}
||H(\nu)_t||_{TV} \leq 1 < R
\end{equation}
for all $t \in [0,T]$, so choosing $T < T_R := \frac{1}{2 \lambda R}$ guarantees that $H(\nu)_t \in \mc{X}_R$ for $t \in [0,T]$, and we can that conclude our mapping $H(\cdot)$ satisfies condition $(i)$.

Next, we look to verify condition $(ii)$, showing that the mapping $H(\nu)$ satisfies a contraction estimate. To do this, we use Equation \eqref{eq:hlinearrepresentation} to write that, for any test function $v(x) \in W^{1,\infty}\left([0,1]\right)$,
\begin{equation}
\begin{aligned}
& \bigg| \int_0^1 v(x) \mu_t^{\nu}(dx) - \int_0^1 v(x) \mu_t^{\tilde{\nu}}(dx) \bigg| \\ & \leq \bigg| \int_0^1 v(\phi_t(x)) \int_0^1 v(x) \exp\left(\lambda \int_0^t \left\{ 2 \int_0^1 \rho(\phi_s(x),y) \nu_s(dy)  ds   \right\} ds - t \right)  \mu_0(dx)  \\ & - \int_0^1 v(\phi_t(x)) \exp\left(\lambda \int_0^t \left\{ 2 \int_0^1 \rho(\phi_s(x),y) \tilde{\nu}_s(dy)  ds   \right\} ds - t \right)  \mu_0(dx) \bigg| \\
& \leq \int_0^1 \bigg| v(\phi_t(x)) e^{-\lambda t} \bigg| \bigg| e^{2 \lambda \int_0^t \left[ \int_0^1 \rho(\phi_s(x),y) \nu_s(dy) \right] ds} - e^{2 \lambda \int_0^t \left[ \int_0^1 \rho(\phi_s(x),y) \tilde{\nu}_s(dy) \right] ds} \bigg| \mu_0(dx)  \\
& \leq ||v||_{\infty} \int_0^1  \bigg| e^{2 \lambda \int_0^t \left[ \int_0^1 \rho(\phi_s(x),y) \nu_s(dy) \right] ds} - e^{2 \lambda \int_0^t \left[ \int_0^1 \rho(\phi_s(x),y) \tilde{\nu}_s(dy) \right] ds} \bigg| \mu_0(dx)
\end{aligned}
\end{equation}
Using our assumptions that $0 \leq \rho(x,y) \leq 1$ and that $\nu_t$ and $\tilde{\nu}_t$ are non-negative measures satisfying $||\nu_t||_{TV}, ||\tilde{\nu}||_{TV} \leq R$ for $t \in [0,T]$, we may estimate that
\begin{equation}
\begin{aligned}
0 &\leq \int_0^t \int_0^1 \rho\left(\phi_s(x),y\right) \nu_s(dy) ds \leq R \\ 0 & \leq  \int_0^t \int_0^1 \rho\left(\phi_s(x),y\right) \tilde{\nu}_s(dy) ds \leq R
\end{aligned}
\end{equation}
for all $t \in [0,T]$. Combining this with the fact that the exponential function $e^a$ is locally Lipschitz for $a \in [0,2 \lambda R T]$ with Lipschitz constant $e^{2 \lambda R T}$, we may estimate that
\begin{equation}
\begin{aligned}
&\bigg| e^{2 \lambda \int_0^t \left[ \int_0^1 \rho(\phi_s(x),y) \nu_s(dy) \right] ds} - e^{2 \lambda \int_0^t \left[ \int_0^1 \rho(\phi_s(x),y) \tilde{\nu}_s(dy) \right] ds} \bigg| \\
&\leq e^{2 \lambda R T} \bigg| \int_0^t \int_0^1 \rho\left(\phi_s(x),y\right) \nu_s(dy) ds - \int_0^t \int_0^1 \rho\left(\phi_s(x),y\right) \tilde{\nu}_s(dy) ds \bigg| \\
&\leq e^{2 \lambda R T} \int_0^t \bigg| \int_0^1 \rho(\phi_s(x),y) \nu_s(dy) -   \int_0^1 \rho(\phi_s(x),y) \tilde{\nu}_s(dy) \bigg| ds \\
& \leq e^{2 \lambda R T}   \left( || \rho ||_{L^{\infty}\left([0,1]^2\right)} + \bigg| \bigg| \dsdel{\rho}{y} \bigg| \bigg|_{L^{\infty}\left([0,1]^2\right)} \right) \int_0^t || \nu_s - \tilde{\nu}_s ||_{BL} ds \\
& \leq e^{2 \lambda R T}   \left( || \rho ||_{L^{\infty}\left([0,1]^2\right)} + \bigg| \bigg| \dsdel{\rho}{y} \bigg| \bigg|_{L^{\infty}\left([0,1]^2\right)} \right) T || \nu - \tilde{\nu} ||_{C\left([0,T] ; \mc{M}\left([0,1] \right) \right)}
\end{aligned}
\end{equation}
We can then apply this estimate, use the fact that $\mu_0(dx)$ is a probability measure, and consider time $T$ satisfying
\begin{equation}T < \min\left\{ \frac{1}{2 \lambda R} , \eta \left(|| \rho ||_{L^{\infty}\left([0,1]^2\right)} + \bigg| \bigg| \dsdel{\rho}{y} \bigg| \bigg|_{L^{\infty}\left([0,1]^2\right)} \right)^{-1} \right\}
\end{equation}
for some $\eta \in (0,1)$, which allows us to deduce that
\begin{equation}
\begin{aligned}
& \bigg| \int_0^1 v(x) \mu_t^{\nu}(dx) - \int_0^1 v(x) \mu_t^{\tilde{\nu}}(dx) \bigg| \\
& \leq ||v||_{\infty} \int_0^1 e^{2 \lambda R T}   \left( || \rho ||_{L^{\infty}\left([0,1]^2\right)} + \bigg| \bigg| \dsdel{\rho}{y} \bigg| \bigg|_{L^{\infty}\left([0,1]^2\right)} \right) T || \nu - \tilde{\nu} ||_{C\left([0,T] ; \mc{M}\left([0,1] \right) \right)} \mu_0(dx) \\
& \leq ||v||_{W^{1,\infty}\left([0,1]\right)} \eta || \nu - \tilde{\nu} ||_{C\left([0,T] ; \mc{M}\left([0,1] \right) \right)}
\end{aligned}
\end{equation}
for all $t \in [0,T]$. Because this estimate holds for all test functions $v(x) \in W^{1,\infty}\left([0,1]\right)$ satisfying $||v||_{W^{1,\infty}\left([0,1]\right)} \leq 1$, we may further deduce that there is an $\eta \in (0,1)$ such that
\begin{equation}
|| H(\nu)_t - H(\tilde{\nu})_t ||_{BL} := || \mu_t^{\nu} - \mu_t^{\tilde{\nu}} ||_{BL} \leq \eta || \nu - \tilde{\nu} ||_{C\left([0,T] ; \mc{M}\left([0,1] \right) \right)}
\end{equation}
for all $t \in [0,T]$. We can then take the maximum of both sides over the interval $[0,T]$ for our choice of $T$ to conclude that 
\begin{equation}
|| H(\nu) - H(\tilde{\nu})||_{C\left([0,T];\mc{M}\left([0,1]\right)\right)} \leq \eta || \nu - \tilde{\nu} ||_{C\left([0,T] ; \mc{M}\left([0,1] \right) \right)}
\end{equation}
for some $\eta \in (0,1)$, so we have show that condition $(ii)$ holds for some $T$ sufficiently close to $0$.

We have shown the two conditions required to apply the Banach fixed point theorem to $H(\cdot)$ if we consider $T < \min(T_{\eta},T_{R})$. Therefore we know that there is a unique flow of measures $\nu^{\flat} := \nu_t^{\flat} \in C\left([0,T];\mc{M}([0,1])\right) \cap C^1\left([0,T];(C^1([0,1]))^*\right)$ that satisfy the fixed point relationship
\begin{equation}
H(\nu^{\flat}) = H(\nu^{\flat}_t)  := \{\mu_t^{\nu^{\flat}}\}_{t \in [0,T]} = \{\nu_t^{\flat}\}_{t \in [0,T]}. 
\end{equation}
Because the group-level victory probability $\rho(x,y)$ of an $x$-cooperator group is a $C\left([0,1]\right)$ function in $y$ for each $x \in [0,1]$, we see that $\rho(x,\cdot)$ is a valid test-function for our measure $\mu_t^{\nu^{\flat}}$, we can apply the fixed-point relationship to $\rho(x,\cdot)$ to write that
\begin{equation}
\int_0^1 \rho(x,y) \mu_t^{\nu^{\flat}}(dy) = \int_0^1 \rho(x,y) \nu_t^{\flat} (dy).
\end{equation}
Using this fixed-point property for the test-function $\rho(x,\cdot)$, we can then see from Equation \eqref{eq:PDEhlinear} that the measure $\mu_t^{\nu^{\flat}}$ satisfies
\begin{equation}
\begin{aligned}
\dsddt{} \int_0^1 v(x) \mu_t^{\nu^{\flat}}(dx) &= - \int_0^1 v'(x) x (1-x) \pi(x) \mu_t^{\nu^{\flat}}(dx) + \int_0^1 v(x) \left\{ 2\int_0^1 \rho(x,y) \nu_t^{\flat}(dy) - 1\right\} \mu_t^{\nu^{\flat}}(dx) \\ &=  - \int_0^1 v'(x) x (1-x) \pi(x) \mu_t^{\nu^{\flat}}(dx) + \int_0^1 v(x) \left\{ 2 \int_0^1 \rho(x,y) \mu_t^{\nu^{\flat}}(dy) - 1 \right\} \mu_t^{\nu^{\flat}}(dx)
\end{aligned}
\end{equation}
for each test function $v(x) \in C^1\left([0,1]\right)$, and we can therefore conclude that the flow of measures $\mu^{\nu^{\flat}} := \{ \mu_t^{\nu^{\flat}}\}_{t \in [0,T]}$ is a weak solution to the full multilevel dynamics of Equation \eqref{eq:PDEmeasure} for all $t \in [0,T]$.

Furthermore, because $\mu_t^{\nu_{\flat}}(dx)$ solves Equation \eqref{eq:PDEmeasure} for any test function $v(x)$, we see by using the test function $v(x) = 1$ and the fact that $\rho(x,y) = 1 - \rho(y,x)$
\begin{equation}
\begin{aligned}
\dsdel{}{t} \int_0^1 \mu_t^{\nu_{\flat}}(dx) &= \lambda \int_0^1  \left[ 2 \left( \int_0^1 \rho(x,y) \mu_t^{\nu^{\flat}}(dy) \right) - 1 \right] \mu_t^{\nu^{\flat}}(dx) \\
&= \lambda \int_0^1 \left[ \int_0^1 \left( \rho(x,y) - \rho(y,x) \right) \mu_t^{\nu^{\flat}}(dy) \right]  \mu_t^{\nu^{\flat}}(dx) \\
&= \lambda \int_{[0,1]^2} \left[ \rho(x,y) - \rho(y,x) \right] \left( \mu_t^{\nu^{\flat}} \times \mu_t^{\nu^{\flat}} \right)(dx,dy).
\end{aligned}
\end{equation}
Because $\rho(x,y) - \rho(y,x)$ is antisymmetric about the line $y=x$, we can conclude that the integral over the unit square on the righthand side vanishes, so we can use the fact that the initial condition $\mu_0^{h_{\flat}}(dx) = \mu_0(dx)$ is a a probability measure to deduce that
\begin{equation}
    \dsdel{}{t} \int_0^1 \mu_t^{\nu^{\flat}}(dx) = 0 \Longrightarrow \int_0^1 \mu_t^{\nu^{\flat}}(dx)  = 1 \: \: \mathrm{for} \: \: t \in [0,T].
\end{equation}

Furthermore, we can use the fixed point relationship for $\nu^{\flat}(t,x)$ and the representation formula from Equation \eqref{eq:hlinearrepresentation} for the solution to linear PDE of Equation \eqref{eq:PDEhlinear} to obtain the following implicit representation formula for the measure $\mu_t^{\nu^{\flat}}$ for each $t \in [0,T]$ and for each test function $v(x) \in C([0,1])$:
\begin{equation} \label{eq:representationformulaimplicit}
\begin{aligned}
\int_0^1 v(x) \mu_t^{\nu^{\flat}}(dx) &= \int_0^1 v(\phi_t(x)) \exp\left(2 \lambda \int_0^t \rho(\phi_s(x),y) \nu_{s}^{\flat}(dy)  ds - \lambda t \right) \mu_0(dx) \\
&= \int_0^1 v(\phi_t(x)) \exp\left(2 \lambda \int_0^t \int_0^1 \rho(\phi_s(x),y) \mu_s^{\nu^{\flat}}(dy)  ds - \lambda t \right) \mu_0(dx).
\end{aligned}
\end{equation}
In particular, we note from the form of our implicit representation formula from Equation \eqref{eq:representationformulaimplicit} and our assumption that the initial measure $\mu_0$ is a probability measure that the measures $\mu_t^{\nu^{\flat}}$ will be probability measures for all $t \in [0,T]$. 

We now show that the existence and uniqueness of solutions $\mu_t(dx)$ to Equation \eqref{eq:PDEmeasure} can be demonstrated globally in time. Because $\mu_t^{h_{\flat}}(dx)$ is a probability measure for $t \in [0,T]$, we can take the solution %
$\mu_T^{\nu^{\flat}}$ as a new initial probability measure for the linear PDE of Equation \eqref{eq:PDEhlinear}. Because the contraction mapping argument and the size of the value $T < \min(T_{\eta},T_{\epsilon})$ only relies on the initial distribution through the assumption that it is a probability measure, we can repeat our argument with initial measure $\mu_{T}^{\nu^{\flat}}$ to establish existence of a solution to the nonlinear dynamics of Equation \eqref{eq:PDEmeasure} for $t \in [T,2T]$. We can repeat this process as many times as needed, allowing us to extend the existence of a unique solution $\mu$ to Equation \eqref{eq:PDEmeasure} with given initial measure $\mu_0(dx)$ to any positive time $t$.
 \end{proof}

\section{Proof of Preservation of Infimum and Supremum H{\"o}lder Exponents for Measure-Valued Solutions}
\label{sec:infsupHolderpreserved}

\begin{proof}[Proof of Proposition \ref{prop:Holderpreserve}]
Because $\mu_t(dx)$ is a measure-valued solution to Equation \eqref{eq:PDEmeasure}, we may use the implicit representation formula for our solution to write that
\begin{equation} \label{eq:useimplicit}
\frac{\mu_t\left(\left[1-x,1\right] \right)}{x^{\Theta}} = \frac{\int_{1-x}^1 \mu_t(dx)}{x^{\Theta}} = x^{-\Theta} \int_{\phi_t^{-1}(1-x)}^1 \exp\left( 2 \int_0^t  \int_0^1 \rho(x,y) \mu_s(dy) ds - t \right) \mu_0(dx)
\end{equation}
We may now use this the expression on the righthand side of Equation \eqref{eq:useimplicit} to obtain upper and lower bounds on $\mu_t([x,1])$. For the upper bound, we use the assumption that $\rho(x,y)$ is a probability to note that $|\rho(x,y)| \leq ||\rho||_{L^{\infty}([0,1]^2)} \leq 1$, which allows us to deduce that
\begin{equation}
\begin{aligned}
\frac{\mu_t\left(\left[1-x,1\right] \right)}{x^{\Theta}}  &\leq  
x^{-\Theta} \int_{\phi_t^{-1}(1-x)}^1 \exp\left( \lambda \left[ 2 ||\rho||_{L^{\infty}([0,1]^2)} - 1  \right] t\right) \mu_0(dy) \\
& \leq 
e^{\lambda t} \left( \frac{\int_{\phi_t^{-1}(1 - x)}^1 \mu_0(dy)}{x^{\Theta}} \right) \\
&= e^{\lambda t}  \left(\frac{\mu_0\left( \left[\phi_t^{-1}(1 - x),1 \right]\right)}{x^{\Theta}} \right).
\end{aligned}
\end{equation}
We now use our assumption that $\pi(\cdot) > 0$ on a neighborhood of $1$ to note that $\phi_t^{-1}(1-x) \geq 1-x$ for $x$ sufficiently close to $0$, which allows us to estimate that
\begin{equation}
\frac{\mu_t\left(\left[1-x,1\right] \right)}{x^{\Theta}}   \leq e^{\lambda t}  \left(\frac{\mu_0\left( \left[\phi_t^{-1}(1-x) ,1 \right]\right)}{x^{\Theta}} \right)
\end{equation}
for sufficiently small $x$. This allows us to deduce that
\begin{subequations}
 \begin{align}
\label{eq:liminfupper} \liminf_{x \to 0} \frac{\mu_t\left(\left[1-x,1\right] \right)}{x^{\Theta}}   &\leq e^{\lambda t}  \liminf_{x \to 0} \left(\frac{\mu_0\left( \left[1-x ,1 \right]\right)}{x^{\Theta}} \right) \\
\label{eq:limsupupper} \limsup_{x \to 0} \frac{\mu_t\left(\left[1-x,1\right] \right)}{x^{\Theta}}   &\leq e^{\lambda t}  \limsup_{x \to 0} \left(\frac{\mu_0\left( \left[ 1-x ,1 \right]\right)}{x^{\Theta}} \right) 
 \end{align}   
\end{subequations}
\sloppy{Because we assume that $\mu_0(dx)$ has infimum H{\"o}lder exponent $\overline{\theta}$, we know that, for $\Theta < \overline{\theta}$, $ \liminf_{x \to 0} x^{-\Theta} \mu_0\left([1-x,1]\right) = 0$. We can then combine this with the bound of Equation \eqref{eq:liminfupper} and the fact that $\mu_t(dx)$ is a probability measure to deduce that}
\begin{equation}
 \liminf_{x \to 0} \frac{\mu_t\left(\left[1-x,1\right] \right)}{x^{\Theta}}  \leq 0, \: \: \textnormal{and therefore} \: \:  \liminf_{x \to 0} \frac{\mu_t\left(\left[1-x,1\right] \right)}{x^{\Theta}} = 0\: \: \mathrm{for} \: \: \Theta \leq \overline{\theta}. 
\end{equation}
\sloppy{This allows us to conclude that the infimum H{\"o}lder exponent for $\mu_t(dx)$ near $x=1$ satisfies $\overline{\theta}_t \geq \overline{\theta}$. Similarly, we know from the assumption that $\mu_0(dx)$ has supremum H{\"o}lder exponent $\underline{\theta}$ near $x=1$ that, for $\Theta < \underline{\theta}$, $\limsup_{x \ to 0} x^{-\Theta} x^{-\Theta} \mu_0\left([1-x,1]\right) = 0$ We can use this to show that}
\begin{equation}
\limsup_{x \to 0} \frac{\mu_t\left(\left[1-x,1\right] \right)}{x^{\Theta}}  = 0 \: \: \mathrm{for} \: \: \Theta < \underline{\theta},
\end{equation}
and we can conclude that $\underline{\theta}_t \leq \underline{\theta}$. 

To prove the complementary inequalities for the infimum and supremum H{\"o}lder exponents of $\mu_t(dx)$ near $x=1$, we know looking to find a lower bound on the righthand side of Equation \eqref{eq:useimplicit}. Noting that $\rho(x,y)$ is a probability for all $(x,y) \in [0,1]^2$, we have that $\rho(x,y) \geq 0$ for $(x,y) \in [0,1]^2$, and we can estimate that 
\begin{equation} \label{eq:extremallimfirstlowerbound}
\begin{aligned}
 \frac{\mu_t\left(\left[1-x,1\right] \right)}{x^{\Theta}} &\geq  \frac{e^{-\lambda t} \int_{\phi_t^{-1}(1-x)}^1 \mu_0(dx) }{ x^{-\Theta}}  \mu_0(dx) \\
 &= e^{-\lambda t} \frac{\mu_0\left(\left[ \phi_t^{-1}(1-x),1\right] \right)}{x^{\Theta}} \\
 &= e^{-\lambda t} \left( \frac{1 - \phi_t^{-1}(1-x)}{x} \right)^{\Theta} \left( \frac{\mu_0\left(\left[ 1 - \left( 1 - \phi_t^{-1}1-x) \right),1\right] \right)}{\left[ 1 - \phi_t^{-1}\left(1-x\right) \right]^{\Theta}} \right)
 \end{aligned}
 \end{equation}
We can use Equation \eqref{zest} from Lemma \ref{lem:backwardcharacteristics} to write  $\phi_t^{-1}(1-x)$ in the following form
\begin{equation}
\begin{aligned}
1-\phi_t^{-1}(1-x)&= \left[1-\left(1-x\right) \right] e^{- \pi(1) t} \exp\left(\int_{1-x}^{\phi_t^{-1}(1-x)} \frac{Q(s)ds}{s\pi(s)}\right) \\
&= x e^{-\pi(1) t} \exp\left(\int_{1-x}^{\phi_t^{-1}(1-x)} \frac{Q(s)ds}{s\pi(s)}\right).
\end{aligned}
\end{equation}
We can then use this expression to rewrite the estimate of Equation \eqref{eq:extremallimfirstlowerbound} as
\begin{equation} \label{eq:extremallimsubsequentlowerbound}
 \frac{\mu_t\left(\left[1-x,1\right] \right)}{x^{\Theta}} \geq \underbrace{e^{- \left[\lambda + \Theta \pi(1)\right] t}  \exp\left(\Theta \int_{1-x}^{\phi_t^{-1}(1-x)} \frac{Q(s)ds}{s\pi(s)}\right)}_{:= C(t,x)}  \left( \frac{\mu_0\left(\left[ 1 - \left( 1 - \phi_t^{-1}(1-x) \right),1\right] \right)}{\left[ 1 - \phi_t^{-1}\left(1-x\right) \right]^{\Theta}} \right)
\end{equation}
We can then derive an upper bound for the function $C(t,x)$ by restricting attention to $x$ sufficiently close to $0$. For any $\delta < 1$, we can use the fact from Lemma \ref{lem:backwardcharacteristics} that $Q(s)$ is bounded for $s \in [1-\delta,1]$ to estimate that, for all $x \in [0,\delta]$ and all $t \geq 0$, the function $C(t,x)$ satisfies
\begin{equation}
C(t,x) \geq e^{- \left[\lambda + \Theta \pi(1)\right]  t} \underbrace{\exp\left( \Theta \int_{1-\delta}^{1} \frac{\left[ Q(s) \right]_{-}}{s \pi(s)} ds\right)}_{:= C_{\delta}} > 0,
\end{equation}
where $\left[Q(s)\right]_{-}$ denotes the negative part of the function $Q(s)$. Because this bound holds for $x \in [0,\delta]$, we may further apply this estimate to Equation \eqref{eq:extremallimsubsequentlowerbound} that, for all $x$ close enough to $0$, 
\begin{equation}
\frac{\mu_t\left(\left[1-x,1\right] \right)}{x^{\Theta}} \geq  C_{\delta}  e^{- \left[\lambda + \Theta \pi(1)\right]  t} \left( \frac{\mu_0\left(\left[ 1 - \left( 1 - \phi_t^{-1}(1-x) \right),1\right] \right)}{\left[ 1 - \phi_t^{-1}\left(1-x\right) \right]^{\Theta}} \right)
\end{equation}
We can then use the substitution $z = 1 - \phi_t^{-1}(1-x)$, the fact that $\phi_t^{-1}(1-x)$ is continuous in $x$, and the fact that $\phi_t^{-1}(1-x) \to 1$ as $x \to 0$ to deduce that
\begin{subequations} \label{eq:extremallimlowerbounds}
\begin{align}
\label{eq:liminflowerbound} \liminf_{x \to 0} \frac{\mu_t\left(\left[1-x,1\right] \right)}{x^{\Theta}}  & \geq C_{\delta}  e^{- \left[\lambda + \Theta \pi(1)\right]  t} \geq \liminf_{x \to 0} \left( \frac{\mu_0\left(\left[ 1 - \left( 1 - \phi_t^{-1}(1-x) \right),1\right] \right)}{\left[ 1 - \phi_t^{-1}\left(1-x\right) \right]^{\Theta}} \right) \nonumber \\ &\geq  C_{\delta}  e^{- \left[\lambda + \Theta \pi(1)\right]  t}  \liminf_{z \to 0} \frac{\mu_t\left(\left[1-z,1\right] \right)}{z^{\Theta}} \\
\limsup_{x \to 0} \frac{\mu_t\left(\left[1-x,1\right] \right)}{x^{\Theta}}  & \geq  C_{\delta}  e^{- \left[\lambda + \Theta \pi(1)\right]  t} \limsup_{x \to 0} \left( \frac{\mu_0\left(\left[ 1 - \left( 1 - \phi_t^{-1}(1-x) \right),1\right] \right)}{\left[ 1 - \phi_t^{-1}\left(1-x\right) \right]^{\Theta}} \right) \nonumber \\ &\geq  C_{\delta}  e^{- \left[\lambda + \Theta \pi(1)\right]  t} \limsup_{z \to 0} \frac{\mu_t\left(\left[1-z,1\right] \right)}{z^{\Theta}}.
\end{align}
\end{subequations}
Because we assume that $\mu_0(dx)$ has infimum and supremum H{\"o}lder exponents $\overline{\theta}$ and $\underline{\theta}$ for $x=1$, we know that $\liminf_{z \to 0} z^{-\Theta} \mu_0\left([1-z,1]\right) > 0$ for $\Theta > \overline{\theta}$ and $\limsup_{z \to 0} x^{-\Theta} \mu_0\left([1-z,1]\right) > 0$ for $\Theta > \underline{\theta}$. We can then use this to show that 
\begin{subequations}
\begin{align}
\liminf_{x \to 0} \frac{\mu_t\left(\left[1-x,1\right] \right)}{x^{\Theta}} & \geq 0 \: \: \mathrm{for} \: \: \Theta > \overline{\theta} \\
\limsup_{x \to 0} \frac{\mu_t\left(\left[1-x,1\right] \right)}{x^{\Theta}} & \geq 0 \: \: \mathrm{for} \: \: \Theta > \underline{\theta},
\end{align}
\end{subequations}
which allows us to deduce that $\overline{\theta}_t \leq \overline{\theta}$ and $\underline{\theta}_t \leq \underline{\theta}$ for any $t \geq 0$. Combining this the complementary inequalities shown above, we can now conclude that $\overline{\theta}_t = \overline{\theta}$ and $\underline{\theta}_t = \underline{\theta}$ for all $t \geq 0$.
 \end{proof}

\end{document}